\documentclass[a4paper,twocolumn,11pt,accepted=2021-06-18]{quantumarticle}
\pdfoutput=1

\usepackage{placeins}
\usepackage{amsfonts, amsmath,amssymb,amsthm,bm}
\usepackage{mathrsfs}
\usepackage{graphicx}
\usepackage[usenames,dvipsnames]{color}
\usepackage{palatino}
\usepackage{fancyhdr}
\usepackage{mathrsfs}
\usepackage{multirow}
\usepackage{hyperref}
\usepackage{bigints}
\usepackage{soul}
\usepackage{cite}

\usepackage{framed}
\definecolor{shadecolor}{rgb}{.9,.9,.95}
\definecolor{gr}{rgb}{0,.5,0}

\usepackage{braket}

\hypersetup{
   colorlinks=true,
   linkcolor=blue,
   citecolor=blue,
   urlcolor=Violet
}

\newcommand{\ca}{\mathcal A}

\newcommand{\cg}{\mathcal G}
\newcommand{\ch}{\mathcal H}

\newcommand{\cu}{\mathcal U}

\newcommand{\tr}{\mathrm{tr}}
\newcommand{\R}{\mathbb{R}}
\newcommand{\C}{\mathbb{C}}
\newcommand{\N}{\mathbb{N}}

\def\nn{\nonumber}
\def\q{{\quad}}

\def\be{\begin{equation}}
\def\ee{\end{equation}}
\def\ba{\begin{eqnarray}}
\def\ea{\end{eqnarray}}

\theoremstyle{plain}
\newtheorem{theorem}{Theorem}

\newtheorem{lemma}[theorem]{Lemma}

\newtheorem{example}[theorem]{Example}

\newtheorem{definition}[theorem]{Definition}

\theoremstyle{definition}

\begin{document}
\title{Quantum reference frame transformations as symmetries and the paradox of the third particle}
\author{Marius Krumm}
\affiliation{Institute for Quantum Optics and Quantum Information, Austrian Academy of Sciences, Boltzmanngasse 3, A-1090 Vienna, Austria}
\affiliation{Vienna Center for Quantum Science and Technology (VCQ), Faculty of Physics, University of Vienna, Vienna, Austria}
\email{marius.krumm@univie.ac.at}
\author{Philipp A.\ H\"ohn}
\affiliation{Okinawa Institute of Science and Technology Graduate University, Onna, Okinawa 904 0495, Japan}
\affiliation{Department of Physics and Astronomy, University College London, London, United Kingdom}
\email{philipp.hoehn@oist.jp}
\author{Markus P.\ M\"uller}
\affiliation{Institute for Quantum Optics and Quantum Information, Austrian Academy of Sciences, Boltzmanngasse 3, A-1090 Vienna, Austria}
\affiliation{Vienna Center for Quantum Science and Technology (VCQ), Faculty of Physics, University of Vienna, Vienna, Austria}
\affiliation{Perimeter Institute for Theoretical Physics, 31 Caroline Street North, Waterloo ON N2L 2Y5, Canada}
\email{markus.mueller@oeaw.ac.at}

\begin{abstract}
In a quantum world, reference frames are ultimately quantum systems too --- but what does it mean to ``jump into the perspective of a quantum particle''? In this work, we show that quantum reference frame (QRF) transformations appear naturally as symmetries of simple physical systems. This allows us to rederive and generalize known QRF transformations within an alternative, operationally transparent framework, and to shed new light on their structure and interpretation. We give an explicit description of the observables that are measurable by agents constrained by such quantum symmetries, and apply our results to a puzzle known as the `paradox of the third particle'. We argue that it can be reduced to the question of how to relationally embed fewer into more particles, and give a thorough physical and algebraic analysis of this question. This leads us to a generalization of the partial trace (`relational trace') which arguably resolves the paradox, and it uncovers important structures of constraint quantization within a simple quantum information setting, such as relational observables which are key in this resolution. While we restrict our attention to finite Abelian groups for transparency and mathematical rigor, the intuitive physical appeal of our results makes us expect that they remain valid in more general situations.
\end{abstract}

\maketitle

\onecolumngrid

\vskip 1cm
\begin{figure}[hbt]
\begin{center}
\includegraphics[width=.35\textwidth]{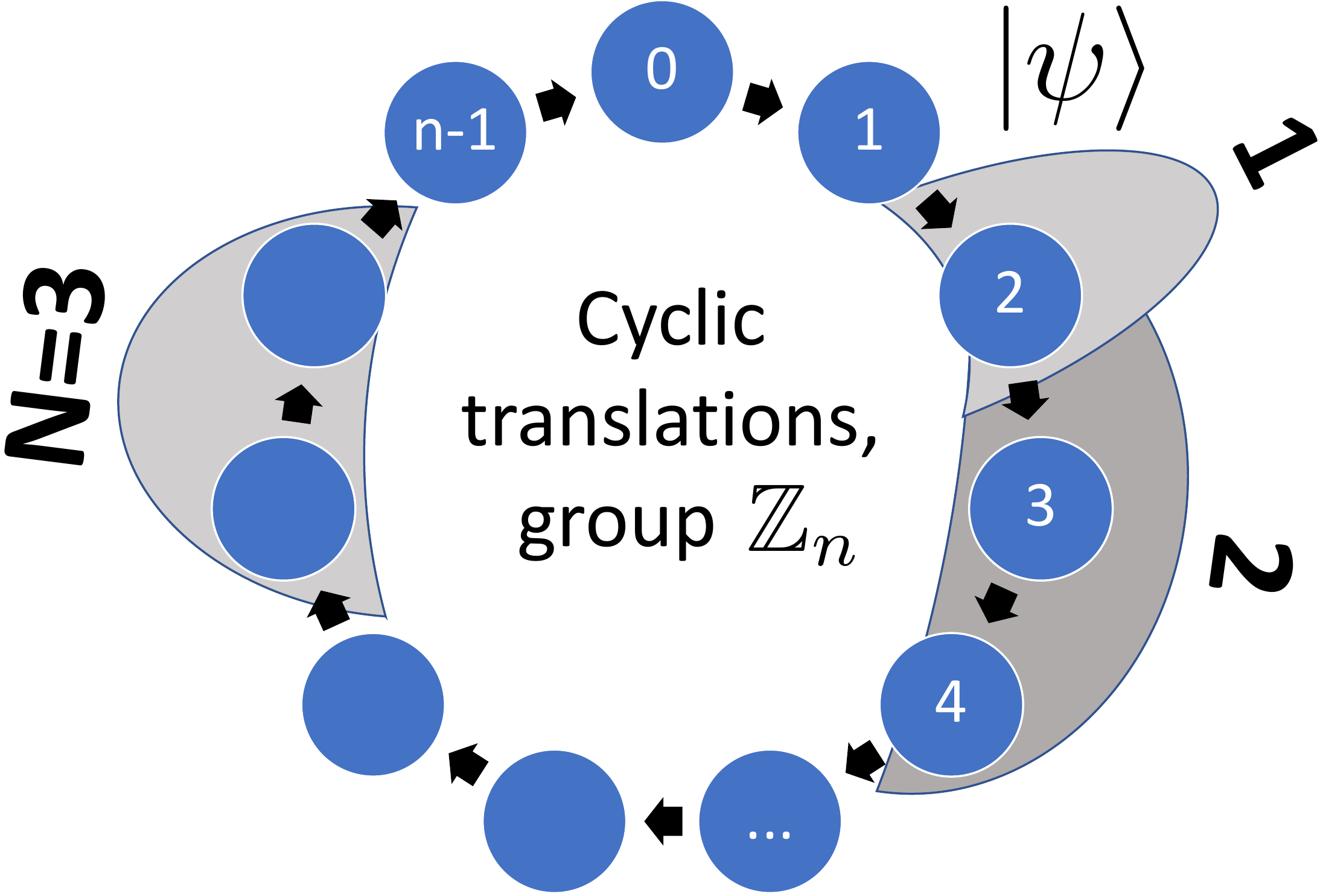}
\caption{The simplest example of this article's setup: a discretization of wave functions in one spatial dimension under translation symmetry. The configuration space is the cyclic group $\mathbb{Z}_n$, and the one-particle Hilbert space is $\mathcal{H}=\ell^2(\mathbb{Z}_n)\simeq\C^n$. We have $N$ distinguishable particles in a joint quantum state $|\psi\rangle\in\mathcal{H}^{\otimes N}$, and we study QRF transformations that switch between the ``perspectives of the particles''.}
\label{fig_cyclic}
\end{center}
\end{figure}

\newgeometry{margin=0.55in, voffset=0.2in, bottom=1in}
\fancyheadoffset{0pt}
\fancyfootoffset{0pt}

\hypersetup{linkcolor=black}
\setcounter{tocdepth}{2}
\begin{shaded}
\tableofcontents
\end{shaded}

\bigskip

\twocolumngrid

\section{Introduction}
All physical quantities are described relative to some frame of reference. But since all physical systems are fundamentally quantum, reference frames must ultimately be quantum systems, too. This simple insight is of fundamental importance in a variety of physical fields, including the foundations of quantum physics~\cite{Aharonov1,Aharonov2,Aharonov3,Wigner,Araki,Yanase,Loveridge2017,Loveridge2018,Miyadera,Loveridge2020,HoehnMueller}, quantum information theory~\cite{Bartlett,Marvian,Frameness,Modes,Smith2019,ResourceTheoryQRF,Palmer,Smith2016}, quantum thermodynamics~\cite{LJR,LKJR,Aberg,LostaglioMueller,MarvianSpekkens,Erker,Cwiklinski,Woods1,Woods2}, and quantum gravity~\cite{Rovellibook,Rovelli1,Rovelli2,Rovelli3,Dittrich1,Dittrich2,Chataignier,Thiemann,Tambornino}.

Recently, there has been a wave of interest in a specific approach to quantum reference frames (QRFs) that we can broadly classify as \emph{structural} in nature, including e.g.\ Refs.~\cite{Giacomini,Vanrietvelde,Hamette,Vanrietvelde2,Hoehn:2018aqt,Hoehn:2018whn,Hoehn:2019owq,Hoehn:2020epv,Castro,Giacomini-spin1,Giacomini-spin2,Yang}. This approach extends the usual concept of reference frames by associating them with quantum systems, and by describing the physical situation of interest from the ``internal perspective'' of that quantum system. For example, if an interferometer has a particle travelling in a superposition of paths, how ``does the particle see the interferometer''~\cite{Angelo}?

A central topic in this approach is the QRF dependence of observable properties like superposition, entanglement~\cite{Giacomini,Vanrietvelde,Hamette}, classicality \cite{Vanrietvelde,Darwinism1,Darwinism2}, or of quantum resources~\cite{Savi}. The corresponding QRF transformations admit an unambiguous definition of spin in relativistic settings by transforming to a particle's rest frame~\cite{Giacomini-spin1,Giacomini-spin2}, they describe the comparison of  quantum clock readings \cite{Hoehn:2018aqt,Hoehn:2020epv}, and they yield an alternative approach to indefinite causal structure~\cite{Castro,Guerin}. Among other conceived applications, they are furthermore conjectured to play a crucial rule in the implementation of a ``quantum equivalence principle''~\cite{Hardy} as well as in spacetime singularity resolution~\cite{Gielen:2020abd} and the description of early universe power spectra \cite{Giesel1,Giesel2} in quantum gravity and cosmology.

Despite the broad appeal, several fundamental and conceptual questions remain open. For example, how should we make concrete sense of the idea of ``jumping into the reference frame of a particle''? How are QRF transformations different from any other unitary change of basis in Hilbert space? What kind of physical symmetry claim is associated with the intuition that QRF changes ``leave the physics invariant''? Furthermore, there are reported difficulties to extend basic quantum information concepts into this context. For example, Ref.~\cite{Angelo} describes a `paradox of the third particle', an apparent inconsistency arising from determining reduced quantum states in different QRFs.

In this article, we shed considerable light on all of these questions. We introduce a class of physical systems subject to simple principles, and derive the QRF transformations as the physical symmetries of these systems. On the one hand, this gives us a transparent operational framework for QRFs that makes sense of the `jumping' metaphor. On the other hand, it allows us to identify QRF transformations as elements of a natural symmetry group, and to describe the structure of the observables that are invariant under such transformations. This algebraic structure turns out to be key to elucidate the paradox of the third particle, which we do by introducing a relational notion of the partial trace.

To keep the mathematical structures as transparent and accessible as possible, we restrict our attention in this article to finite Abelian groups. But this already includes interesting physical settings like the discretization of translation-invariant quantum particles on the real line (see Figure~\ref{fig_cyclic}), admitting the formulation of intriguing thought experiments. Within this familiar quantum information regime of finite-dimensional Hilbert spaces, we uncover a variety of structures that not only shed light on the questions raised above, but that also reflect important aspects of constraint quantization~\cite{Dirac,HenneauxTeitelboim}, which for example underlies canonical approaches to quantum gravity and cosmology. This includes the notions of a ``physical Hilbert space'' encoding the relational states of the theory~\cite{Giulini:1998kf,Marolf:2000iq,Thiemann}, of relational and Dirac observables \cite{Rovellibook,Thiemann,Tambornino,Rovelli1,Rovelli2,Rovelli3,Dittrich1,Dittrich2,Chataignier,Hoehn:2018aqt,Hoehn:2018whn,Hoehn:2019owq,Hoehn:2020epv}, and a simple demonstration of how constraints can in general arise from symmetries. In particular, these notions will assume key roles in our proposed resolution of the paradox of the third particle.\\

\textbf{Overview and summary of results.} Our article is organized as follows. In Section~\ref{SecOperationalStructural}, we begin with a thorough operational comparison of this structural approach to QRFs with the more common quantum information approach. This sets the stage by embedding the notion of QRF transformations into a broader conceptual framework.

\begin{figure}[hbt]
\begin{center}
\includegraphics[width=.4\textwidth]{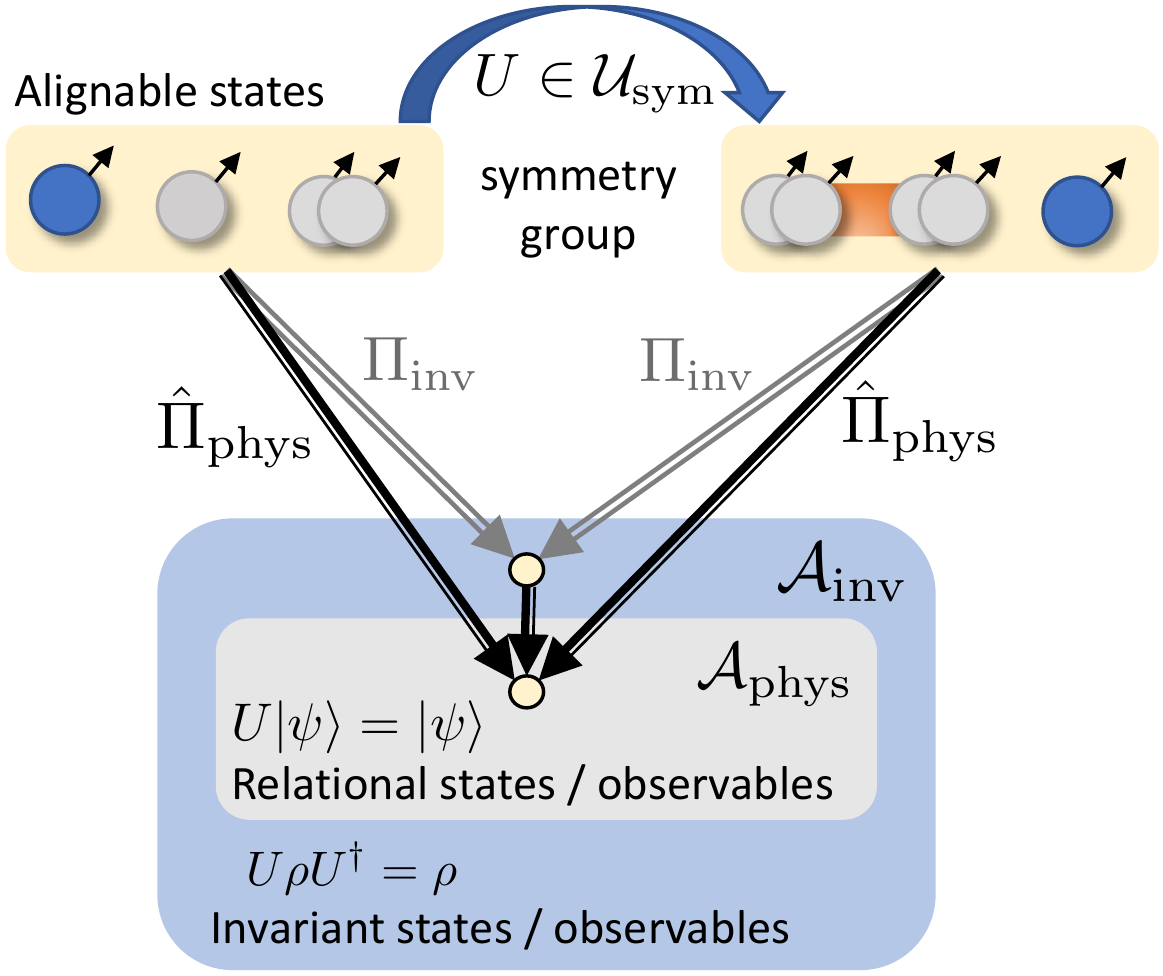}
\caption{Some of the structures we uncover in Section~\ref{SecTechnical}. We axiomatically derive and analyze the quantum symmetry group $\mathcal{U}_{\rm sym}$, and characterize a class of ``alignable states'' that can be transformed into a form that is ``relative to one of the particles''. As described in~Refs.~\cite{Giacomini,Vanrietvelde,Hamette}, ``jumping from the first to the third particle'', for example (sketched on top), can transform separable into entangled states, owing to the fact that, as we will show, $\mathcal{U}_{\rm sym}$ is larger than the classical group of translations. We identify two subalgebras of operators that are invariant under all quantum symmetries, $\mathcal{A}_{\rm phys}\subset\mathcal{A}_{\rm inv}$, and corresponding projections that extract the ``invariant part'' of a state.}
\label{fig_overview}
\end{center}
\end{figure}

In Section~\ref{SecTechnical}, we specialize to a concrete class of physical systems (``$\mathcal{G}$-systems'') which hold a finite Abelian group as their classical configuration space. We prove the existence and elucidate the group structure of QRF transformations for such systems, and introduce a notion of ``alignable states'' which are those that can be described ``relative to one of the particles''. We determine the invariant observables measurable by observers constrained by such symmetries. As sketched in Figure~\ref{fig_overview}, we find that there are two important, but distinct notions of invariant observables, depending on whether symmetry transformations may induce superselection sector dependent phases or not. While the role of invariant observables in the structural approach has been stressed before \cite{Vanrietvelde,Vanrietvelde2,Hoehn:2018aqt,Hoehn:2018whn,Hoehn:2019owq,Hoehn:2020epv}, attention was thus far restricted to their description on the space of invariant pure states (``physical Hilbert space''). Furthermore, we uncover important aspects of constraint quantization, and obtain representation-theoretic notions of physical concepts like the ``total momentum'' and its role as a constraint.

In Section~\ref{Section:Paradox}, we apply our insights to the paradox of the third particle. We argue that the problem reduces to the physical question of when two groups of particles hold ``the same relation'' to each other within two distinct configurations, such that the corresponding branches should interfere (see Figure~\ref{fig_superposition} on page~\pageref{fig_superposition}). Mathematically, this corresponds to the question of how to embed the algebra of invariant $N$-particle observables into that of $N+M$ particles. We show that no unique answer to this question exists for the full set of invariant observables in $\mathcal{A}_{\rm inv}$: the answer always depends on the physical choice of how to determine the particle group interrelations operationally.

However, we show that a unique and natural embedding \emph{does} exist for the subset of relational observables in $\mathcal{A}_{\rm phys}$. The trick is to use a \emph{coherent superposition} of all operationally conceivable particle group relations, and it turns out this construction preserves the algebraic structure of the $N$-particle observables. We use this to define a relational notion of the partial trace which arguably resolves the paradox, and we compare this resolution to the one proposed by Angelo et al.~\cite{Angelo} before concluding in Section~\ref{SecConclusions}.

\section{ Quantum information vs.\ structural approach to reference frames}
\label{SecOperationalStructural}
Let us begin with the main element that both the quantum information as well as the structural approach to QRFs have arguably in common: a physical system with a symmetry such that all observable quantities are invariant, or even fully relational. This is also the starting point of Refs.~\cite{Loveridge2017,Loveridge2018,Miyadera,Loveridge2020,Smith2016,Smith2019}.

\subsection{Describing physics with or without external relatum}

Consider a physical system $S$ of interest. We assume that there is a set $\mathcal{S}$ of \emph{states} in which the system $S$ can be prepared. Furthermore, there is a group of \emph{symmetry transformations} $\mathcal{G}_{\rm sym}$ that acts on $\mathcal{S}$. Specifying $\mathcal{S}$ and $\mathcal{G}_{\rm sym}$ amounts to making a specific physical claim:
\begin{shaded}
\textbf{Assumption 1.} If the system $S$ is considered \emph{in isolation}, then it is impossible  to distinguish (even probabilistically) whether it has been prepared in some state $\omega$ or in another state $G(\omega)$. This is true for all states $\omega\in\mathcal{S}$ and all symmetry transformations $G\in\mathcal{G}_{\rm sym}$.
\end{shaded}
'In isolation' here means that any other physical structure to which $S$ could be related is disregarded, either because it does not exist in the first place, one does not have access to it, or it is deliberately ignored. This setting is schematically depicted in Figure~\ref{fig_symmetry}. Examples include:
\begin{itemize}
	\item[(i)] Minkowski spacetime of special relativity, with $\mathcal{S}$ the set of all possible states of matter (say, of classical point particles), and the Poincar\'e group $\mathcal{G}_{\rm sym}$ as the group of symmetry transformations.
	\item[(ii)] Electromagnetism in some bounded region of spacetime. This is a gauge theory with $\mathcal{G}_{\rm sym}$ the group of local ${\rm U}(1)$-transformations as its symmetry group.
	\item[(iii)] A spin  in quantum mechanics with Hilbert space $\mathcal{H}$ and projective representation $g\mapsto U_g$ of the rotation group $\mathcal{G}=\rm{SO}(3)$. Here, $\mathcal{S}$ is the set of density matrices $\rho$, and $\mathcal{G}_{\rm sym}$ consists of all maps of the form $\rho\mapsto U_g\rho U_g^\dagger$.
\end{itemize}
These three examples illustrate an important subtlety: to claim that $\rho$ and $G\rho$ are physically indistinguishable, one needs to speak about $\rho$ and $G\rho$ as different objects in the first place. In other words, one has to \emph{somehow} define $\rho$ and $G\rho$ as distinct states. But in order to do so, one would need something \emph{external} to the system $S$ to refer to.

\begin{figure}[hbt]
\begin{center}
\includegraphics[width=.3\textwidth]{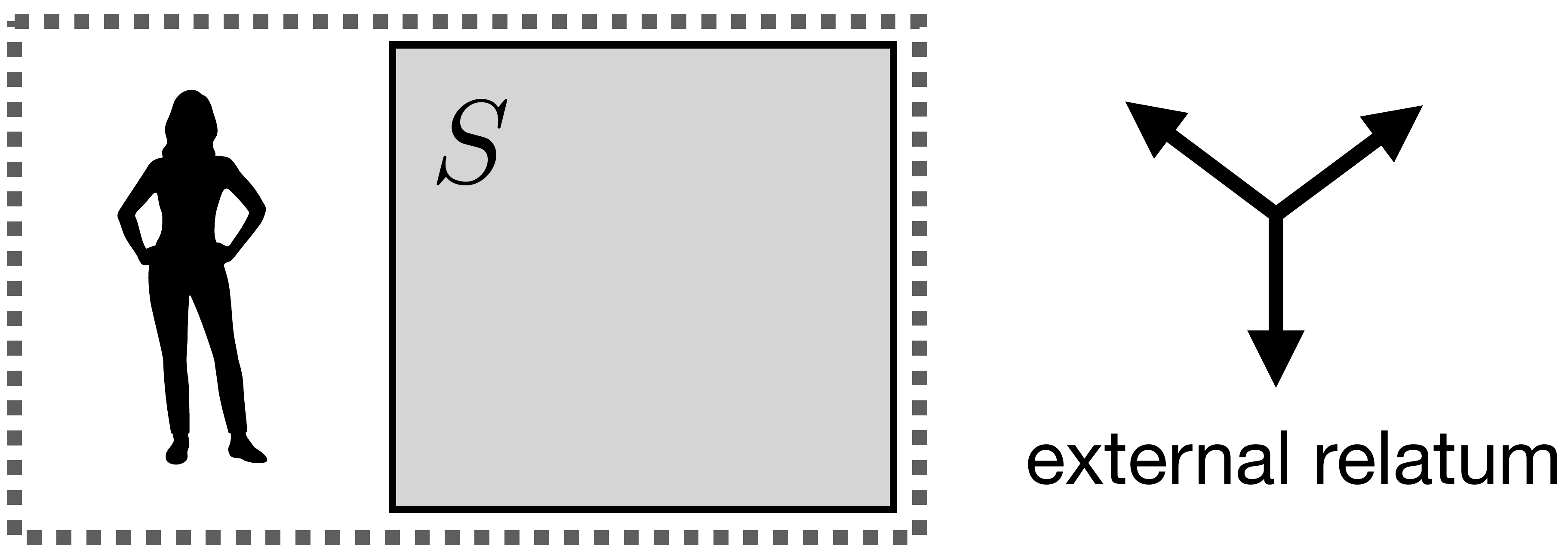}
\caption{What both approaches have in common: a system $S$ with a symmetry group $\mathcal{G}_{\rm sym}$ acting on its states $\rho\in\mathcal{S}$. States are implicitly defined via some (physical or fictional) external relatum, but \emph{internally} (that is, for observers without access to the relatum) $\rho$ and $G\rho$ are indistinguishable, for all $G\in\mathcal{G}_{\rm sym}$. }
\label{fig_symmetry}
\end{center}
\end{figure}

In example (i), there simply is no material external relatum, while in example (ii), it is given by electromagnetism \emph{outside} of the bounded region. As emphasized in Ref.~\cite{Rovelli}, while gauge symmetries do not change the physics of a given system, they alter the way that the system interacts with other systems. This observation is at the heart of the recent pivot to edge modes in gauge theory and gravity \cite{Donnelly:2014fua,Donnelly:2016auv,Geiller:2019bti,Freidel:2020xyx,Gomes:2018shn,Riello:2020zbk,Wieland:2017zkf,Wieland:2017cmf} and our resolution of the paradox of the third particle in Section~\ref{Section:Paradox} can also be viewed in this light. In case (iii), the external relatum would be best described as an external classical reference frame, for example the laboratory of an agent experimenting with $S$. This illustrates that to consider a system ``in isolation'' in the sense of Assumption 1 does \emph{not} imply that the system $S$ is literally a physically isolated system. It simply means that we have chosen to describe the system without the external relatum relative to which the action of the symmetry group is defined. Moreover, the setting does \emph{not} imply that the agent who treats $\rho$ and $G\rho$ as indistinguishable is herself part of the system $S$, but only that the agent considers $S$ without the external relatum.

Here we argue that the essential difference between the two approaches to quantum references frames can succinctly be stated as follows:
\begin{shaded}
The \textbf{quantum information (QI) approach} as in e.g.\ Refs.~\cite{Bartlett,Marvian,Frameness,Modes,ResourceTheoryQRF,Smith2019} emphasizes the fact that quantum states are often only defined relative to an external relatum (as in Figure~\ref{fig_symmetry}), and that this relatum may ultimately be a quantum system, too. This leads to questions like: how can quantum information-theoretic protocols be performed in the absence of a shared reference frame~\cite{Bartlett}? How well can quantum states stand in as resources of asymmetry if there is no shared frame~\cite{Marvian,ResourceTheoryQRF}? What are fundamental quantum limits for communicating or aligning reference frames~\cite{Bartlett}? Addressing questions as these often involves encoding information in quantum states in an external relatum independent manner and, as such, requires external relatum independent descriptions of states.

The \textbf{structural approach} as in e.g.\ Refs.~\cite{Giacomini,Vanrietvelde,Hamette} is not primarily concerned with operational protocols. While it shares the aim of external relatum independent descriptions of states with the QI approach, it goes further: it disregards the external relatum altogether, and instead asks  whether and how \emph{physical subsystems} of $S$ can be promoted to an \emph{internal} reference frame. This emphasizes the fact that the distinction between quantum systems and their reference frames is not fundamental, but merely conventional. It leads to questions like: what is the description of the quantum state relative to one of its particles? Can we find a Hilbert space basis in which the description of the physics is simplified, e.g., in which superpositions of subsystems of interest may be removed? More generally, what are the ``QRF transformations'' that relate the descriptions relative to different choices of internal reference frame?
\end{shaded}
In the QI approach, it is usually not necessary to take the extra step to internal frame choices and to ask how a system is described relative to one of its subsystems, as we will explain shortly. It suffices to focus on invariant properties of $S$ which have a meaning relative to an arbitrary choice of external frame in order to successfully carry out communication tasks in the absence of a shared frame. It is also worth emphasizing that there does not exist a sharp distinction between the two approaches in the body of literature on QRFs. Since the structural approach shares external relatum independent state descriptions with the QI approach, there exist ``hybrid'' works which arguably incorporate elements from both. For example, Refs.\ \cite{Loveridge2017,Loveridge2018,Miyadera,Loveridge2020,Smith2016,Smith2019,Palmer,Angelo} use standard quantum information techniques to define external relatum independent states, but also use the latter to explore to some degree the question of how a quantum state is described relative to a subsystem. However, these works do not study the relations between the different such descriptions and thus, in particular, do not study the QRF transformations.

The structural approach to QRFs is sometimes illustrated in ways that seem at first sight to be in conflict to the characterization above. For example, Figure 1 in Ref.~\cite{Giacomini} suggests to think of QRFs as physically attached to an observer and its laboratory (defined by its own quantum state), similarly as reference frames in Special Relativity are often thought of as being attached to an observer (defined by its state of motion). QRF transformations would then relate the descriptions of ``quantum'' observers who are relative to each other in superposition in a Wigner's-friend-type fashion.

However, we will show below that the structural framework of QRFs can be derived and analyzed exactly under an alternative and simpler interpretation. As we will elaborate and generalize below, choosing a QRF amounts to \emph{aligning one's description of the physics with respect to some choice of internal quantum subsystem}, such as the position of one of the particles. Two different observers can choose two different subsystems (say, particles) that are relative to each other in superposition, even if the observers themselves are fully classical. Their descriptions will then be related by QRF transformations. The observer who assigns the quantum state may thus retain the status of a \emph{classical} entity external to the quantum system (at least in laboratory situations), as illustrated in Figure~\ref{fig_symmetry}. While more conservative, this new interpretation is operationally more immediate, and it is sufficient to reconstruct and extend the full machinery of QRF transformations, as we will see.

The characterization above is also in line with another version of the structural approach: the so-called perspective-neutral approach~\cite{Vanrietvelde,Vanrietvelde2,Hoehn:2018aqt,Hoehn:2018whn,Hoehn:2019owq,Hoehn:2020epv} which, motivated by quantum gravity, is formulated in the language of constrained Hamiltonian systems ~\cite{Dirac,HenneauxTeitelboim}. The starting point of this approach is a deep physical and operational motivation: \emph{take the idea seriously that there are no reference frames, such as rods or clocks, that are external to the universe}. To implement this idea, one starts with a ``kinematical Hilbert space'' that defines all the involved quantum degrees of freedom and some gauge symmetry, but is interpreted as purely auxiliary. The absence of external references is then implemented by restricting to the gauge-invariant subset of states where the description becomes purely relational.

The actual mathematical machinery applied in this approach still fits the description above: the kinematical Hilbert space can be viewed as being described relative to a \emph{fictional} external relatum. The insight that there is nothing external to the universe motivates to ask --- purely formally at first --- whether some of the \emph{internal} degrees of freedom of the theory can be promoted to a frame of reference, such as a rod or clock. One may finally ask whether observers who are part of the theory may in fact have good operational access to that chosen frame of reference, but this is an \emph{additional} (though important) question that we here regard as secondary.

\subsection{Communication scenarios illustrating the two approaches}

Before we turn to the structural approach in detail, and relate the verbal description above to the mathematical formalization, let us elaborate on the distinction by means of two communication scenarios. To do so, let us informally introduce some piece of notation that we will later on define more formally. If $\rho\in\mathcal{S}$ is some state, denote by $[\rho]$ the set of all states that are symmetrically equivalent to $\rho$, i.e.\ $[\rho]:=\{G\rho\,\,|\,\, G \in \mathcal{G}_{\rm sym}\}$. The $[\rho]$ can be viewed as equivalence classes of states, or as orbits of the symmetry group.

Adapting the quantum information terminology from Ref.~\cite{Bartlett}, we refer to physical properties of $S$ that only depend on the equivalence class $[\rho]$ as \emph{speakable information}. Being invariant under the action of $\mathcal{G}_{\rm sym}$ and thus not requiring an external relatum in order to be defined, two agents can agree on the description of these properties by classical communication even in the absence of a shared frame. By contrast, we refer to physical properties of $S$ that depend on the concrete representative $\rho$ from an equivalence class $[\rho]$ of states as \emph{unspeakable information}. These properties thus require the external relatum to be meaningful and cannot be communicated purely classically between two agents who do not share a frame.

\subsubsection{The quantum information approach: communicating quantum systems}\label{sssec_qiapp}

Consider the scenario in
Figure~\ref{fig_qit}. Alice holds a quantum system $S$ that she has
prepared in some state $\rho\in\mathcal{S}(\mathcal{H})$, and
$\mathcal{S}(\mathcal{H})$ denotes the density matrices on the
corresponding Hilbert space $\mathcal{H}$. We assume that there is a
(for now, for simplicity) compact group $\mathcal{G}$ of symmetries and
a projective representation $\mathcal{G}\ni g\mapsto U_g$ such that
$\mathcal{G}$ acts on $\mathcal{S}$ via $\mathcal{U}_g(\rho)=U_g \rho
U_g^\dagger$. In this case, the symmetry group is $\mathcal{G}_{\rm
sym}=\{\mathcal{U}_g\,\,|\,\,g\in\mathcal{G}\}$. If we assume that
Alice's quantum system $S$ has the properties of Assumption 1, then the
very definition of $\rho$ is relative to her local frame of reference.

\begin{figure}[hbt]
\begin{center}
\includegraphics[width=.47\textwidth]{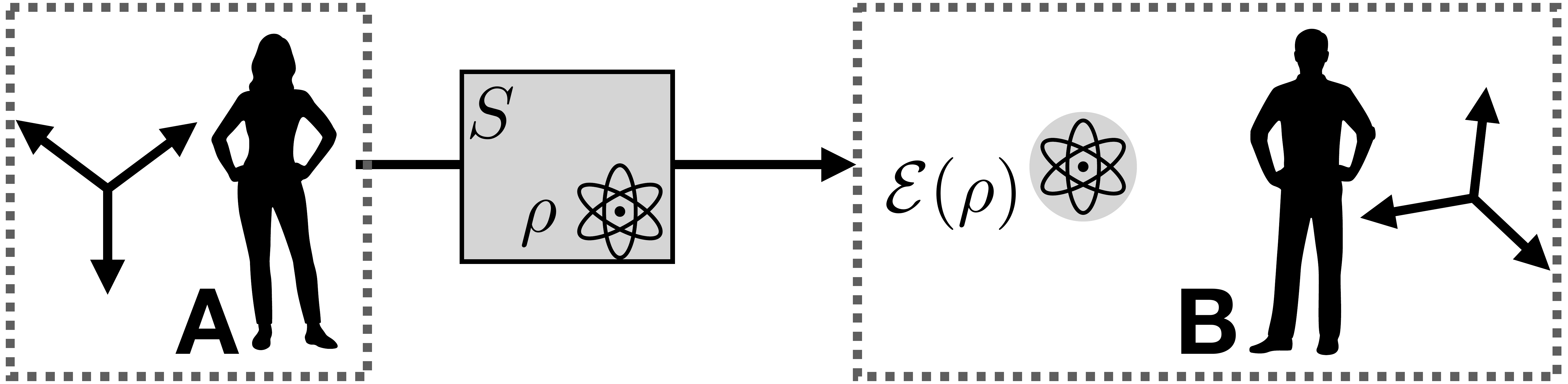}
\caption{A communication scenario within the \textbf{quantum information approach} as in Ref.~\cite{Bartlett}. The focus is on sending and recovering actual physical (quantum) states that are defined (as in Assumption 1) with respect to some external relatum, i.e.\ that may contain unspeakable information. This task becomes interesting if Alice's and Bob's reference frames are initially unaligned.}
\label{fig_qit}
\end{center}
\end{figure}

Suppose that Alice sends the quantum system physically to Bob. Since Bob's reference frame is not aligned with Alice's, he will describe the situation as receiving a randomly sampled representative of the equivalence class $[\rho]$. Thus, he will assign the state $\mathcal{E}(\rho):=\int_{\mathcal{G}} U_g \rho U_g^\dagger\, dg$ to the incoming quantum system.

The QI approach is concerned with the possibility to devise protocols that can be performed even in the absence of a shared reference frame. For example, the task to send quantum information from Alice to Bob can be accomplished by encoding it into a \emph{decoherence-free subspace}, i.e.\ a subsystem within the set of $\rho\in\mathcal{S}(\mathcal{H})$ for which $\mathcal{E}(\rho)=\rho$ (see e.g. Ref.~\cite[Sec.\ A.2]{Bartlett} for a concrete example). Another possibility to do so is by sending several quantum systems (e.g.\ spin-coherent states) that break the symmetry, and that allow Bob to partially correlate his reference frame with Alice's via suitable measurements on those states. The key to carrying out communication protocols without a shared frame is thus to focus on invariant physical properties that are meaningful in any external laboratory frame. This does not require describing $S$ relative to one of its subsystems.

Nevertheless, in the QI approach, the quantum nature of reference frames is sometimes taken into account, for example, by ``quantizing'' them to overcome superselection rules that arise in the absence of a shared classical frame~\cite{Bartlett}. This ``quantization'' of a frame means \emph{adding} a reference quantum system $R$ to the system of interest $S$ in order to define relative quantities between $R$ and $S$, such as relative phases~\cite{Bartlett} or relative distances~\cite{Smith2016,Smith2019}, that are invariant under $\cg_{\rm sym}$ and thereby meaningful relative to \emph{any} external laboratory frame. In a communication scenario between two parties Alice and Bob who do not share a classical frame, the reference system $R$ will typically be communicated together with $S$. While this also constitutes an internalization of a frame in the sense thats the reference system $R$ is now a quantum system too, it is still external to $S$. Furthermore, since the relative quantities between $R$ and $S$ are meaningful relative to any external laboratory frame with respect to which a measurement will be carried out, it is not necessary to take an extra step and ask how $S$ is described ``from the perspective'' of $R$ in order for Alice and Bob to succeed in their communication task.

In summary, in the QI approach, the quantum system $S$ of interest (say, a set of spins) is treated as a distinct entity from the reference frame (say, a gyroscope). Thus, ``QRF transformations'' relating descriptions relative to different subsystems (which may be in relative superposition) are typically not studied in this approach.\footnote{This includes Ref.~\cite{Palmer}, where transformations between different ``quantized'' reference systems $R_1$ and $R_2$ are studied. However, in the spirit of the QI approach, the derived transformations proceed between different invariant states (i.e.\ essentially $\cg$-twirls of $\rho_S\otimes\rho_{R_i}$, $i=1,2$) and are thus \emph{not} transformations between descriptions of the quantum state of $S$ relative to different choices of subsystem, as we will see them later. In particular, the descriptions of the quantum state of $S$ relative to different subsystems will be different descriptions of one and the same invariant state.}
The focus is on \emph{correlating} (aligning) Alice's and Bob's frames, and it is the absence of alignment that is modelled by the $\cg$-twirl, $\rho\mapsto\mathcal{E}(\rho)$. The external relatum independent (or relational) state descriptions of the QI approach are thus the \emph{incoherently} group-averaged states.

\subsubsection{The structural approach: agreeing on a redundancy-free internal description of quantum states}
\label{sssec_structural}

The structural approach does not stop at an external relatum independent state description. It also asks for a description of a quantum state relative to an internal frame that is part of the system of interest.

A transparent way to understand the structural approach operationally is as follows. Alice and Bob in their respective labs would like to agree on a concrete description of the quantum state of a system \emph{without} external relatum, i.e.\ in particular without shared reference frame. They have the option of describing $S$ in terms of the equivalence classes $[\rho]$ of quantum states. However, there is an evident \emph{redundancy} in the description of each equivalence class in terms of concrete quantum states: each member of the equivalence class is a legitimate (and non-unique) description of it. In order to break this redundancy and succeed in this task, they can take advantage of the fact that any equivalence class $[\rho]$ of states admits certain ``canonical choices'' for its description which are associated with different internal reference frame choices. The transformations relating these different canonical choices amount to ``QRF transformations'' and they will be elements of the symmetry group $\mathcal{G}_{\rm sym}$ defining the equivalence classes.

\begin{figure}[hbt]
\begin{center}
\includegraphics[width=.47\textwidth]{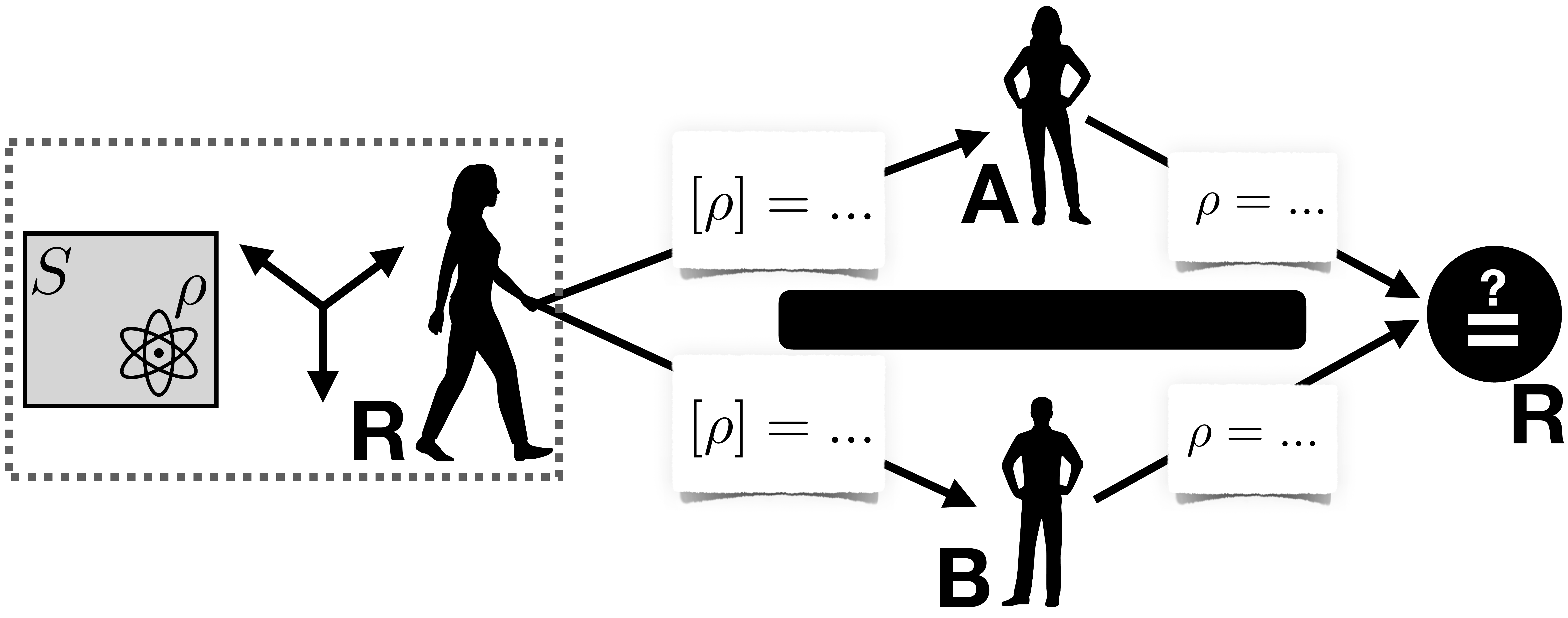}
\caption{A simple communication scenario which we choose for illustrating the operational essence of the \textbf{structural approach} as in Refs.~\cite{Giacomini,Vanrietvelde,Hamette}. The focus is on agents agreeing on a (redundancy-free) \emph{description} of quantum states in the absence of an external relatum.}
\label{fig_structural}
\end{center}
\end{figure}

For example, one could imagine the following communication scenario depicted in Figure~\ref{fig_structural} to illustrate the role of ``canonical choices'':
\begin{itemize}
\item Referee Refaella informs Alice and Bob in their separate laboratories that she will prepare quantum states of a particular system $S$ (subject to Assumption 1) relative to her (freely aligned) frame, but that she will only communicate the \emph{description} of the respective equivalence classes $[\rho]$ to Alice and Bob separately. 
\item Alice's and Bob's task is to separately return a concrete (redundancy-free) description of each quantum state to Refaella and they will win the game provided their descriptions always agree (either for all states of $S$, or for a particular class $\mathcal{C}$ of states).
\item Alice and Bob are only permitted to communicate prior to the beginning of the game to agree on a strategy.
\end{itemize}
Let us consider two examples for how this can be accomplished. These examples illustrate that there will generally exist multiple ``canonical choices'' for describing $[\rho]$ in terms of concrete quantum states, however, that Alice and Bob can  always agree in their communication beforehand which such choice to pick. This will also give a hint on the relation to quantum reference frames as described in Refs.~\cite{Giacomini,Vanrietvelde,Hamette}, and we will elaborate on this further in the following sections.
\begin{example}
\label{ExTrivial}
Consider a single quantum spin-$1/2$ particle, with state space $\mathcal{S}(\C^2)$.  Let us assume that the symmetry group is the full projective unitary group, i.e.\ $\mathcal{G}_{\rm sym}=\{\rho\mapsto U\rho U^\dagger\,\,|\,\, U^\dagger U=\mathbf{1}\}$, which is isomorphic to the rotation group ${\rm SO}(3)$.

Let $\rho$ be an arbitrary state that Refaella is for some reason interested in preparing. The equivalence class $[\rho]$ consists of all states with the same eigenvalues $\lambda_1,\lambda_2$ as $\rho$. Describing $[\rho]$ is equivalent to listing the eigenvalues $\lambda_1,\lambda_2$ and this is what Refaella may communicate to Alice and Bob. Clearly, there are many ways to represent this information in terms of a concrete quantum state $\rho$.

The strategy that Alice and Bob can agree on in order to win the game, but prior to it starting, is trivial: they can agree to always choose a basis (i.e.\ a specific reference frame alignment) such that $\rho$ described relative to it is a diagonal matrix. This leaves two ``canonical choices'' of representation: ordering the eigenvalues such that $\lambda_1\geq\lambda_2$ they could decide to always return either $\rho={\rm diag}(\lambda_1,\lambda_2)$ or $\rho={\rm diag}(\lambda_2,\lambda_1)$ to Refaella. The transformation relating the two descriptions is the unitary ``QRF transformation'' $U=\left(\begin{array}{cc} 0 & 1 \\ 1 & 0 \end{array}\right)$.
\end{example}
This trivial example relies on the simple fact that every quantum state has a \emph{canonical description}: the matrix representation in its own eigenbasis (up to a choice of order of eigenvalues). In some sense, every quantum state defines a finite set of natural representations of itself. It is in this sense that the structural approach interprets quantum systems as \emph{quantum reference frames}: the system's state breaks the fundamental symmetry, and admits, at least on the level of classical descriptions, a canonical choice of representation.

Example~\ref{ExTrivial} illustrates a general consequence of the symmetry structure: for any particular choice of QRF, the \emph{set} of state descriptions relative to that QRF corresponds in general only to a \emph{subset} or \emph{subspace} of states. In this example, any such choice only allows to describe a \emph{subset of states that corresponds to a classical bit}: namely, the convex hull of the density matrices ${\rm diag}(1,0)$ and ${\rm diag}(1/2,1/2)$. The following example demonstrates how a full subspace of states can be encoded.

\begin{example}
\label{ExTwoQubits}
Consider \emph{two} spin	-$1/2$ particles with rotational symmetry. That is, the symmetry group is $\mathcal{G}_{\rm sym}=\{\rho\mapsto U\otimes U \rho U^\dagger \otimes U^\dagger\,\,|\,\,U\in{\rm SU}(2)\}$, acting on states in $\mathcal{S}(\mathbb{C}^2\otimes\mathbb{C}^2)$. Let us make a somewhat arbitrary, but nonetheless illustrative choice of a class $\mathcal{C}$ of states for which the above communication game can be played. These will be the pure states
\be
   \mathcal{C}=\left\{\cos\frac\theta 2 |\phi_-\rangle+e^{i\varphi}\sin\frac\theta 2 |\phi\rangle\otimes|\phi\rangle\right\},
\ee
where $0\leq\theta\leq\pi$, $0\leq\varphi<2\pi$, $|\phi_-\rangle$ is the singlet state, and $|\phi\rangle\in\mathbb{C}^2$ an arbitrary normalized state. The set of states $\mathcal{C}$ is the disjoint union of the sets $\mathcal{C}_{\theta,\varphi}$ for which the two angles are fixed and $|\phi\rangle$ is still an arbitrary qubit state. Since $U\otimes U|\phi_-\rangle=|\phi_-\rangle$, the $\mathcal{C}_{\theta,\varphi}$ are orbits of the symmetry group, i.e.\ equivalence classes of states.

If Refaella gives Alice and Bob a description of such an equivalence class $[|\psi\rangle]=\mathcal{C}_{\theta,\varphi}$, they can agree on returning the standard description $|\psi'\rangle=\cos\frac\theta 2 |\phi_-\rangle+e^{i\varphi}\sin\frac\theta 2 |0\rangle\otimes|0\rangle$, for example. This prescription has the added benefit of \emph{preserving superposition across different equivalence classes}. Namely, if for $i=1,2$, we have $|\psi_i\rangle=\alpha_i|\phi_-\rangle+\beta_i|\phi\rangle\otimes|\phi\rangle$ such that $|\alpha_1|\neq |\alpha_2|$, then $|\psi_1\rangle$ and $|\psi_2\rangle$ are in different equivalence classes, and so are (in general) their superpositions. But the states that Alice and Bob return respect superpositions: if $|\psi\rangle=\kappa|\psi_1\rangle+\lambda|\psi_2\rangle$, then the returned states satisfy $|\psi'\rangle=\kappa|\psi'_1\rangle+\lambda|\psi'_2\rangle$. That is, this choice of QRF admits the description of a subspace, a \emph{qubit}, inside the joint state space. Other choices of QRF do so as well. These would correspond to canonical descriptions where $|0\rangle\otimes|0\rangle$ is replaced by some arbitrary $|\phi_0\rangle\otimes|\phi_0\rangle$, and they are related by ``QRF transformations'' $U\otimes U$.
\end{example}

There are also seemingly natural choices of QRF that, however, are deficient in that the set of admissible descriptions relative to them cannot encompass a state space, as the following example illustrates.
\begin{example}
Consider again two spin-$1/2$ particles, but now under slightly different circumstances. There is a canonical choice of factorization of the Hilbert space: by looking at the system in isolation, observers can determine the decomposition into two distinguishable particles. If we assume that this is the only structure that can be determined by such observers, then we have the symmetry group
\[
   \mathcal{G}_{\rm sym}=\{\rho\mapsto (U\otimes V) \rho (U^\dagger\otimes V^\dagger)\,\,|\,\, U^\dagger U=V^\dagger V=\mathbf{1}\}.
\]
Under these circumstances, a canonical choice of frame is such that any pure state $|\psi\rangle$ becomes identical to its own Schmidt representation, $\displaystyle |\psi\rangle=\sum_{i=0}^1 \sqrt{\alpha_i} |ii\rangle$, where $\alpha_0\geq \alpha_1$.

While Alice and Bob could easily agree on such a convention, the ensuing canonical description would not preserve complex superpositions and, in particular, not lead to a subspace of states as its image, owing to the real nature and ordering of the Schmidt-coefficients.
\end{example}
\emph{A priori}, a choice of QRF in the structural approach can therefore be quite arbitrary. However, as the examples above motivate, a ``good'' choice of QRF will correspond to one that admits the description of a set of states relative to it which carries sufficient convex or linear structure to encode classical or quantum information. Preferably, that set of states should correspond to a subspace of maximal size within $\mathcal{C}$.

In the remainder of this article, we will focus on a more interesting realization of such a scenario which reproduces the notion of QRFs in the structural picture. We will define particular systems $S$ that we call ``$\mathcal{G}$-systems'', and we will see that these carry an interesting group of symmetries $\mathcal{G}_{\rm sym}$. If we ask what kind of canonical choices of (redundancy-free) description $\mathcal{G}$-systems admit, such that Alice and Bob can succeed in the communication scenario of Figure~\ref{fig_structural}, we will find that these correspond to choosing one of the subsystems of $S$ as a reference system and to describing the remaining degrees of freedom relative to it. In this manner, we will recover and generalize the ``quantum states relative to a particle'' of Refs.~\cite{Giacomini,Vanrietvelde,Hamette}. In particular, the transformations among the canonical choices of description of $S$ are elements of the symmetry group $\mathcal{G}_{\rm sym}$ and exactly the QRF transformations of Ref.~\cite{Hamette}, which are also equivalent to the ones in \cite{Giacomini,Vanrietvelde} (restricted to a discrete setting). In Ref.~\cite{MMP} we will further explicitly demonstrate the equivalence with the perspective-neutral approach to QRFs \cite{Vanrietvelde} and elucidate that any equivalence class $[\rho]$ of quantum states above corresponds precisely to a perspective-neutral quantum state. As we will see, this means that the relational states of the structural approach are \emph{coherently} (not incoherently as in the QI approach) group-averaged states.

\section{From symmetries to QRF transformations and invariant observables}
\label{SecTechnical}
Quantum reference frames as described in Refs.~\cite{Giacomini,Vanrietvelde,Hamette} have first been considered for the case of wave functions on the real line. We have a Hilbert space of square-integrable functions, $\mathcal{H}=L^2(\R)$, and a physical claim that there is no absolute notion of origin. In other words, the ``physics'' does not change under translations (we will soon formulate what this means in detail). If we have $N$ particles on the real line, the total Hilbert space is $L^2(\R)^{\otimes N}$.

As noted in Ref.~\cite{Hamette}, the real numbers $\R$ play a double role in this case: on the one hand, they label the configuration space on which the wave functions are supported; on the other hand, they also label the possible translations, i.e.\ the fundamental symmetry group $(\R,+)$.

In this section, we will analyze this particular situation in a simplified setting: one in which the  group is finite and Abelian. In the simplest case, we discretize the real line and make it periodic, as in Figure~\ref{fig_cyclic}. Formally, for some $n\in\N$, we consider the \emph{cyclic group}
\be
   \mathbb{Z}_n:=\{0,1,2,\ldots,n-1\}
\ee
with addition modulo $n$ as its group operation. To this, we associate a single-particle Hilbert space
\be
   \mathcal{H}=\ell^2(\mathbb{Z}_n)={\rm span}\{|0\rangle,|1\rangle,\ldots,|n-1\rangle\}
\ee
and a total Hilbert space $\mathcal{H}^{\otimes N}$ for $N$ distinguishable particles. We will denote the particles with labels $A,B,C,\ldots$, and later in this paper with integers $1,2,3,\ldots$. Within this formalism, we can realize the main ideas of quantum references frames as in Refs.~\cite{Giacomini,Vanrietvelde,Hamette}. For the case $N=2$, consider the quantum state
\be
   |\psi\rangle_{AB}=|0\rangle_A\otimes \frac 1 {\sqrt{2}}\left(|1\rangle+|2\rangle\right)_B.
\ee
We are interested in a situation where ``only the relation between the particles'' matters, but not their total position. That is, in some sense, ``applying elements of $\mathbb{Z}_n$ to a quantum state doesn't change the physics''. Intuitively, this means, for example, that the quantum state
\be
   |\psi'\rangle_{AB}=|1\rangle_A\otimes \frac 1 {\sqrt{2}}\left(|2\rangle+|3\rangle\right)_B
\ee
should be an equivalent description of the system's properties, since it is related to $|\psi\rangle$ by a translation. Motivated by Ref.~\cite{Giacomini}, we can do something more interesting. First, in the terminology of Refs.~\cite{Giacomini,Vanrietvelde,Hamette}, the form of $|\psi\rangle$ can be interpreted as saying that ``particle $B$, as seen by $A$, is in the state $\frac 1 {\sqrt{2}}(|1\rangle+|2\rangle)$''. Second, we can then use the prescription of Refs.~\cite{Giacomini,Vanrietvelde,Hamette} to ``jump into $B$'s reference frame'', and consider the state
\be
   |\psi''\rangle=\frac 1 {\sqrt{2}}(|n-2\rangle+|n-1\rangle)_A\otimes |0\rangle_B
\ee
and conclude that ``particle $A$, as seen by $B$, is in the state $\frac 1 {\sqrt{2}}(|n-1\rangle+|n-2\rangle)$''. After all, this still expresses the fact that with amplitudes $\frac 1 {\sqrt{2}}$, $B$ is either one or two positions to the right of $A$.

We will now show that we can understand these transformations as natural symmetry transformations in a simple class of physical systems which we call ``$\mathcal{G}$-systems''. Choosing one of the particles as a reference frame (as sketched above) will correspond to a choice of canonical representation of a state as in the structural approach outlined above. This will give the idea of ``jumping into a particle's perspective'' a thorough operational interpretation.

\subsection{$\mathcal{G}$-systems and their symmetries}
\label{Section:Algebra}
We begin by considering a specific physical system which is motivated by translation-invariant quantum physics on the real line with Hilbert space $L^2(\R)$. Here we consider a finite, discrete group-theoretic analogue, again using the group as both the configuration space and set of transformations. Some aspects of QRF transformations in this case were also considered in Ref.~\cite{Hamette}. In contrast to Ref.~\cite{Hamette}, we restrict our attention to finite Abelian groups $\mathcal{G}$ for simplicity. Due to the structure theorem~\cite{Simon}, every such $\mathcal{G}$ can be interpreted as the group of translations of a discrete torus of some dimension. In the simplest case where $\mathcal{G}=\mathbb{Z}_n$, this torus is the circle\footnote{This representation is not unique. For example, we can interpret $\mathbb{Z}_6$ as the translation group of six points on a circle, but the structure theorem tells us that $\mathbb{Z}_6\simeq\mathbb{Z}_2\times\mathbb{Z}_3$. Thus, we can also interpret this group as the translations of a two-dimensional $(2\times 3)$-torus.}, and we are in the setting of Figure~\ref{fig_cyclic}.
\begin{definition}[$\mathcal{G}$-system]
\label{DefGSystem}
Fix some finite Abelian group $\mathcal{G}$, interpreted as a classical configuration space. That is, we regard the $g\in\mathcal{G}$ as perfectly distinguishable orthonormal basis vectors $|g\rangle$, spanning a Hilbert space $\mathcal{H}$. Formally, this Hilbert space is $\mathcal{H}=\ell^2(\mathcal{G})$, and it carries a distinguished basis $\{|g\rangle\}_{g\in\mathcal{G}}$, similarly as quantum mechanics on the real line carries a distinguished position basis.

Consider $N$ distinguishable particles on such a classical configuration space, where $N\in\N$. That is, the total Hilbert space is $\mathcal{H}^{\otimes N}$, and it carries a natural orthonormal basis
\be
   \mathcal{H}^{\otimes N}={\rm span}\{ |g_1,\ldots,g_N\rangle\,\,|\,\, g_i\in\mathcal{G}\}.
\ee
The physical system $S$ described by this Hilbert space will carry a group of symmetries $\mathcal{G}_{\rm sym}$ as introduced in Assumption 1 and Figure~\ref{fig_symmetry}. Clearly, the basic Hilbert space structure of $S$, i.e.\ the notion of linearity and the inner product, must not depend on the orientation of the external reference frame. Hence, the symmetry group will be of the form
\be
   \mathcal{G}_{\rm sym}=\{U\bullet U^\dagger\,\,|\,\,U\in\mathcal{U}_{\rm sym}\},
\ee
for $\mathcal{U}_{\rm sym}$ some group of unitaries. Furthermore, we assume that the classical configuration space, i.e.\ the \emph{set} of basis vectors, $\{|g_1,\ldots,g_N\rangle\,\,|\,\,g_i\in\mathcal{G}\}$, is an internal structure of $S$ that is defined without the external reference frame. We now postulate that the classical configurations carry $\mathcal{G}$-symmetry. In particular, any given configuration
\be
   |\mathbf{g}\rangle:=|g_1,g_2,\ldots,g_N\rangle
\ee
and its ``translated'' version
\be
   U_g^{\otimes N} |\mathbf{g}\rangle=|g\mathbf{g}\rangle:=|g g_1, gg_2,\ldots,g g_n\rangle
\ee
are internally indistinguishable. On the other hand, we postulate that the \emph{relation between the particles \textbf{is} accessible to observers without the external frame.} To formalize this, consider some tuple $\mathbf{h}\in\mathcal{G}^{N-1}$ of group elements, i.e.\ $\mathbf{h}=(h_1,\ldots,h_{N-1})$. Any state of the form
\begin{equation}
   |g,h_1 g, h_2 g, \ldots, h_{N-1}g\rangle=:|g,\mathbf{h}g\rangle
   \label{eqRelationFirstParticle}
\end{equation}
has the same pairwise relations between its particles, no matter what the state $|g\rangle$ of the first particle is. We now define $\mathcal{G}_{\rm sym}$ as the largest possible symmetry group that is compatible with these postulates. To this end, $\mathcal{U}_{\rm sym}$ must be the group of unitary transformations with the following properties:
\begin{enumerate}
	\item $U$ maps classical configurations to classical configurations, i.e. $U\ket{g_1,...,g_n} = \ket{g'_1,...,g'_n}$.
	\item On classical configurations, $U$ preserves relative positions, i.e. $U \ket{g,\mathbf h g} = \ket{g', \mathbf h g'}$.
	\item If two classical configurations are $g$-translations of each other, then $U$ preserves this fact, i.e.
\be
   |\mathbf{g}\rangle=U_g^{\otimes N}|\mathbf{j}\rangle \Rightarrow U|\mathbf{g}\rangle=U_g^{\otimes N} \left(\strut U |\mathbf{j}\rangle\right).
\ee
\end{enumerate}
\end{definition}
A few words of justification are in place. While two choices of external reference frame may yield a different description of any configuration, they must agree on the \emph{set} of all possible configurations that $S$ can be in, for otherwise their descriptions of $S$ cannot be placed in full relation with one another.\footnote{This assumes that the external frame choices in the ambient laboratory that an agent may have access to are complete in the sense that all quantum properties of $S$ can be described relative to them.} The \emph{set} of basis vectors $\{|\mathbf{g}\rangle\}_{\mathbf{g}\in\mathcal{G}^N}$ must thus be independent of the external frame and hence should remain invariant under $\mathcal{G}_{\rm sym}$. It is also clear that the symmetry group must preserve the linear and probabilistic structure of quantum theory and thereby leave the inner product on $\mathcal{H}^{\otimes N}$ invariant.\footnote{{In constraint quantization, $\mathcal{H}^{\otimes N}$ corresponds to the kinematical Hilbert space and so the preservation refers here to the kinematical inner product. While one is usually only interested in the physical inner product (i.e.\ the inner product on the space of solutions to the constraints), it nevertheless holds that also the kinematical inner product is left invariant by the group generated by the (self-adjoint) constraints.}} After all, by Assumption 1, symmetry related quantum states should be indistinguishable even probabilistically. Furthermore, the $\mathbf{h}$ label the 'relative positions' among the $N$ particles. These are internal properties of $S$ and so independent of any external relatum. Finally, configurations that are $g$-translations of each other are by assumption internally indistinguishable. The symmetry group must preserve this indistinguishability.

Note that it is possible to drop assumption 3., and to assume only 1.\ and 2. In this case, one obtains similar results to those presented here, but with modified structures: the algebra of invariant operators then becomes what we call $\mathcal{A}_{\rm alg}$ in Lemma~\ref{LemAAlg}, and the symmetry group becomes the group of conditional \emph{permutations} (not only conditional translations). Physically, this does not seem particularly well-motivated, and it leads to the loss of certain uniqueness results, including the uniqueness of $U\in\mathcal{U}_{\rm sym}$ in Theorem~\ref{lem_align1}.

The symmetry group of a $\mathcal{G}$-system can now easily be written down. To this end, define the subspaces
\be
   \mathcal{H}_{\mathbf{h}}:={\rm span}\{|g,\mathbf{h}g\rangle\,\,|\,\, g\in\mathcal{G}\}
\ee
and the corresponding orthogonal projectors $\Pi_{\mathbf{h}}$. Note that $\mathcal{H}^{\otimes N}=\bigoplus_{\mathbf{h}\in\mathcal{G}^{N-1}} \mathcal{H}_{\mathbf{h}}$, and the $\{\Pi_{\mathbf{h}}\}_{\mathbf{h}\in\mathcal{G}^{N-1}}$ define a projective measurement.

\begin{lemma}
\label{LemDecomposeU}
The symmetry group of a $\mathcal{G}$-system is 
\begin{equation}
   \mathcal{U}_{\rm sym}=\left\{\left.U=\bigoplus_{\mathbf{h}\in\mathcal{G}^{N-1}} U^{\otimes N}_{g_{\mathbf{h}}}\,\,\right|\,\, g_{\mathbf{h}}\in\mathcal{G}\right\},
  \label{eqRep}
\end{equation}
where $U_{g_{\mathbf{h}}}^{\otimes N}$ denotes the global translation by $g_{\mathbf{h}}$, but restricted to the subspace $\mathcal{H}_{\mathbf{h}}$.
\end{lemma}
That is, the symmetries in $\mathcal{U}_{\rm sym}$ act as \emph{relation-conditional global translations}: every classical configuration is globally translated via some $U_{g_{\mathbf{h}}}^{\otimes N}$, but the amount of translation $g_{\mathbf{h}}$ may depend on the relation $\mathbf{h}$ between the particles. We will soon identify the QRF transformations of Refs.~\cite{Giacomini,Vanrietvelde,Hamette} with elements of this group. Thus, the above highlights that these transformations make sense in a purely classical context; indeed, the corresponding classical frame transformations were also studied in \cite{Vanrietvelde,Hamette} and shown to be conditional on the interparticle relations.\footnote{More precisely, in the perspective-neutral approach these classical reference frame transformations correspond to \emph{conditional} gauge transformations, i.e.\ the gauge flow distance depends on the subsystem relations, see Appendix B of Ref.~\cite{Vanrietvelde} and also Refs.~\cite{Vanrietvelde2,Hoehn:2018aqt,Hoehn:2018whn}.} For example, they can also be applied if one deals with statistical mixtures of particle positions instead of superpositions. Nonetheless, their unitary extension to all of $\mathcal{H}^{\otimes N}$ leads to interesting quantum effects like the frame-dependence of entanglement~\cite{Giacomini,Vanrietvelde,Hamette}. This is similar to the behavior of the CNOT gate in quantum information theory, which is defined by its classical action on two bits, but nonetheless can create entanglement.
\begin{proof}
Due to conditions 1.\ and 2.\ of Definition~\ref{DefGSystem}, the $U\in{\mathcal{U}_{\rm sym}}$ leave every $\mathcal{H}_{\mathbf{h}}$ invariant. Thus, $U$ decomposes into a direct sum $U=\bigoplus_{\mathbf{h}\in\mathcal{G}^{N-1}} U_{\mathbf{h}}$. Fix some $\mathbf{h}\in\mathcal{G}^{N-1}$. Since $\mathcal{H}_{\mathbf{h}}$ is invariant, there exists some $g_{\mathbf{h}}\in\mathcal{G}$ such that $U|e,\mathbf{h}\rangle=|g_{\mathbf{h}},\mathbf{h} g_{\mathbf{h}}\rangle$. Now, for every $g\in\mathcal{G}$, we have $|g,\mathbf{h}g\rangle=U_g^{\otimes N}|e,\mathbf{h}\rangle$. Thus, condition 3.\ of Definition~\ref{DefGSystem} implies that
\begin{eqnarray}
U|g,\mathbf{h}g\rangle &=& U_g^{\otimes N}\left(\strut U|e,\mathbf{h}\rangle\right)=U_g^{\otimes N} |g_{\mathbf{h}},\mathbf{h} g_{\mathbf{h}}\rangle\nn\\
&=& U_{g g_{\mathbf{h}}}^{\otimes N}|e,\mathbf{h}\rangle=U_{g_{\mathbf{h}}}^{\otimes N}|g,\mathbf{h}g\rangle.
\end{eqnarray}
This shows that $U_{\mathbf{h}}$ acts like $U_{g_{\mathbf{h}}}^{\otimes N}$ on $\mathcal{H}_{\mathbf{h}}$.
\end{proof}

When working with pure state vectors, we sometimes want to allow global phases. Thus, we use the notation
\[
   \mathcal{U}_{\rm sym}^*:=\mathcal{U}_{\rm sym}\times {\rm U}(1)=\{e^{i\theta}U\,\,|\,\,U\in\mathcal{U}_{\rm sym},\theta\in\R\}.
\]
Above, we have decided to denote the state of the particles \emph{relative to the first particle}, but this also defines the relations between all other pairs of particles: the equation $|\mathbf{g}\rangle=|g,\mathbf{h}g\rangle\in\mathcal{H}_{\mathbf{h}}$ means that $g_i=h_{i-1}g_1$ for $i\geq 2$, but this implies that $g_i=(h_{i-1}h_{j-1}^{-1}) g_j$ for all $i,j$ if we set $h_0:=e$, the unit element of the group. Thus, the $\mathcal{H}_{\mathbf{h}}$ decompose the global Hilbert space into sectors of equal pairwise relations.

It is clear that global $\mathcal{G}$-translations are elements of the symmetry group, but they do not exhaust it:
\begin{example}
Given any $\mathcal{G}$-system, the global translations $U_g^{\otimes N}$ are symmetry transformations. Since they represent the global action of $\mathcal{G}$ on the $N$-particle Hilbert space, this can be written as
\be
   \mathcal{G}\subset \mathcal{G}_{\rm sym}.
\ee
However, there are other symmetries that are not global translations. For example, for $N=2$ particles, the unitary $U$ which acts on all basis vectors $|g_1,g_2\rangle$ as
\be
   U|g_1,g_2\rangle:=|g_2,g_1^{-1}g_2^2\rangle
\ee
is a symmetry transformation, i.e.\ $U\in\mathcal{U}_{\rm sym}$. Namely, $|g_1,g_2\rangle\in\mathcal{H}_h$ for $h=g_1^{-1}g_2$, and $U$ implements the global translation $U_{g(h)}^{\otimes 2}$ on $\mathcal{H}_h$, where $g(h)=h$.

On the other hand, the transformation
\be
   V|g_1,g_2\rangle:=|g_2^{-1},g_1^{-1}\rangle
\ee
is \emph{not} a symmetry transformation: it satisfies conditions 1.\ and 2.\ of Definition~\ref{DefGSystem}, but violates condition 3.
\end{example}
We will later see that QRF transformations correspond to elements in $\cg_{\rm sym}\setminus\cg$.

\subsection{Invariant observables and Hilbert space decomposition}\label{ssec_invobs}

Which observables can we internally measure in a $\mathcal{G}$-system, i.e.\ without access to the external relatum that was used to define the state space and the symmetry group? These must be the observables that are invariant under all symmetry transformations and which thus correspond to speakable information:
\begin{definition}[Invariant observable]\label{def_invobs} 
We define the \emph{invariant subalgebra} $\mathcal{A}_{\rm inv}$ as
\[
   \mathcal{A}_{\rm inv}=\{A\in\mathcal{L}(\mathcal{H}^{\otimes N})\,\,|\,\, [U,A]=0\mbox{ for all }U\in\mathcal{U}_{\rm sym}\},
\]
where $\mathcal{L}(\mathcal{H})$ denotes the set of linear operators on Hilbert space $\mathcal{H}$. These are the operators $A$ that are invariant under all symmetry transformations $A\mapsto U A U^\dagger$. A self-adjoint element $A=A^\dagger\in\mathcal{A}_{\rm inv}$ is called an \emph{invariant observable}.
\end{definition}
Since all observable properties of our system are assumed to be invariant under $\mathcal{G}_{\rm sym}$, it follows that the observables in Definition~\ref{def_invobs} comprise the set of \emph{all} observables that can be physically measured by an observer who does not have access to the external reference frame.

Clearly, all the $\Pi_{\mathbf{h}}$ are invariant observables, i.e.\ $\Pi_{\mathbf{h}}\in\mathcal{A}_{\rm inv}$. However, due to the fact that we have declared a classical basis to be a distinguished structure of the $\mathcal{G}$-system, there are many more invariant observables. To determine the algebra $\mathcal{A}_{\rm inv}$, recall the decomposition of $U\in\mathcal{U}_{\rm sym}$ from Lemma~\ref{LemDecomposeU}. We can regard $\mathcal{U}_{\rm sym}$ as a representation of several copies of the group $\mathcal{G}$, and thus further refine this decomposition via basic representation theory of finite Abelian groups~\cite{Simon}.

A major role is played by the \emph{characters} of $\mathcal{G}$. These are the homomorphisms $\chi:\mathcal{G}\to S^1$, i.e.\ the maps from $\mathcal{G}$ to the complex unit vectors $S^1:=\{z\in\mathbb{C},\,|\,\,|z|=1\}$ with $\chi(gh)=\chi(g)\chi(h)$. In other words, the characters are the one-dimensional irreducible representations (irreps) of $\mathcal{G}$, and these turn out to exhaust all irreps. The set of all characters of $\mathcal{G}$ is denoted $\mathcal{\hat G}$.

Denote the order of the group by $n:=|\mathcal{G}|$, then $g^n=e$ for all $g\in\mathcal{G}$~\cite{Simon}. Thus, every $\chi(g)$ must be among the $n$-th roots of unity: $\chi(g)^n=1$. Moreover, there are exactly $n$ characters, i.e.\ $|{\mathcal{\hat G}}|=n$.

Furthermore, note that $\dim\mathcal{H}_{\mathbf{h}}=n$. We claim that these subspaces can be decomposed as follows:
\be
   \mathcal{H}_{\mathbf{h}}=\bigoplus_{\chi\in\mathcal{\hat G}}\mathcal{H}_{\mathbf{h};\chi}
\ee
with $\mathcal{H}_{\mathbf{h};\chi}$ the one-dimensional subspace spanned by the vector
\begin{equation}
   |\mathbf{h};\chi\rangle:=\frac 1 {\sqrt{|\mathcal{G}|}} \sum_{g\in\mathcal{G}} \chi(g^{-1})|g,\mathbf{h}g\rangle.
   \label{eqChi}
\end{equation}
Indeed, due to Ref.~\cite[Proof of Corollary III.2.3]{Simon}, we have the well-known orthogonality relations $\sum_{g\in\mathcal{G}}\overline{\chi(g)}\chi'(g)=n\delta_{\chi,\chi'}$. Using this, direct calculation shows that the $|\mathbf{h};\chi\rangle$ are orthonormalized states, and
\begin{equation}
   U_g^{\otimes N}|\mathbf{h};\chi\rangle=\chi(g)|\mathbf{h};\chi\rangle\quad\mbox{for all }g\in\mathcal{G}.
   \label{EqEigenbasis}
\end{equation}

\begin{example}
\label{ExCyclic}
 	As a simple example, consider the cyclic group $\mathcal G = \mathbb{Z}_n = \{0,1,.\ldots, n-1\}$ with addition modulo $n$, see Figure 4. This group can be interpreted as a finite analogue of a part of the real line with periodic boundary conditions, by distributing finitely many possible positions along a ring. Its irreducible representations and the respective characters are given by $\chi_k(g) := e^{i \frac{2 \pi}{n} k g}$ with $k\in\{0,1,\ldots, n-1\}$~\cite{Tinkham}. Indeed, one can directly verify that the $\chi_k$ form one-dimensional representations of $\mathbb{Z}_n$, and they are inequivalent. We explicitly obtain
 	\begin{align}
 		\ket{\mathbf h; \chi_k} = \frac{1}{\sqrt n} \sum_{g=0}^{n-1} e^{-i \frac{2 \pi}{n} k g} \ket{g, g+ \mathbf h},
 	\end{align}
 	where $g + \mathbf h$ means that $g$ is added to each component of $\mathbf h$, modulo $n$. Similarly, one can directly verify that 
 	\begin{align}
 		U_g^{\otimes N} \ket{\mathbf h; \chi_k} = e^{i\frac{2\pi}{n} k g} \ket{\mathbf h; \chi_k}.
 	\end{align}
 	 One can see that the $\ket{\mathbf h; \chi_k}$ are obtained via a kind of discrete Fourier transform~\cite{NielsenChuang} from the classical configurations, and therefore they are reminiscent of momentum eigenstates. 
	
 Since elements of $\mathcal U_{\rm sym}$ translate all particles by the same amount, and momentum is the generator of translations, one may identify $\ket{\mathbf h; \chi_k}$ with the eigenstates of total momentum. However, since we are not explicitly interested in dynamics in this paper, we will postpone any elaboration on this analogy to our upcoming work, Ref.~\cite{MMP}.
\end{example}
From Eqs.~\eqref{eqRep}--\eqref{EqEigenbasis} it is clear that the subspace spanned by the eigenstates with trivial character $\chi=\mathbf{1}$ is the subspace of $\mathcal{U}_{\rm sym}$-invariant states, $|\psi\rangle=U|\psi\rangle$ for all $U\in\cu_{\rm sym}$. We denote it by
\begin{equation}
   \mathcal{H}_{\rm phys}:=\bigoplus_{\mathbf{h}\in\mathcal{G}^{N-1}} \mathcal{H}_{\mathbf{h};\mathbf{1}}
   ={\rm span}\left\{|\mathbf{h};\mathbf{1}\rangle\,\,|\,\,\mathbf{h}\in\mathcal{G}^{N-1}\right\}.
   \label{DefH1}
\end{equation}
We have equipped the total invariant subspace with the label ``$\rm phys$'' because it is the finite group version of the so-called physical Hilbert space of constraint quantization. When the symmetry group is generated by (self-adjoint) constraints, the physical Hilbert space corresponds to the set of solutions to the quantum constraints and is thereby precisely the Hilbert space on which the group acts trivially. It is usually called `physical' because quantum states of a gauge system are required to satisfy the constraints imposed by gauge symmetry. Nevertheless, we will see that we can give quantum states that are not invariant under the symmetry group a useful physical interpretation, and we will clarify their relation with the `physical' states in $\mathcal{H}_{\rm phys}$ in Ref.~\cite{MMP}. Being spanned by the states $\ket{\mathbf{h};\mathbf{1}}$ which encode the particle relations in an invariant manner, we shall henceforth also refer to $\ch_{\rm phys}$ as the subspace of \emph{relational states}. Its dimension is $|\cg|^{N-1}$, and thus:
\begin{lemma}\label{lem_hphys}
The subspace of relational states $\mathcal{H}_{\rm phys}$ is isomorphic to $\mathcal{H}^{\otimes (N-1)}$.
\end{lemma}

To determine the invariant subalgebra, consider any $A\in\mathcal{L}(\mathcal{H}^{\otimes N})$ and develop it into the $|\mathbf{h};\chi\rangle$-eigenbasis: $A=\sum_{\mathbf{h},\mathbf{h}',\chi,\chi'}a_{\mathbf{h},\mathbf{h'},\chi,\chi'}|\mathbf{h};\chi\rangle\langle\mathbf{h'};\chi'|$. Using Eqs.~(\ref{eqRep}) and~(\ref{EqEigenbasis}), conjugation with some $U\in\mathcal{U}_{\rm sym}$ yields
\[
   UAU^\dagger=\sum_{\mathbf{h},\mathbf{h'},\chi,\chi'}\chi(g_{\mathbf{h}})\chi'(g_{\mathbf{h'}})^{-1}a_{\mathbf{h},\mathbf{h'},\chi,\chi'}|\mathbf{h};\chi\rangle\langle\mathbf{h'};\chi'|.
\]
This is equal to $A$ for all $U$ if and only if for all $\mathbf{h},\mathbf{h'},\chi,\chi'$, one of the following is true: either $a_{\mathbf{h},\mathbf{h'},\chi,\chi'}=0$ or $\chi(g_{\mathbf{h}})=\chi'(g_{\mathbf{h'}})$ for all possible choices of $g_{\mathbf{h}}$, $g_{\mathbf{h'}}$. The latter condition is automatically satisfied if $\chi=\chi'=\mathbf{1}$. Thus, all operators $A$ that are fully supported on the relational subspace $\ch_{\rm phys}$
will be elements of $\mathcal{A}_{\rm inv}$. Let us denote the set of such operators by $\mathcal{A}_{\rm phys}$, then we have just shown that $\mathcal{A}_{\rm phys}\subset \mathcal{A}_{\rm inv}$. For reasons that will become clear later, we will call the observables in $\mathcal{A}_{\rm phys}$ \emph{relational observables}.

Now consider the other cases in which at least one of $\chi$ or $\chi'$ differs from $\mathbf{1}$. Clearly, if $\mathbf{h}=\mathbf{h'}$ and $\chi=\chi'$ then the character condition $\chi(g_{\mathbf{h}})=\chi'(g_{\mathbf{h'}})$ is trivially satisfied, and $a_{\mathbf{h},\mathbf{h},\chi,\chi}$ does not need to be zero. Consider the case $\mathbf{h}=\mathbf{h'}$ and $\chi\neq\chi'$. Choosing any $g_{\mathbf{h}}$ with $\chi(g_{\mathbf{h}})\neq \chi'(g_{\mathbf{h}})$ shows that we must have $a_{\mathbf{h},\mathbf{h},\chi,\chi'}=0$. Finally, if $\mathbf{h}\neq\mathbf{h'}$ and at least one of $\chi$ or $\chi'$ (say, $\chi$) differs from $\mathbf{1}$, choose $g_{\mathbf{h'}}=e$ and $g_{\mathbf{h}}$ such that $\chi(g_{\mathbf{h}})\neq 1$. This violates the character condition and implies $a_{\mathbf{h},\mathbf{h'},\chi,\chi'}=0$. In summary, we have proven the following:
\begin{lemma} \label{lem_algdecomp}
The invariant algebra consists exactly of the block matrices of the form
\be
   \mathcal{A}_{\rm inv}=\left\{A_{\rm phys}\oplus\bigoplus_{\mathbf{h}\in\mathcal{G}^{N-1}}\bigoplus_{\chi\neq\mathbf{1}} a_{\mathbf{h};\chi} |\mathbf{h};\chi\rangle\langle\mathbf{h};\chi|\right\},
\ee
where $A_{\rm phys}\in\mathcal{A}_{\rm phys}$ is supported on the relational subspace $\mathcal{H}_{\rm phys}$ defined in Eq.~(\ref{DefH1}), and the $a_{\mathbf{h};\chi}$ are complex numbers.
\end{lemma}
A few words are in place regarding the physical interpretation of these observables. Due to Eq.~(\ref{eqChi}), $\chi$ labels the irreps of the global translations on state space. As already mentioned in Example~\ref{ExCyclic}, they can thus be interpreted as a discrete analog of (an exponentiated version of) \emph{total momentum}. We can hence interpret the operator $|\mathbf{h};\chi\rangle\langle\mathbf{h};\chi|$ as describing a projective measurement that asks \emph{whether the relation between the particles is $\mathbf{h}$, and whether the total momentum corresponds to $\chi$}. Since this operator is contained in $\mathcal{A}_{\rm inv}$, this measurement can be performed by an observer without access to the external reference frame. In the special case if $\chi=\mathbf{1}$, i.e.\ on the relational subspace $\ch_{\rm phys}$ which corresponds to ``total momentum zero'', such an observer can also perform measurements that correspond to superpositions of different particle relations $\mathbf{h}$. However, for ``non-zero total momentum'' ($\chi\neq\mathbf{1}$), we obtain an emergent superselection rule that forbids such superpositions and the corresponding measurements.

The reader familiar with constraint quantization will notice that the invariant observables $A_{\rm phys}$ on the subspace $\mathcal{H}_{\rm phys}$ are the finite group analog of so-called Dirac observables~\cite{Dirac,Rovellibook,Thiemann}. Given some continuous group that is generated by an algebra of constraints, Dirac observables are operators that commute with the constraint operators (up to terms proportional to the constraints themselves). As such, they are invariant under the group generated by the constraints and observables on solutions to the constraints, i.e.\ on the so-called physical Hilbert space.

There is, however, a subtlety in this analogy: usually the (continuous) group generated by the constraints would be the analog of the `classical' group $\cg$ given here which is a strict subgroup of $\cg_{\rm sym}$. Thus, it is natural to ask whether the $\mathcal{G}_{\rm sym}$-invariant subspace $\mathcal{H}_{\rm phys}$ is a strict subset of the subspace of $\mathcal{G}$-invariant states. Accordingly, one may wonder whether the entire invariant algebra $\ca_{\rm inv}$ defined in terms of invariance under the larger group $\cg_{\rm sym}$ in Definition~\ref{def_invobs} is a strict subset of the algebra that is invariant under the smaller group $\cg$. It is this latter algebra which thus gives rise to the actual analog of Dirac observables for the finite groups considered here. We will address these questions in the next subsection.

\subsection{Group averaging states}
\label{SubsecGroupAveraging}
Although we work with a representation of the larger group $\cg_{\rm sym}$, Eqs.~\eqref{eqRep}--\eqref{DefH1} indicate that the total Hilbert space decomposes naturally in terms of the representation of the smaller group $\cg$; e.g., the physical Hilbert space is also precisely the subspace invariant under $\mathcal{G}$. We will now clarify this observation by considering the corresponding (coherent) group averaging operations, 
\be
\!\!\!\Pi_{\rm phys}:=\frac{1}{|\mathcal{U}_{\rm sym}|}\,\sum_{U\in\cu_{\rm sym}}\!\!\!U,\quad\Pi_{\rm phys}':=\frac{1}{|\cg|}\,\sum_{g\in\cg}\,U_g^{\otimes N},\label{projector}
\ee
which are standard in constraint quantization~\cite{Giulini:1998kf,Marolf:2000iq,Thiemann}, and for which the following holds.
\begin{lemma}\label{lem_proj}
The two coherent group averaging operations coincide, $\Pi_{\rm phys}=\Pi_{\rm phys}'$, and $\Pi_{\rm phys}$ is the orthogonal projector onto the relational subspace $\ch_{\rm phys}$.
\end{lemma}
\begin{proof}
Direct calculation shows that $\Pi_{\rm phys}^\dag=\Pi_{\rm phys}$ and $\Pi_{\rm phys}'^\dag=\Pi_{\rm phys}'$, and that $\Pi_{\rm phys}'=\Pi_{\rm phys}'^2$ and $\Pi_{\rm phys}=\Pi_{\rm phys}^2$. Thus, $\Pi_{\rm phys}$ and $\Pi_{\rm phys}'$ are orthogonal projectors. Since $\mathcal{H}^{\otimes N}$ is spanned by the $|g,\mathbf{h}g\rangle$ for $g\in\mathcal{G}$ and $\mathbf{h}\in\mathcal{G}^{N-1}$, the image ${\rm im}(\Pi_{\rm phys}')$ of $\Pi_{\rm phys}'$ is spanned by
\ba
\Pi'_{\rm phys}\ket{g,\mathbf{h}g} = \frac{1}{|\cg|}\,\sum_{g'\in\cg}\,U_{g'}^{\otimes N}\ket{g,\mathbf{h}g} = \frac{|\mathbf{h};\mathbf{1}\rangle}{\sqrt{|\cg|}}.
\label{eqProjPhys}
\ea
Since these states span $\mathcal{H}_{\rm phys}$, this proves that $\Pi_{\rm phys}'$ is the orthogonal projector onto the physical subspace. By construction, every $|\psi\rangle\in {\rm im}(\Pi_{\rm phys})$ is invariant under every $U\in\mathcal{U}_{\rm sym}$, and thus in particular under every $U_g^{\otimes N}\in\mathcal{U}_{\rm sym}$. Thus, ${\rm im}(\Pi_{\rm phys})\subseteq \mathcal{H}_{\rm phys}$. On the other hand, decomposing $U\in\mathcal{U}_{\rm sym}$ as in~(\ref{eqRep}), we get
\be
   \Pi_{\rm phys}|\mathbf{h};\mathbf{1}\rangle=\frac 1 {|\mathcal{U}_{\rm sym}|} \sum_{U\in\mathcal{U}_{\rm sym}}U_{g_{\mathbf{h}}}^{\otimes N}|\mathbf{h};\mathbf{1}\rangle =|\mathbf{h};\mathbf{1}\rangle
\ee
since $|\mathbf{h};\mathbf{1}\rangle$ is invariant under global translations. Thus, ${\rm im}(\Pi_{\rm phys})\supseteq \mathcal{H}_{\rm phys}$, and so $\Pi_{\rm phys}=\Pi_{\rm phys}'$.
\end{proof}
In conclusion, any basis state in $\ch_\mathbf{h}$ projects to the same invariant subnormalized state $\Pi_{\rm phys}\ket{g,\mathbf{h}g}=\Pi_{\rm phys}\ket{g',\mathbf{h}g'}=\frac{1}{\sqrt{|\cg|}}|\mathbf{h};\mathbf{1}\rangle$ under coherent group averaging, and it does not matter whether one averages with respect to the larger group $\cg_{\rm sym}$ or the smaller $\cg$.

However, we will now see that the set of invariant observables, i.e.\ the observables resulting from \emph{incoherent} group averaging ($\mathcal{G}$-twirling), differs for the two choices, but only outside of the relational subspace $\ch_{\rm phys}$. These operations are defined by
\begin{eqnarray}
   \Pi_{\rm inv}(\rho)&:=&\frac 1 {|\mathcal{U}_{\rm sym}|} \sum_{U\in\mathcal{U}_{\rm sym}} U\rho U^\dagger,\\
   \Pi'_{\rm inv}(\rho)&:=&\frac 1 {|\mathcal{G}|} \sum_{g\in\mathcal{G}} U_g^{\otimes N}\rho (U_g^{\otimes N})^\dagger.
\end{eqnarray}
It is well-known, and easy to check by direct calculation, that these maps are projectors, i.e.\ $\Pi_{\rm inv}^2=\Pi_{\rm inv}$ and ${\Pi'}_{\rm inv}^2=\Pi'_{\rm inv}$, and that they are orthogonal with respect to the Hilbert-Schmidt inner product, i.e.\ for all $A,B\in\mathcal{L}(\mathcal{H}^{\otimes N})$,
\be
   {\rm tr}\left(A^\dagger \Pi_{\rm inv}(B)\right)={\rm tr}\left(\Pi_{\rm inv}(A)^\dagger B\right).
\ee
If $A\in\mathcal{L}(\mathcal{H}^{\otimes N})$ satisfies $[U,A]=0$ for all $U\in\mathcal{U}_{\rm sym}$, then $\Pi_{\rm inv}(A)=A$. Conversely, if $B\in{\rm im}(\Pi_{\rm inv})$, then $UBU^\dagger=B$, i.e.\ $[U,B]=0$, for all $U\in\mathcal{U}_{\rm sym}$. Thus, $\Pi_{\rm inv}$ projects into the invariant algebra $\mathcal{A}_{\rm inv}$. Similarly, $\Pi'_{\rm inv}$ projects into
\[
   \mathcal{A}'_{\rm inv}=\{A\in\mathcal{L}(\mathcal{H}^{\otimes N})\,\,|\,\, [U_g^{\otimes N},A]=0\mbox{ for all }g\in\mathcal{G}\}.
\]
Clearly, $\mathcal{A}_{\rm inv}\subseteq\mathcal{A}'_{\rm inv}$, but are these algebras equal? The following lemma collects the above insights, and answers this question in the negative.
\begin{theorem}
\label{TheAlgebraProjections}
$\Pi_{\rm inv}$ is the orthogonal projector onto the invariant subalgebra $\mathcal{A}_{\rm inv}$. It can also be written in the form
\[
   \Pi_{\rm inv}(\rho)=\Pi_{\rm phys}\rho\Pi_{\rm phys}+\sum_{\mathbf{h},\chi\neq\mathbf{1}} \langle\mathbf{h};\chi|\rho|\mathbf{h};\chi\rangle |\mathbf{h};\chi\rangle\langle\mathbf{h};\chi|.
\]
Similarly, $\Pi'_{\rm inv}$ is the orthogonal projector onto the strictly larger subalgebra
\be
   \mathcal{A}'_{\rm inv}=\left\{\bigoplus_{\chi\in{\mathcal{\hat G}}}A_{\chi}\right\}=\left\{A_{\rm phys}\oplus\bigoplus_{\chi\neq\mathbf{1}}A_{\chi}\right\},
\ee
where every $A_{\chi}$ is an arbitrary operator supported on the subspace $\mathcal{H}_{\chi}:={\rm span}\{|\mathbf{h};\chi\rangle\,\,|\,\,\mathbf{h}\in\mathcal{G}^{N-1}\}$ (note that $\mathcal{H}_{\mathbf{1}}=\mathcal{H}_{\rm phys}$, so $A_{\rm phys}=A_{\mathbf{1}}$).
\end{theorem}
\begin{proof}
To see the claimed form of $\Pi_{\rm inv}$, note that the combination of projections is a Hilbert-Schmidt-orthogonal projection with image $\mathcal{A}_{\rm inv}$. It remains to be shown that $\mathcal{A}'_{\rm inv}$ has the claimed form. Note that $g\mapsto U_g^{\otimes N}$ is a representation of the finite Abelian group $\mathcal{G}$. It thus decomposes into one-dimensional irreps, and the equivalence classes of irreps are labelled by the characters $\chi$. Thus, the form of $\mathcal{A}'_{\rm inv}$ follows again from Schur's lemma.
\end{proof}
This theorem has interesting implications for the physical properties of $\mathcal{G}$-systems $S$. Recall our initial scenario as depicted in Figure~\ref{fig_symmetry}. Suppose that we \emph{only} demand symmetry of $S$ with respect to ordinary, \emph{unconditional} global translations $U_g^{\otimes N}$, and ask which observables can be measured by an observer without access to the external relatum. The answer is: all observables in $\mathcal{A}'_{\rm inv}$. On the other hand, if we demand symmetry with respect to all \emph{conditional} global translations in $\mathcal{U}_{\rm sym}$ --- and we will soon see that the QRF transformations of Refs.~\cite{Giacomini,Vanrietvelde,Hamette} are among those --- then this turns out to be a more stringent requirement. In this case, fewer observables are measurable, namely only those in $\mathcal{A}_{\rm inv}$.

In this sense, QRF transformations have fewer frame-independent observables than classical transformations: if all QRF transformations are symmetries, then superpositions of different particle relations $\mathbf{h}$ are forbidden by an emergent superselection rule whenever the ``total momentum is non-zero'', i.e.\ $\chi\neq \mathbf{1}$. On the other hand, if we only demand that global classical translations are symmetries, then these superpositions remain allowed.

Let us return to the analogy with contraint quantization discussed in the previous subsection. Lemma~\ref{lem_proj} and Theorem~\ref{TheAlgebraProjections} show that an observable $A_{\rm phys}$, i.e.\ the analog of a Dirac observable in our context, does \emph{not} depend on whether it is constructed relative to $\cg_{\rm sym}$ or its subgroup $\cg$. Later, we will also see that $\mathcal{A}_{\rm phys}$ is the finite group analog of the algebra generated by so-called \emph{relational} Dirac observables. These are invariant observables that encode relations between the subsystems, and they are common use in canonical quantum gravity \cite{Rovellibook,Thiemann,Tambornino,Rovelli1,Rovelli2,Rovelli3,Dittrich1,Dittrich2,Chataignier,Hoehn:2018aqt,Hoehn:2018whn,Hoehn:2019owq,Hoehn:2020epv}. This explains why we have called the observables in $\mathcal{A}_{\rm phys}$ ``relational observables''. They will become crucial in the resolution of the paradox of the third particle in Section~\ref{Section:Paradox}, and they will  turn out to be tomographically complete for the QRF states which we introduce in the next subsection.

Recall the notion of equivalence classes $[\rho]$ from Section~\ref{SecOperationalStructural}. We are now ready to introduce this notion formally for $\mathcal{G}$-systems:
\begin{definition}
\label{DefEquivalence}
We call two quantum states $\rho,\sigma\in\mathcal{S}(\mathcal{H}^{\otimes N})$ \emph{symmetry-equivalent}, and write $\rho\simeq\sigma$, if there exists some symmetry $U\in\mathcal{U}_{\rm sym}	$ such that $\sigma=U\rho U^\dagger$. We call them \emph{observationally equivalent}, and write $\rho\sim\sigma$, if ${\rm tr}(A\rho)={\rm tr}(A\sigma)$ for all invariant observables (and thus all operators) $A\in\mathcal{A}_{\rm inv}$.
\end{definition}
Clearly, if $\rho\simeq\sigma$ then $\rho\sim\sigma$, but the converse is not in general true. The equivalence class $[\rho]$ from Section~\ref{SecOperationalStructural} can now be defined as $[\rho]=\{\sigma\,\,|\,\,\sigma\simeq\rho\}$. In the case of pure state vectors $|\psi\rangle,|\psi'\rangle$, we must allow global phases and write $\psi\simeq\psi'$ if and only if there is some $U\in\mathcal{U}_{\rm sym}^*$ such that $|\psi'\rangle=U|\psi\rangle$.

Observational equivalence can be characterized in terms of the projection into the invariant subalgebra:
\begin{lemma}
\label{LemEquivProj}
Two states $\rho$ and $\sigma$ are observationally equivalent, i.e.\ $\rho\sim\sigma$, if and only if
\be
   \Pi_{\rm inv}(\rho)=\Pi_{\rm inv}(\sigma).
\ee
\end{lemma}
\begin{proof}
This follows from the chain of equivalences
\begin{eqnarray*}
&&\rho\sim\sigma\nn\\
&\Leftrightarrow& \langle A,\rho\rangle_{\rm HS}=\langle A,\sigma\rangle_{\rm HS}\quad\forall A\in\mathcal{A}_{\rm inv}\\
&\Leftrightarrow& \langle \Pi_{\rm inv}(B),\rho\rangle_{\rm HS}=\langle \Pi_{\rm inv}(B),\sigma\rangle_{\rm HS}\quad\forall B\in\mathcal{L}(\mathcal{H}^{\otimes N})\\
&\Leftrightarrow& \langle B,\Pi_{\rm inv}(\rho)\rangle_{\rm HS}=\langle B,\Pi_{\rm inv}(\sigma)\rangle_{\rm HS}\quad\forall B\in\mathcal{L}(\mathcal{H}^{\otimes N})\\\
&\Leftrightarrow& \Pi_{\rm inv}(\rho)=\Pi_{\rm inv}(\sigma).
\end{eqnarray*}
\vskip -2em
\end{proof}

\subsection{Alignable states as states with a canonical representation}
With these technical insights at hand, we are ready to return to the discussion of Section~\ref{SecOperationalStructural}. In the structural approach to QRFs, we ask whether a given state has a natural representation, depending only on internal data, such that the communication task of Figure~\ref{fig_structural} can be successfully accomplished. In the following, let us focus on pure states $|\psi\rangle\in\mathcal{H}^{\otimes N}$ for simplicity. Our task is to find another state $|\psi'\rangle\in\mathcal{H}^{\otimes N}$ that is symmetry-equivalent to $|\psi\rangle$ and that is in some sense distinguished, i.e.\ yields a ``canonical choice'' for describing the set of symmetry-equivalent states, cf.\ Section~\ref{sssec_structural}.

In general, there may be many different possible ways to define such a ``canonical choice''. Let us pick one possible choice. Suppose that we fix one of the particles, say, particle $i$, where $1\leq i \leq N$. Can we set the external reference frame such that this particle ends up at the ``origin'' --- the unit element of the group? In other words, can we align our state ``relative to particle $i$''? This is certainly possible in classical mechanics of $N$ point particles in one dimension, given translation-invariance. Classically, it would indeed define us a unique representation. We will now see that a similar construction can be done for $\mathcal{G}$-systems, and that it leads to the notion of QRF of Refs.~\cite{Giacomini,Hamette}.
\begin{definition}
\label{DefAlignable}
Let $i\in\{1,2,\ldots,N\}$. A pure state $|\psi\rangle\in\mathcal{H}^{\otimes N}$ is called \emph{$i$-alignable} if there exists some state $|\psi'\rangle\in\mathcal{H}^{\otimes N}$ with $\psi\simeq\psi'$ such that
\be
   |\psi'\rangle\equiv |\psi'\rangle_{1,\ldots,N}=|e\rangle_i\otimes |\varphi\rangle_{1,\ldots,i-1,i+1,\ldots,N}.
\ee
In the following, we will also use the notation $|\varphi\rangle_{\bar i}$ for the vector $|\varphi\rangle_{1,\ldots,i-1,i+1,\ldots,N}$.
\end{definition}
The state $|\varphi\rangle_{\bar{i}}$ in Definition~\ref{DefAlignable} is exactly what is interpreted in Refs.~\cite{Giacomini,Vanrietvelde,Hamette} as ``the state of the remaining $N-1$ particles as seen by particle $i$''. Similarly, we will thus interpret $|e\rangle_i\otimes|\varphi\rangle_{\bar{i}}$ as the description of the $N$ particle system relative to the QRF `perspective' defined by particle $i$ (which defines the origin).

Not all pure states are $i$-alignable. For example, the relational state $|\mathbf{h};\mathbf{1}\rangle$ is an element of the subspace $\mathcal{H}_{\rm phys}$, hence every $U\in\mathcal{U}_{\rm sym}^*$ satisfies $U|\mathbf{h};\mathbf{1}\rangle=e^{i\theta}|\mathbf{h};\mathbf{1}\rangle$ for some $\theta\in\R$. Thus this state cannot be $i$-alignable for any $i$. While devoid of alignable states, we will see later that $\ch_{\rm phys}$ contains the complete relational information about all alignable states.

To analyze this notion further, the following lemma will be useful.
\begin{lemma}\label{LemUniqueBasisVector}
For every $i\in\{1,\ldots,N\}$ and $\mathbf{h}\in\mathcal{G}^{N-1}$,	 there is a unique basis vector $|\mathbf{g}\rangle\in\mathcal{H}_{\mathbf{h}}$ with $g_i=g$, namely $|h_{i-1}^{-1}g,\mathbf{h} h_{i-1}^{-1}g\rangle$, with $h_0:=e$ the unit element of $\mathcal{G}$.
\end{lemma}
\begin{proof}
If $|\mathbf{g}\rangle\in\ch_{\mathbf{h}}$ then $|\mathbf{g}\rangle=|l,\mathbf{h}l\rangle$ for some $l\in\cg$. Hence $g_i=h_{i-1}l$ which is equal to $g$ if and only if $l=h_{i-1}^{-1}g$.
\end{proof}
This allows us to show that the state $|\psi'\rangle$ in Definition~\ref{DefAlignable} is unique, and thus defines indeed a natural representation of the symmetry-equivalence class of $|\psi\rangle$:
\begin{lemma}
\label{LemUniqueness}
If $|\psi\rangle\in\mathcal{H}^{\otimes N}$ is $i$-alignable, then the state $|\varphi\rangle_{\bar i}$ in Definition~\ref{DefAlignable} is unique up to a global phase.
\end{lemma}
\begin{proof}
Suppose that both $\psi'$ and $\tilde\psi'$ are states that satisfy the conditions of Definition~\ref{DefAlignable}. In particular, this means that $\psi\simeq\psi'$ and $\psi\simeq\tilde\psi'$, and so there exists some $U\in\mathcal{U}_{\rm sym}^*$ such that $|\tilde\psi'\rangle=U|\psi'\rangle$, and $U=e^{i\theta}V$ for $V\in\mathcal{U}_{\rm sym}$. In $|\psi'\rangle=|e\rangle_i\otimes|\varphi\rangle_{\bar i}$ and $|\tilde\psi'\rangle=|e\rangle_i\otimes|\tilde\varphi\rangle_{\bar i}$, we decompose $\varphi$ and $\tilde\varphi$ into product basis vectors: $|\varphi\rangle=\sum_{\mathbf{g}\in\mathcal{G}^{N-1}} \alpha_{\mathbf{g}} |g_1,\ldots,g_{N-1}\rangle$, and similarly for $|\tilde\varphi\rangle$ with amplitudes $\tilde\alpha_{\mathbf{g}}$. This implies that
\begin{eqnarray}
   U\sum_{\mathbf{g}\in\mathcal{G}^{N-1}} \alpha_{\mathbf{g}} |g_1,\ldots,g_{i-1},e,g_i,\ldots,g_{N-1}\rangle\nn\\
   =\sum_{\mathbf{g}\in\mathcal{G}^{N-1}} \tilde\alpha_{\mathbf{g}} |g_1,\ldots,g_{i-1},e,g_i,\ldots,g_{N-1}\rangle.
\end{eqnarray}
Now, according to Lemma~\ref{LemUniqueBasisVector}, both of the decompositions $\sum_{\mathbf{g}\in\mathcal{G}^{N-1}}\ldots$ contain \emph{at most one basis vector from every subspace $\mathcal{H}_{\mathbf{h}}$} with non-zero amplitude, namely $|h_{i-1}^{-1},\mathbf{h} h_{i-1}^{-1}\rangle$. But since $U$ (and thus $V$) leaves the subspaces $\mathcal{H}_{\mathbf{h}}$ invariant, the last equation implies that $V |h_{i-1}^{-1},\mathbf{h} h_{i-1}^{-1}\rangle=|h_{i-1}^{-1},\mathbf{h} h_{i-1}^{-1}\rangle$ for every such vector that appears with non-zero amplitude $\alpha_{\mathbf{g}}\neq 0$ (and thus $\tilde\alpha_{\mathbf{g}}\neq 0$). Hence $\tilde\alpha_{\mathbf{g}}=e^{-i\theta}\alpha_{\mathbf{g}}$, and so $\tilde\varphi=e^{-i\theta}\varphi$.
\end{proof}

\subsection{QRF transformations as symmetry group elements}

In the structural approach in Refs.~\cite{Giacomini,Vanrietvelde,Hamette}, we can ``jump'' from one particle's reference frame into any other's. How is this idea expressed in our formalism? To see this, let us first show the following.
\begin{theorem}[QRF state transformations] \label{lem_align1}
If there is some $i\in\{1,2,\ldots,N\}$ such that $|\psi\rangle$ is $i$-alignable, then $|\psi\rangle$ is $j$-alignable for every $j\in\{1,2,\ldots,N\}$. We will then simply call $|\psi\rangle$ \emph{alignable}. Moreover, for every $i,j\in\{1,\ldots,N\}$, there is a unique symmetry transformation $U\in\mathcal{U}_{\rm sym}$ such that $U\left(|e\rangle_i\otimes|\varphi\rangle_{\bar i}\right)=|e\rangle_j\otimes|\varphi\rangle_{\bar j}$ for all $|\varphi\rangle_{\bar i}$. Furthermore, if $i\neq j$ then $U$ is a proper \emph{conditional} global translation, i.e.\ $U\bullet U^\dagger\in\mathcal{G}_{\rm sym}\setminus\mathcal{G}$. Every such $U$ induces a \emph{unique} unitary (``QRF transformation'') $V_{i\to j}$ such that $V_{i\to j}|\varphi\rangle_{\bar i}=|\varphi\rangle_{\bar j}$. This transformation can be written
\be
   V_{i\to j}=\mathbb{F}_{i,j}\sum_{g\in\mathcal{G}}|g^{-1}\rangle\langle g|_j\otimes U_{g^{-1}}^{\otimes(N-2)},
\ee
where $\mathbb{F}_{i,j}$ flips (swaps) particles $i$ and $j$. This is the discrete version of the form given in Refs.~\cite{Giacomini,Hamette}.
\end{theorem}
\begin{proof}
Fix $i,j\in\{1,\ldots,N\}$. For every $\mathbf{h}\in\mathcal{G}^{N-1}$, let $g_{\mathbf{h}}:=h_{j-1}^{-1}h_{i-1}$ (setting, as before, $h_0:=e$). Then the global translation by $g_{\mathbf{h}}$ satisfies
\be
   U_{g_{\mathbf{h}}}^{\otimes N} |h_{i-1}^{-1},\mathbf{h} h_{i-1}^{-1}\rangle=|h_{j-1}^{-1},\mathbf{h} h_{j-1}^{-1}\rangle.
\ee
Set $U:=\bigoplus_{\mathbf{h}\in\mathcal{G}^{N-1}} {U^{\otimes N}_{g_{\mathbf{h}}}}$, then $U\in\mathcal{U}_{\rm sym}$. According to Lemma~\ref{LemUniqueBasisVector}, for every $\mathbf{h}$, $U$ maps the unique basis vector $|\mathbf{g}\rangle\in\mathcal{H}_{\mathbf{h}}$ with $g_i=e$ to the unique basis vector $|\mathbf{g'}\rangle\in\mathcal{H}_{\mathbf{h}}$ with $g'_j=e$, and it is clear that $U$ is the only symmetry transformation that does this. Thus, $U$ maps all states of the form $|e\rangle_i\otimes|\varphi\rangle_{\bar i}$ to states of the form $|e\rangle_j\otimes|\tilde\varphi\rangle_{\bar j}$. Furthermore, if $i\neq j$ then there exist $\mathbf{h},\mathbf{j}$ such that $g_{\mathbf{h}}\neq g_{\mathbf{j}}$. Thus, any such $U$ is an $\mathbf{h}$-dependent transformation and thus cannot be a global translation.

Fix an arbitrary orthonormal basis $\{|\varphi\rangle_{\bar i}\}_{\varphi}$ of $\mathcal{H}^{\otimes(N-1)}$, then $U\left(|e\rangle_i\otimes|\varphi\rangle_{\bar i}\right)=|e\rangle_j\otimes|\varphi\rangle_{\bar j}$ yields another orthonormal basis. Thus, we can view this as a unitary $V_{i\to j}$ from $\mathcal{H}^{\otimes(N-1)}$ into another copy of $\mathcal{H}^{\otimes(N-1)}$. Since its action on basis vectors is fixed, there can be no more than one such map. To determine that it has the form as claimed, simply look at its action on the basis vectors $|g_1,\ldots,g_{N-1}\rangle$.
\end{proof}
So indeed, for every alignable state, and any particle $j\in\{1,\ldots,N\}$, there is a unique representation of that state ``relative to the $j$th particle''. Furthermore, the QRF transformation from $i$'s to $j$'s `perspective' at the level of the \emph{full} Hilbert space $\ch^{\otimes N}$ corresponds to a symmetry transformation which lies in $\mathcal{G}_{\rm sym}$, but \emph{not} in $\mathcal{G}$ (if $i\neq j$). This observation highlights the physical significance of the symmetry group $\cg_{\rm sym}$. While we have seen that the set of invariant states $\mathcal{H}_{\rm phys}$ is independent of whether one constructs it through coherently averaging over $\cg_{\rm sym}$ or its `classical translation subgroup' $\cg$, the symmetry group $\cg_{\rm sym}$ is key for understanding the meaning of the QRF transformations (which transform non-invariant descriptions): they are \emph{conditional} symmetry transformations that depend on the interparticle relation $\mathbf{h}$.

To clarify the notation used in the definition of the QRF transformation $V_{i\to j}$, we give a simple example.
\begin{example}
Suppose we have $N=4$ particles, and an alignable state $|\psi\rangle$ such that
\be
   |\psi\rangle\simeq |e\rangle_2\otimes |g_1,g_3,g_4\rangle.
\ee
Thus, the state relative to the second particle is $|g_1,g_3,g_4\rangle$. To determine the state relative to the third particle, compute
\begin{eqnarray}
V_{2\to 3}|g_1,g_3,g_4\rangle&=&\mathbb{F}_{2,3}\sum_{g\in\mathcal{G}}|g^{-1}\rangle\langle g|_3\otimes U_{g^{-1}}^{\otimes 2} |g_1,g_3,g_4\rangle\nn\\
&=& \mathbb{F}_{2,3}|g_3^{-1}\rangle\langle g_3|\otimes U_{g_3^{-1}}^{\otimes 2}|g_1,g_3,g_4\rangle\nn\\
&=& \mathbb{F}_{2,3} |\underbrace{g_3^{-1}g_1}_{1},\underbrace{g_3^{-1}}_{3},\underbrace{g_3^{-1}g_4}_{4}\rangle\nn\\
&=& |\underbrace{g_3^{-1}g_1}_{1},\underbrace{g_3^{-1}}_{2},\underbrace{g_3^{-1}g_4}_{4}\rangle,
\end{eqnarray}
where the integers at the bottom denote the particle labels.
\end{example}
In Ref.~\cite{MMP}, we will demonstrate equivalence of the above QRF transformations with the ``quantum coordinate changes'' of the perspective-neutral approach \cite{Vanrietvelde,Vanrietvelde2,Hoehn:2018aqt,Hoehn:2018whn,Hoehn:2019owq,Hoehn:2020epv}.

As Theorem~\ref{lem_align1} has shown, the symmetry group $\mathcal{U}_{\rm sym}$ contains the QRF transformations which switch from state descriptions relative to particle $i$ to descriptions relative to particle $j$. But these do not exhaust the symmetry group. The following lemma gives another physically motivated example.
\begin{example}[Center of mass]
\label{ExCenterMass}
Consider again the cyclic group $\mathbb{Z}_n$ with addition modulo $n$, as shown in Figure~\ref{fig_cyclic}. Let $m_1,\ldots,m_N$ be non-negative real numbers and $m:=m_1+\ldots+m_N$. For $\mathbf{h}\in\mathcal{G}^{N-1}$, define the group elements
\be
   g(\mathbf{h}):=-\left\lfloor \frac 1 m\left(m_2 h_1+\ldots+m_N h_{N-1}\right)\right\rfloor,
\ee
and set $U:=\bigoplus_{\mathbf{h}\in\mathcal{G}^{N-1}} U_{g(\mathbf{h})}^{\otimes N}$. If we interpret the $m_i$ as the \emph{masses} of the particles, with the origin as the position of particle $1$ (i.e.\ $h_0=0$), then this symmetry transformation $U$ describes a change of quantum coordinates such that the (integer part of) the ``center of mass'' becomes the origin.
\end{example}

\subsection{Characterization of alignable states}

We have already seen that not all global states are alignable. The following lemma characterizes those that are.
\begin{lemma} \label{lem_align}
A state $|\psi\rangle\in\mathcal{H}^{\otimes N}$ is alignable if and only if it can be written in the form
\begin{equation}
   |\psi\rangle=\sum_{\mathbf{h}\in\mathcal{G}^{N-1}} \alpha_{\mathbf{h}} |g_{\mathbf{h}},\mathbf{h}\, g_{\mathbf{h}}\rangle
   \label{eqFormAlignable}
\end{equation}
for some $g_{\mathbf{h}}\in\mathcal{G}$ and $\alpha_{\mathbf{h}}\in\C$. Moreover, the $\alpha_{\mathbf{h}}$ characterize the alignable state up to symmetry-equivalence. That is, for two alignable states $|\psi\rangle$ and $|\psi'\rangle$, we have $\psi\simeq\psi'$ if and only if their coefficients satisfy $\alpha_{\mathbf{h}}=e^{i\theta}\alpha'_{\mathbf{h}}$ for some $\theta\in\R$.
\end{lemma}
\begin{proof}
If $|\psi\rangle$ is alignable, then it is in particular $1$-alignable. That is, there exists $U\in\mathcal{U}_{\rm sym}^*$ such that
\[
   |\psi\rangle=U\left(\strut |e\rangle_1\otimes |\varphi\rangle_{\bar 1}\right)=e^{i\theta}\sum_{\mathbf{h}\in\mathcal{G}^{N-1}} \alpha_{\mathbf{h}} {U^ {\otimes N}_{g_{\mathbf{h}}}} |e,\mathbf{h}\rangle,
\]
where the $\alpha_{\mathbf{h}}$ are the coefficients of $|\varphi\rangle_{\bar 1}$ in terms of the product group basis, and the $U^ {\otimes N}_{g_{\mathbf{h}}}$ translate the basis vectors of the subspaces $\mathcal{H}_{\mathbf{h}}$. Hence we get the claimed form for $|\psi\rangle$. The converse direction of the proof of the first part of this lemma is analogous. We also see that the $\alpha_{\mathbf{h}}$ only depend on the symmetry-equivalence class of the state $|\psi\rangle$ (up to a global phase); and, in the case of equality of those coefficients, two states must have the same $|\varphi\rangle_{\bar 1}$ (up to a global phase), hence they must be symmetry-equivalent.
\end{proof}
The above form shows that any maximal subspace of $\mathcal{H}^{\otimes N}$ which is contained in the set of alignable states has dimension $|\mathcal{G}|^{N-1}$, i.e.\ is isomorphic to $\mathcal{H}^{\otimes(N-1)}$. This shows that the QRFs as defined above are indeed ``good'' QRFs as explained in Subsection~\ref{sssec_structural}. Due to Lemma~\ref{lem_hphys}, it also shows that the relational subspace $\ch_{\rm phys}$ has the right dimension for its states to contain ``all the particles' internal QRF perspectives at once''. This observation is corroborated by the following useful Lemma and will be further discussed in Ref.~\cite{MMP}. 
\begin{lemma}
\label{LemProjection}
Given any alignable state $|\psi\rangle$ (which is hence of the form~(\ref{eqFormAlignable})), its projection onto the invariant subalgebra is
\ba
   &&\Pi_{\rm inv}(|\psi\rangle\langle\psi|)\\
   &=&\sum_{\mathbf{h},\mathbf{j}\in\mathcal{G}^{N-1}} \frac{\alpha_{\mathbf{h}}\overline{\alpha_{\mathbf{j}}}}{|\mathcal{G}|}|\mathbf{h};\mathbf{1}\rangle\langle\mathbf{j};\mathbf{1}|+\sum_{\mathbf{h}\in\mathcal{G}^{N-1}} \frac{|\alpha_{\mathbf{h}}|^2}{|\mathcal{G}|}\Pi_{\mathbf{h};\chi\neq\mathbf{1}},\nn
\ea
where $\Pi_{\mathbf{h};\chi\neq\mathbf{1}}:=\sum_{\chi\neq\mathbf{1}}|\mathbf{h};\chi\rangle\langle\mathbf{h};\chi|$. Thus, for any two alignable states $|\psi\rangle,|\psi'\rangle$, we have $\psi\sim\psi'$ if and only if $\psi\simeq\psi'$: such states are symmetry-equivalent if and only if they are observationally equivalent. Furthermore, they are equivalent if and only if $\langle\psi|A|\psi\rangle=\langle\psi'|A|\psi'\rangle$ for all $A\in\mathcal{A}_{\rm phys}$; that is, all invariant information of alignable states is fully contained in their projection $\Pi_{\rm phys}|\psi\rangle$.
\end{lemma}
\begin{proof}
Using~(\ref{eqProjPhys}), we obtain
\be
   \Pi_{\rm phys}|\psi\rangle=\frac 1 {\sqrt{|\mathcal{G}|}} \sum_{\mathbf{h}\in\mathcal{G}^{N-1}}\alpha_{\mathbf{h}} |\mathbf{h};\mathbf{1}\rangle.
\ee
Direct calculation shows that $|\langle \psi|\mathbf{h};\chi\rangle|^2=|\alpha_{\mathbf{h}}|^2/|\mathcal{G}|$. The result then follows by using the form of $\Pi_{\rm inv}$ as given in Theorem~\ref{TheAlgebraProjections}. Both notions of equivalence boil down to the fact that the states have the same amplitudes $\alpha_{\mathbf{h}}$ up to a global phase. Finally, for $\rho$ any alignable state, the above form of $\Pi_{\rm inv}(\rho)$ implies
\begin{eqnarray}
   &&{\rm tr}(\rho A)\label{eqSubnormalized}\\&=&{\rm tr}(\rho_{\rm phys} A)+\sum_{\mathbf{h}\in\mathcal{G}^{N-1}}\langle\mathbf{h};\mathbf{1}|\rho_{\rm phys}|\mathbf{h};\mathbf{1}\rangle{\rm tr}(\Pi_{\mathbf{h};\chi\neq\mathbf{1}}A)\nn
\end{eqnarray}
for all $A\in\mathcal{A}_{\rm inv}$, where $\rho_{\rm phys}:=\Pi_{\rm phys}\rho\Pi_{\rm phys}$. That is, the expectation values of all invariant observables, and thus the notion of observational equivalence, depends only on the state's projection into $\mathcal{A}_{\rm phys}$.
\end{proof}
In other words, the algebra $\ca_{\rm phys}$ is tomographically complete for the invariant information in alignable states. In this sense, it can be said that the external relatum independent states of the structural approach are what we called the relational states, namely the ones in $\ch_{\rm phys}$. That is, the external relatum independent states of the structural approach are the \emph{coherently} group-averaged states, while their counterparts in the QI approach are the \emph{incoherently} group-averaged states (cf.\ Subsection~\ref{sssec_qiapp}).\footnote{We thank A.\ R.\ H.\ Smith for suggesting us to emphasize this technical distinction.}

\subsection{Alignable and relational observables}

Given the notion of alignable states and the duality between states and observables, it is natural to also define alignable observables in the obvious way.
\begin{definition}[Alignable observables]
An operator $A\in\mathcal{L}(\ch^{\otimes N})$ is called \emph{$i$-alignable} for $i\in\{1,\ldots,N\}$ if there exists $U\in\cu_{\rm sym}$ such that
\be
    UAU^\dag = \ket{e}\!\bra{e}_i\otimes A_{\bar{i}}
\ee
for some $A_{\bar{i}}\in\mathcal{L}(\ch^{\otimes(N-1)})$. If $A$ is an observable, it is called an $i$-alignable observable.
\end{definition}
This leads to the following extension of Theorem~\ref{lem_align1} from QRF state to observable transformations. The proof is analogous and thus omitted.

\begin{theorem}[QRF observable transformations]\label{thm_alignobs}
If $A\in\mathcal{L}(\ch^{\otimes N})$ is $i$-alignable, then it is also $j$-alignable for every $j\in\{1,\ldots,N\}$, and so we will call $A$ simply \emph{alignable}. In particular, there is a unique $U\in\cu_{\rm sym}$ such that 
\ba
U\left( \ket{e}\!\bra{e}_i\otimes A_{\bar{i}}\right)U^\dag = \ket{e}\!\bra{e}_j\otimes A_{\bar{j}},
\ea
where $A_{\bar{j}} = V_{i\to j}\,A_{\bar{i}}\,V_{j\to i}$. Here, $U$ is the unique symmetry transformation from Theorem~\ref{lem_align1} such that $U\left(\ket{e}_i\otimes\ket{\varphi}_{\bar{i}}\right)=\ket{e}_j\otimes\ket{\varphi}_{\bar{j}}$ for all $\ket{\varphi}_{\bar{i}}$ and $V_{i\to j}$ is the unitary ``QRF transformation'' induced by it.
\end{theorem}
These are the discrete versions of the observable transformations in Refs.~\cite{Giacomini,Vanrietvelde}.

Finally, note that if $A_{\bar 1}=\sum_{\mathbf{h},\mathbf{j}\in\mathcal{G}^{N-1}}a_{\mathbf{h},\mathbf{j}}|\mathbf{h}\rangle\langle\mathbf{j}|$, then
\be
   \hat\Pi_{\rm phys}\left(\strut |e\rangle\langle e|_1\otimes A_{\bar 1}\right)=\sum_{\mathbf{h},\mathbf{j}\in\mathcal{G}^{N-1}}\frac{a_{\mathbf{h},\mathbf{j}}}{|\mathcal{G}|} |\mathbf{h};\mathbf{1}\rangle\langle\mathbf{j};\mathbf{1}|,
\ee
where $\hat\Pi_{\rm phys}$ henceforth denotes the superoperator that acts as $\hat\Pi_{\rm phys}(A)=\Pi_{\rm phys}\,A\,\Pi_{\rm phys}$, as obvious from Lemma~\ref{LemProjection}. Up to a factor of $|\mathcal{G}|$, this is identical to the original representation of $A_{\bar 1}$, but in another basis. This proves the following theorem, extending a result from~\cite{Hoehn:2019owq,Hoehn:2020epv}:
\begin{theorem}\label{thm_relobs}
Consider the map
\begin{equation}
   A_{\bar i}\mapsto F_{A_{\bar{i}},i}:=|\mathcal{G}|\cdot\hat\Pi_{\rm phys}\left(\strut |e\rangle\langle e|_i\otimes A_{\bar i}\right),
   \label{eqRelationalObservableExpl}
\end{equation}
where $A_{\bar i}\in\mathcal{L}(\mathcal{H}^{\otimes(N-1)})$ is any operator. This defines an isomorphism between the operators $A_{\bar i}$ and $\mathcal{A}_{\rm phys}$, preserving products, linear combinations, and adjoints.
\end{theorem}
This gives us two independent motivations to focus on $\mathcal{A}_{\rm phys}$: for alignable states, the projection into this subalgebra contains all invariant information; and it does so in a way that preserves the natural structure of the alignable observables. 

Furthermore, returning to the comparison with constraint quantization, it follows from Refs.~\cite{Hoehn:2019owq,Hoehn:2020epv} that $F_{A_{\bar{i}},i}$ is (the finite group analog of) the \emph{relational Dirac observable} which encodes in an invariant manner the question ``what is the value of $A_{\bar i}$ given that particle $i$ sits at the origin?'' Such relational observables are a standard tool in canonical quantum gravity, e.g.\ see Refs.~\cite{Rovellibook,Thiemann,Tambornino,Rovelli1,Rovelli2,Rovelli3,Dittrich1,Dittrich2,Chataignier,Hoehn:2018aqt,Hoehn:2018whn}. Theorem~\ref{thm_relobs} is the reason why we refer to $\ca_{\rm phys}$ as the algebra generated by relational observables.

Due to $\mathcal{A}_{\rm phys}\subset\mathcal{A}_{\rm inv}$, we have
\[
   \hat\Pi_{\rm phys}\left(\Pi_{\rm inv}(\ket{e}\!\bra{e}_i\otimes A_{\bar{i}})\right)=\hat\Pi_{\rm phys}\,(\ket{e}\!\bra{e}_i\otimes A_{\bar{i}})\,.
\]
Since $\Pi_{\rm inv}$ is the incoherent $\mathcal{G}$-twirl over $\mathcal{U}_{\rm sym}$, it is clear that the image $\Pi_{\rm inv}(\ket{e}\!\bra{e}_i\otimes A_{\bar{i}})\in\ca_{\rm inv}$ only depends on the symmetry equivalence class of the alignable observable $A$.  Hence, in particular we have 
\ba
   F_{A_{\bar{i}},i}=F_{A_{\bar{j}},j}\q\q\mbox{for all }\,i,j\in\{1,\ldots,N\},
\ea
and so the relational observable in Eq.~(\ref{eqRelationalObservableExpl}) does not depend on the choice of particle $i$.

This systematic equivalence of relational observables  is once more made possible by studying the larger symmetry group $\cg_{\rm sym}$ rather than its subgroup $\cg$ as usual in the literature. Indeed, Theorems~\ref{lem_align1} and~\ref{thm_alignobs} demonstrate that any two particle alignments of an observable are related by a unique symmetry group element which generically lies in $\cg_{\rm sym}\setminus\cg$.

Lemma~\ref{LemProjection} (and its obvious generalization to alignable observables) yields an interesting insight that will become relevant in Section~\ref{Section:Paradox}: if we look at the image of all alignable states and observables under $\Pi_{\rm inv}$, then we do not obtain the full invariant algebra $\mathcal{A}_{\rm inv}$. Instead, we always obtain an operator in the smaller subalgebra
\be
   \mathcal{A}_{\rm alg}:=\left\{A_{\rm phys}+\sum_{\mathbf{h}\in\mathcal{G}^{N-1}} a_{\mathbf{h}} \Pi_{\mathbf{h};\chi\neq\mathbf{1}}\right\},
\ee
where $A_{\rm phys}\in\mathcal{A}_{\rm phys}$, and $a_{\mathbf{h}}\in\C$ are arbitrary complex numbers. According to the definitions in Subsection~\ref{SubsecGroupAveraging}, this gives us the subalgebra inclusions
\be
   \mathcal{A}_{\rm phys}\subset\mathcal{A}_{\rm alg}\subset \mathcal{A}_{\rm inv}\subset \mathcal{A}'_{\rm inv}\subset \mathcal{L}(\mathcal{H}^{\otimes N}).
\ee
Hence, if we denote the orthogonal projection into $\mathcal{A}_{\rm alg}$ by $\Pi_{\rm alg}$, we have $\Pi_{\rm alg}\circ\Pi_{\rm inv}=\Pi_{\rm inv}\circ\Pi_{\rm alg}=\Pi_{\rm alg}$. More specifically, the following holds.
\begin{lemma}
\label{LemAAlg}
$\mathcal{A}_{\rm alg}$ is the smallest subalgebra of $\mathcal{L}(\mathcal{H}^{\otimes N})$ that contains $\Pi_{\rm inv}(|\psi\rangle\langle\psi|)$ for all alignable states $|\psi\rangle$.	
\end{lemma}
\begin{proof}
It is clear that the rank-one projectors $|\psi\rangle\langle\psi|$ for $|\psi\rangle=|e\rangle_1\otimes|\varphi\rangle_{\bar 1}$ linearly span all of $|e\rangle\langle e|_1\otimes \mathcal{L}(\mathcal{H}^{\otimes(N-1)})$, thus we are looking for the subalgebra $\mathcal{A}$ that is generated by the image of these operators under $\Pi_{\rm inv}$. Lemma~\ref{LemProjection} shows that $\Pi_{\rm inv}(|\psi\rangle\langle\psi|)$ is contained in $\mathcal{A}_{\rm alg}$ for every alignable state $|\psi\rangle$, hence $\mathcal{A}\subseteq\mathcal{A}_{\rm alg}$. Conversely, if $A=\sum_{\mathbf{h},\mathbf{j}\in\mathcal{G}^{N-1}}a_{\mathbf{h},\mathbf{j}}|\mathbf{h}\rangle\langle\mathbf{j}|$, then
\begin{eqnarray}
   &&\Pi_{\rm inv}(|e\rangle\langle e|\otimes A)\\&=&\sum_{\mathbf{h},\mathbf{j}\in\mathcal{G}^{N-1}}\frac{a_{\mathbf{h},\mathbf{j}}}{|\mathcal{G}|}|\mathbf{h};\mathbf{1}\rangle\langle\mathbf{j};\mathbf{1}|+\sum_{\mathbf{h}\in\mathcal{G}^{N-1}}\frac{a_{\mathbf{h},\mathbf{h}}}{|\mathcal{G}|}\Pi_{\mathbf{h};\chi\neq\mathbf{1}}.\nn
\end{eqnarray}
Setting $A:=|\mathbf{h}\rangle\langle\mathbf{j}|$ for $\mathbf{h}\neq\mathbf{j}$ shows that $|\mathbf{h};\mathbf{1}\rangle\langle\mathbf{j};\mathbf{1}|\in\mathcal{A}$. But then, we also have $|\mathbf{h};\mathbf{1}\rangle\langle\mathbf{h};\mathbf{1}|=|\mathbf{h};\mathbf{1}\rangle\langle\mathbf{j};\mathbf{1}|)(|\mathbf{j};\mathbf{1}\rangle\langle\mathbf{h};\mathbf{1}|)\in\mathcal{A}$, and so all operators $A_{\rm phys}$ fully supported on $\mathcal{H}_{\rm phys}$ are in $\mathcal{A}$. Finally, considering the image of $A=|\mathbf{h};\mathbf{1}\rangle\langle\mathbf{h};\mathbf{1}|$ shows that $\Pi_{\mathbf{h};\chi\neq\mathbf{1}}\in\mathcal{A}$, and so $\mathcal{A}\supseteq \mathcal{A}_{\rm alg}$.
\end{proof}

\subsection{Communication scenario of the structural approach revisited}
\label{SubsecCommunRevisited}
We are now in a position to revisit the communication scenario elucidating the operational essence of the structural approach to QRFs in section~\ref{sssec_structural} (see also Figure \ref{fig_structural}). It is clear how Alice and Bob can win the game proposed by Refaella in the case of $N$ particles on the configuration space $\cg$ in the absence of a shared external relatum, given that the class of states $\mathcal{C}$ that they are interested in is the set of alignable states. For example, before the game begins, they can agree to always use particle $1$ as the internal reference system relative to which the remaining particles are described. This yields the ``canonical choice'' to represent any equivalence class of (pure) states in the form $|e\rangle_1\otimes|\varphi\rangle_{\bar{1}}$ and Alice's and Bob's return to Refaella will always agree.

Alice and Bob have, of course, the option to choose any of the $N$ particles as a reference system, each likewise defining a ``canonical choice''. It is clear that all these different possible conventions by Alice and Bob are precisely related by the QRF transformations of Lemma~\ref{lem_align1} and that each such transformation is an element of the symmetry group $\cu_{\rm sym}$.

\section{The paradox of the third particle and the relational trace}
\label{Section:Paradox}
Ref.~\cite{Angelo} presents a ``paradox of the third particle'' in the context of QRFs. We will now formulate this apparent paradox in our formalism, and see that the structure of observables, as elaborated in Section~\ref{SecTechnical}, helps to clarify its physical background and to resolve it in terms of relational observables. As we will see, the core of the problem  is how to embed the two-particle observables into the set of three-particle observables, and the key will be to do so in a relational manner. This bears some resemblance to the issue of boundaries and edge modes in gauge theory and gravity \cite{Donnelly:2014fua,Donnelly:2016auv,Geiller:2019bti,Freidel:2020xyx,Gomes:2018shn,Riello:2020zbk,Wieland:2017zkf,Wieland:2017cmf}, which is related to the question of how to embed the gauge-invariant observables of neighbouring subregions in spacetime into the set of gauge-invariant observables associated with the union (`gluing') of these subregions.

 The setup of Ref.~\cite{Angelo} consists of three particles in one dimension, i.e.\ on the real line $\mathbb{R}$. But since only a finite number of positions is relevant for the paradox, we can discretize space and its translations. Thus, we consider a $\mathcal{G}$-system with $\mathcal{G}=\mathbb{Z}_n$, the cyclic group of order $n$, as in Fig.~\ref{fig_cyclic}. The group operation is addition modulo $n$; in the following, whenever we write $a+b$, we actually mean $a+b\mod n$. As described earlier, this is a discrete model of the translation group acting on the real line. We start with two particles that have been prepared (by an external observer with access to the reference frame) in the state
\begin{equation}
   |\psi\rangle=\frac 1 {\sqrt{2}}\left( |-a\rangle_1|b\rangle_2+e^{i\theta}|a\rangle_1|-b\rangle_2\right),
   \label{eqOriginal}
\end{equation}
where $a,b\in\{0,1,2,\ldots,n-1\}$, and $\theta\in\R$. This state is symmetry-equivalent to\footnote{In fact, in the sequel we will assume that $a\neq0$, for otherwise the states in Eqs.~\eqref{eqOriginal} and~\eqref{eqOriginal2} coincide.}
\begin{equation}
   |\psi\rangle\simeq |0\rangle_1\otimes \frac 1 {\sqrt{2}}\left(|a+b\rangle_2+e^ {i\theta}|-a-b\rangle_2\right).
   \label{eqOriginal2}
\end{equation}
In our terminology, this means that the state is alignable. In Ref.~\cite{Angelo}, this form is used as a motivation to declare: \emph{``Therefore, we conclude that particle 1 sees particle 2 in a pure state. Importantly this implies that particle 1 can get access to the phase $\theta$ by interacting with particle 2 alone, i.e.\ without access to the external reference frame.''} In our conceptual framework, we would rather describe the situation as follows: consider an external observer who has access to particles 1 and 2, but has no access to the external reference frame. There are some observables that this observer can measure for which the phase $\theta$ is relevant.

This is because the state that is effectively seen by this observer is the projection of $\psi$ into the invariant subalgebra, which we can determine via Lemma~\ref{LemProjection}. The coefficients of this state are $\alpha_h=1/\sqrt{2}$ for $h=a+b$ and $\alpha_j=e^ {i\theta}/\sqrt{2}$ for $j=-a-b$, thus
\begin{eqnarray}
\Pi_{\rm inv}(|\psi\rangle\langle\psi|)&=& \frac 1 {2n} |h;\mathbf{1}\rangle\langle h;\mathbf{1}|+\frac{e^{-i\theta}}{2n}|h;\mathbf{1}\rangle\langle\ j;\mathbf{1}|\nonumber\\
   &&+\frac{e^{i\theta}}{2n}|j;\mathbf{1}\rangle\langle h;\mathbf{1}|+\frac 1 {2n} |j;\mathbf{1}\rangle\langle j;\mathbf{1}|\nonumber\\
   &&+\frac 1 {2n}\Pi_{h;\chi\neq\mathbf{1}}+\frac 1 {2n}\Pi_{j;\chi\neq\mathbf{1}}
   \label{eqProjection1}
\end{eqnarray}
and this state depends on $\theta$ in a nontrivial way.

Now a third particle is introduced. From the external perspective, it is prepared in a pure state $|c\rangle$ independently of the other two particles, where $c\in\{0,1,\ldots,n-1\}$. From that perspective, the global state thus reads
\begin{equation}
   |\Psi\rangle=\frac 1 {\sqrt{2}}\left( |-a\rangle_1|b\rangle_2+e^{i\theta}|a\rangle_1|-b\rangle_2\right)|c\rangle_3.
   \label{eqThreePInnocuous}
\end{equation}
This state is still alignable. Relative to particle 1, it becomes
\begin{equation}
\small
   |\Psi'\rangle= \frac {|0\rangle_1}{\sqrt{2}}\left( |a+b\rangle_2|a+c\rangle_3 + e^ {i\theta}|-a-b\rangle_2|-a+c\rangle_3\right)
   \label{eqThreeParticleState}
\end{equation}
with $\Psi'\simeq\Psi$. This is a state for which particles 2 and 3 are formally entangled. Now suppose that particle 3 is very far away, such that our external observer (or, as the authors of Ref.~\cite{Angelo} would say, such that particle 1) has no access to particle 3. If one now formally takes the partial trace over particle 3, one obtains a reduced state of particles 1 and 2 that is \emph{independent of the phase $\theta$}. This seems to contradict our earlier claim --- now it looks as if an external observer without access to the external reference frame or to particle 3 \emph{cannot} see any observable consequences of the phase $\theta$.

We arrive at an apparent paradox: computing the partial trace via $\Psi$ or via $\Psi'$ gives different predictions, even though both states are symmetry-equivalent. Moreover, the result from tracing out the third particle via $\Psi'$ seems absurd, given the seemingly innocuous role that the third particle plays in state $\Psi$. \textbf{Can the phase $\theta$ be accessed by an observer without access to the external relatum and \emph{with restricted access to only particles $1$ and $2$}?} It seems like there should be an objective answer to this question which does not depend on whether it is asked in the context of state $\Psi$ or $\Psi'$.

Clearly, since the usual partial trace yields differing results for $\Psi$ and $\Psi'$, it cannot represent the correct rule to compute reduced states in the setting of QRFs. To shed light on the reason for why this is the case, let us reconsider how the standard partial trace can be motivated. Consider three distinguishable particles as one usually does in quantum information theory --- in our setting, this implies that the state of the particles is defined relative to an accessible reference frame. Denote by $A_{12}$ some observable that is measurable if one has access to particles 1 and 2 only (and to the reference frame). Then one can equivalently describe this as an observable on the three particles, such that the third particle is ignored. Formally, this can be done via a map
\ba
   A_{12}\mapsto \Phi(A_{12})=A_{12}\otimes\mathbf{1}_3.\label{em1}
\ea
That is, the 2-particle observables are naturally embedded into the 3-particle observables via some map $\Phi$, which takes the tensor product with the identity observable. This map preserves all relevant structure (as it must): it takes linear combinations to linear combinations, products to products, the adjoint to the adjoint, and the identity to the identity. Formally, this is summarized by saying that $\Phi$ is a \emph{unital $^*$-homomorphism}. It defines what we mean when we talk about ``observables pertaining only to particles $1$ and $2$'' within the set of all $3$-particle observables. In the following, we will refer to  $^*$-homomorphisms simply as \emph{embeddings} (which may or may not be unital).

Now consider any quantum state $\rho_{123}$ on the three particles. We would like to determine the state $\rho_{12}$ that results if one has only access to particles 1 and 2. By this, we mean the state that gives the same expectation values as $\rho_{123}$ on all local observables $A_{12}$. Thus, we demand
\be
   \tr(\rho_{12}A_{12})=\tr(\rho_{123}\Phi(A_{12}))\quad\mbox{for all }A_{12}.
\ee
If we write this in terms of the Hilbert-Schmidt inner product, then
\[
   \langle \rho_{12},A_{12}\rangle_{\rm HS}=\langle\rho_{123},\Phi(A_{12})\rangle_{\rm HS}=\langle\Phi^ \dagger(\rho_{123}),A_{12}\rangle_{\rm HS}.
\]
That is, we must have $\rho_{12}=\Phi^\dagger(\rho_{123})$, with $\Phi^ \dagger$ the Hilbert-Schmidt adjoint of $\Phi$. But given Eq.~\eqref{em1}, it is easy to see that $\Phi^\dagger={\rm Tr}_3$, and this recovers the partial trace.

In the context of QRFs and $\mathcal{G}$-systems, we have a different structure of observables: what is measurable without access to the external reference frame corresponds to the \emph{invariant observables}. We therefore need an analog of the above construction for the invariant subalgebra $\mathcal{A}_{\rm inv}$. In fact, since we are only interested in alignable states, Lemma~\ref{LemAAlg} tells us that we can focus on the subalgebra $\mathcal{A}_{\rm alg}$.

\subsection{The non-uniqueness of invariant embeddings}
\label{SubsecNonUniqueEmbd}
In light of the above considerations, let us try to construct an embedding of the relevant 2-particle observables (or more generally of $\mathcal{A}_{\rm alg}^{(N)}$) into the 3-particle observables (more generally into $\mathcal{A}_{\rm alg}^{(N+M)}$, where the superscript denotes the number of particles). To obtain some crucial physical intuition, we will now define one such embedding (called $\tilde\Phi^{(1)}$) in an intuitive manner, before we turn to a more systematic treatment below.

We start our construction with the orthogonal projector $\Pi_{\mathbf{h}}^{(N)}:=\sum_{\chi\in\mathcal{\hat G}}|\mathbf{h};\chi\rangle\langle\mathbf{h};\chi|$. Its embedding must be on orthogonal projector in $\mathcal{A}_{\rm alg}^{(N+M)}$; measuring this projector amounts to asking whether the first $N$ particles have pairwise relations described by $\mathbf{h}$. Clearly, the answer must be ``yes'' whenever the $M+N$ particles are in some joint relation $(\mathbf{h},\mathbf{g})$ for arbitrary $\mathbf{g}\in\mathcal{G}^M$, and the answer must be ``no'' for all other relations $(\mathbf{h}',\mathbf{g})$ whenever $\mathbf{h}'\neq\mathbf{h}$. This suggests to choose the embedding as
\begin{equation}
   \tilde\Phi^{(1)}(\Pi_{\mathbf{h}}^{(N)})=\sum_{\mathbf{g}\in\mathcal{G}^M} \Pi_{\mathbf{h},\mathbf{g}}^{(N+M)}.
   \label{eqPhi0}
\end{equation}
The reason for the superscript `$(1)$' will become clear shortly. Now consider the orthogonal projector $\Pi_{\rm phys}^{(N)}:=\sum_{\mathbf{h}\in\mathcal{G}^{N-1}}|\mathbf{h};\mathbf{1}\rangle\langle\mathbf{h};\mathbf{1}|$. A state $|\psi\rangle$ is in the image of this projector if and only if it is translation-invariant, i.e.\ $U_g^{\otimes N}|\psi\rangle=|\psi\rangle$ for all $g\in\mathcal{G}$. Suppose that a state of $N+M$ particles is translation-invariant; then we would like to be able to say that the state of the first $N$ particles is translation-invariant, too. This motivates us to demand $\tilde\Phi^{(1)}(\Pi_{\rm phys}^{(N)})=\Pi_{\rm phys}^{(N+M)}$: in other words, we embed the translation-invariant observables of $N$ particles into the translation-invariant observables of $N+M$ particles.

We have $|\mathbf{h};\mathbf{1}\rangle\langle\mathbf{h};\mathbf{1}|=\Pi_{\mathbf{h}}^{(N)}\Pi_{\rm phys}^{(N)}$. Since $\tilde\Phi^{(1)}$ is supposed to be an embedding and thus multiplicative, this implies
\begin{equation}
   \tilde\Phi^{(1)}(|\mathbf{h};\mathbf{1}\rangle\langle\mathbf{h};\mathbf{1}|)=\sum_{\mathbf{g}\in\mathcal{G}^M} |\mathbf{h},\mathbf{g};\mathbf{1}\rangle\langle\mathbf{h};\mathbf{g};\mathbf{1}|.
   \label{eqPhi1}
\end{equation}
To exhaust all of $\mathcal{A}_{\rm alg}^{(N)}$, we still need to embed the operators $|\mathbf{h};\mathbf{1}\rangle\langle\mathbf{j};\mathbf{1}|$ for $\mathbf{h}\neq\mathbf{j}$. Given Eq.~(\ref{eqPhi1}), it seems formally natural to define
\begin{equation}
   \tilde\Phi^{(1)}(|\mathbf{h};\mathbf{1}\rangle\langle\mathbf{j};\mathbf{1}|):=\sum_{\mathbf{g}\in\mathcal{G}^M} |\mathbf{h},\mathbf{g};\mathbf{1}\rangle\langle\mathbf{j},\mathbf{g};\mathbf{1}|.
   \label{eqPhi2}
\end{equation}
Natural as this definition may seem, we will soon see that it relies on quite subtle physical assumptions. For now, let us work with this definition and explore its consequences. First, linearity and Eqs.~(\ref{eqPhi0}) and~(\ref{eqPhi1}) imply
\begin{equation}
   \tilde\Phi^{(1)}\left(\Pi_{\mathbf{h},\chi\neq\mathbf{1}}\right)=\sum_{\mathbf{g}\in\mathcal{G}^M} \Pi_{\mathbf{h},\mathbf{g};\chi\neq\mathbf{1}}.
   \label{eqPhi3}
\end{equation}
These demands yield an embedding that can equivalently be defined as follows:
\begin{lemma}
\label{LemAlternative}
There is a unique unital embedding $\tilde\Phi^{(1)}$ of $\mathcal{A}_{\rm alg}^{(N)}$ into $\mathcal{A}_{\rm alg}^{(N+M)}$ that satisfies
\begin{equation}
   \tilde\Phi^{(1)}\left(\Pi_{\rm alg}^{(N)}\left(\hat e_1\otimes A_{\bar 1}\right)\right)=\Pi_{\rm alg}^{(N+M)}\left(\hat e_1\otimes A_{\bar 1}\otimes \mathbf{1}^{(M)}\right)
   \label{eqEmbAltern}
\end{equation}
for all $A_{\bar 1}\in\mathcal{L}(\mathcal{H}^{\otimes (N-1)})$, where $\hat e_1:=|e\rangle\langle e|_1$ and $\bar 1:=\{2,3,\ldots,N\}$. It is given by the linear extension of Eqs.~(\ref{eqPhi2}) and~(\ref{eqPhi3}).
\end{lemma}
\begin{proof}
Write $A_{\bar 1}=\sum_{\mathbf{h},\mathbf{j}\in\mathcal{G}^{N-1}} a_{\mathbf{h},\mathbf{j}}|\mathbf{h}\rangle\langle\mathbf{j}|$ and use Lemma~\ref{LemProjection} to obtain
\begin{eqnarray}
\Pi_{\rm alg}^{(N)}\left(|e\rangle\langle e|_1\otimes A_{\bar 1}\right)&=& \sum_{\mathbf{h},\mathbf{j}\in\mathcal{G}^{N-1}} \frac{a_{\mathbf{h},\mathbf{j}}}{|\mathcal{G}|} |\mathbf{h};\mathbf{1}\rangle\langle\mathbf{j};\mathbf{1}|\nonumber\\
&& + \sum_{\mathbf{h}\in\mathcal{G}^{N-1}} \frac{a_{\mathbf{h},\mathbf{h}}}{|\mathcal{G}|} \Pi_{\mathbf{h};\chi\neq\mathbf{1}}.
\label{eqProjProof}
\end{eqnarray}
A similar representation can be obtained for the right-hand side of Eq.~(\ref{eqEmbAltern}). Thus, it is clear that the linear extension of Eqs.~(\ref{eqPhi2}) and~(\ref{eqPhi3}) satisfies Eq.~(\ref{eqEmbAltern}). In particular, the unit is preserved.

Now let $\tilde\Phi$ be any embedding of $\mathcal{A}_{\rm alg}^{(N)}$ into $\mathcal{A}_{\rm alg}^{(N+M)}$ which satisfies Eq.~(\ref{eqEmbAltern}). Choose $\mathbf{h},\mathbf{j}\in\mathcal{G}^{N-1}$ arbitrary, and let $A_{\bar 1}$ be the operator with $a_{\mathbf{h},\mathbf{j}}=|\mathcal{G}|$ and all other coefficients zero. Then Eq.~(\ref{eqProjProof}) becomes $|\mathbf{h};\mathbf{1}\rangle\langle\mathbf{j};\mathbf{1}|$, and $\Pi_{\rm alg}^{(N+M)}\left(|e\rangle\langle e|_1\otimes A_{\bar 1}\otimes\mathbf{1}^{(M)}\right)$ becomes the right-hand side of Eq.~(\ref{eqPhi2}). Now, still for $\mathbf{h}\neq\mathbf{j}$,
\be
   \left(|\mathbf{h};\mathbf{1}\rangle\langle\mathbf{j};\mathbf{1}|\right)^\dagger \left(|\mathbf{h};\mathbf{1}\rangle\langle\mathbf{j};\mathbf{1}|\right)=|\mathbf{j};\mathbf{1}\rangle\langle\mathbf{j};\mathbf{1}|.
\ee
Since $\tilde\Phi$ preserves adjoints and products, this proves Eq.~(\ref{eqPhi2}) also in the case $\mathbf{h}=\mathbf{j}$. Finally, choosing $A_{\bar 1}$ such that $a_{\mathbf{h},\mathbf{h}}=|\mathcal{G}|$ and all other coefficients zero proves Eq.~(\ref{eqPhi3}), which implies that $\tilde\Phi$ is the linear extension of Eqs.~(\ref{eqPhi2}) and~(\ref{eqPhi3}).
\end{proof}
What this lemma demonstrates is that our construction of $\tilde\Phi^{(1)}$ can be interpreted in an alternative way --- at least on ``alignable'' observables. What our embedding map does to those observables is as follows: \emph{write them in the form $|e\rangle\langle e|_1\otimes A_{\bar 1}$, and embed them into the total Hilbert space according to its defining tensor product structure. Then demand that $\tilde\Phi^{(1)}$ maps the invariant part of $|e\rangle\langle e|_1\otimes A_{\bar 1}$ to the invariant part of its embedding $|e\rangle\langle e|_1\otimes A_{\bar 1}\otimes\mathbf{1}^{(M)}$}. This defines a particular, natural embedding of $\mathcal{A}_{\rm alg}^{(N)}$ into $\mathcal{A}_{\rm alg}^{(N+M)}$.

But this suggests directly that our candidate relational trace is defective: \emph{its definition is implicitly based on the choice of particle $1$ as our reference.} Indeed, the following lemma shows that there is a large \emph{class of invariant embeddings}. In particular, choosing another particle as the reference particle will in general lead to inequivalent embeddings. The proof is given by a straightforward calculation and thus omitted.
\begin{lemma}
\label{LemPhiU}
Let $U\in\mathcal{U}_{\rm sym}$ be any symmetry transformation (for example, a QRF transformation). Then there is a unique unital embedding $\tilde\Phi^U$ of $\mathcal{A}_{\rm alg}^{(N)}$ into $\mathcal{A}_{\rm alg}^{(N+M)}$ which satisfies
\be
   \small\tilde\Phi^U\left(\Pi_{\rm alg}^{(N)}(\hat e_1\otimes A_{\bar 1})\right)=\Pi_{\rm alg}^{(N+M)}\left(U(\hat e_1\otimes A_{\bar 1})U^\dagger\otimes\mathbf{1}^{(M)}\right).
\ee
Writing $U=\bigoplus_{\mathbf{h}\in\mathcal{G}^{N-1}} U_{g(\mathbf{h})}^{\otimes N}$, it acts as
\ba
&&\tilde\Phi^U\left(|\mathbf{h};\mathbf{1}\rangle\langle\mathbf{j};\mathbf{1}|\right)\nn\\
&&\q\q=\sum_{\mathbf{g}\in\mathcal{G}^M}\!\!|\mathbf{h},g(\mathbf{h})^{-1}\mathbf{g};\mathbf{1}\rangle\langle \mathbf{j},g(\mathbf{j})^{-1}\mathbf{g};\mathbf{1}|,\label{eqEmbedOther}\\
&&\tilde\Phi^U\left(\Pi_{\mathbf{h};\chi\neq\mathbf{1}}\right)=\sum_{\mathbf{g}\in\mathcal{G}^N} \Pi_{\mathbf{h},\mathbf{g};\chi\neq\mathbf{1}}.
\ea
Via $\Phi^U:=\tilde\Phi^U\circ\Pi_{\rm alg}^{(N)}$, this extends to a completely positive unital map $\Phi^U:\mathcal{L}(\mathcal{H}^{\otimes N})\to\mathcal{L}(\mathcal{H}^{\otimes(N+M)})$.
\end{lemma}
Since $\Pi_{\rm alg}^{(N)}=\Pi_{\rm alg}^{(N)}\circ \Pi_{\rm inv}^{(N)}$, we can write the left-hand side of the first equation above as
\[
   \tilde\Phi^U\left(\Pi_{\rm alg}^{(N)}(\hat e_1\otimes A_{\bar 1})\right)=\tilde\Phi^U\left(\Pi_{\rm alg}^{(N)}(U(\hat e_1\otimes A_{\bar 1})U^\dagger)\right).
\]
In particular, if we choose $U$ as the QRF transformation of Theorem~\ref{lem_align1} that ``changes from the perspective of particle $1$ to particle $i$'', we obtain a natural invariant embedding $\tilde\Phi^{(i)}$ ``relative to particle $i$''. It satisfies
\begin{equation}
   \tilde\Phi^{(i)}\left(\Pi_{\rm alg}^{(N)}(\hat e_i\otimes A_{\bar i})\right)=\Pi_{\rm alg}^{(N+M)}\left(\hat e_i\otimes A_{\bar i}\otimes\mathbf{1}^{(M)}\right),
   \label{eqNonInvariantEmbedding}
\end{equation}
and acts on the basis elements of $\mathcal{A}_{\rm phys}$ as
\[
\tilde\Phi^{(i)}\left(|\mathbf{h};\mathbf{1}\rangle\langle\mathbf{j};\mathbf{1}|\right)=\sum_{\mathbf{g}\in\mathcal{G}^M}|\mathbf{h},\mathbf{g};\mathbf{1}\rangle\langle \mathbf{j},h_{i-1}^{-1}j_{i-1}\mathbf{g};\mathbf{1}|.
\]
For a better understanding of the \emph{physical} reason of this non-uniqueness of embedding, let us reconsider our intuitive construction of $\tilde\Phi^{(1)}$ above. First, note that all the $\tilde\Phi^U$ satisfy Eq.~(\ref{eqPhi1}), which had a clear physical motivation. However, the $\Phi^{(i)}$ violate Eq.~(\ref{eqPhi2}), which we had motivated purely by formal analogy. To shed light on Eq.~\eqref{eqPhi2} and its generalization, Eq.~(\ref{eqEmbedOther}), suppose for concreteness that we are interested in embedding $N=2$ particles into $N+M=3$ particles, and our group $\mathcal{G}$ is the cyclic group $\mathbb{Z}_n$ with addition modulo $n$. Consider the orthogonal projector $|\psi\rangle\langle\psi|$, where
\begin{equation}
   |\psi\rangle=\frac 1 {\sqrt{2}}\left(|1;\mathbf{1}\rangle+|2;\mathbf{1}\rangle\right).
   \label{eqPsiFigure}
\end{equation}
This is an element of $\mathcal{H}^{(2)}_{\rm phys}$, the $2$-particle subspace of invariant states, and it describes a superposition of two particles either being one or two places apart, see Figure~\ref{fig_superposition}. Note that there is no origin that would locate the particles absolutely; all we have is their relations.

\begin{figure}[hbt]
\begin{center}
\includegraphics[width=.2\textwidth]{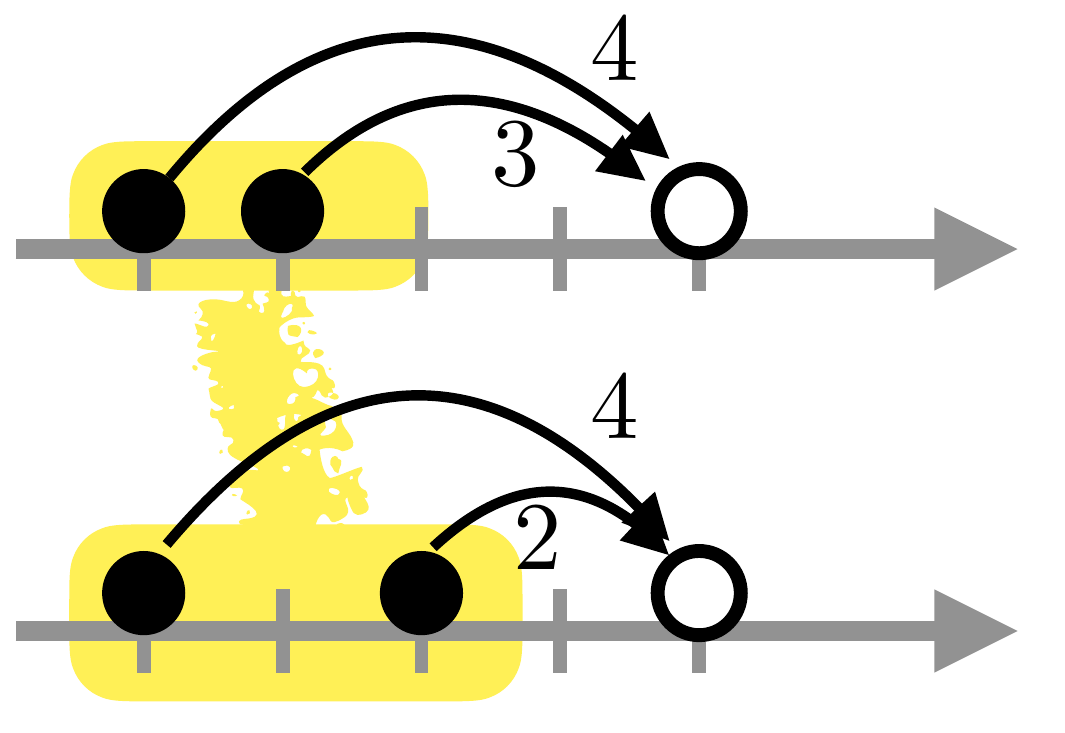}
\caption{The state $|\psi\rangle$ in~(\ref{eqPsiFigure}) describes a superposition (indicated in yellow) of two particles (black dots) being either one or two places apart. If there is a third particle (white dot) such that the resulting state is still identical to the pure state $|\psi\rangle$ on the first two particles, then that particle should carry no information as to ``which branch'' is actualized. That is, its relation to the first two particles should be the same in both branches. But whether this is the case depends on our convention of how we \emph{define} its relation to the first two particles. This results in different embedding maps $\Phi^U$.}
\label{fig_superposition}
\end{center}
\end{figure}

Now suppose we would like to embed the corresponding observable $|\psi\rangle\langle\psi|$ into the three-particle observables. The essence of the problem lies in embedding $|h;\mathbf{1}\rangle\langle j;\mathbf{1}|$, where $h=1$ and $j=2$. The local state of the two particles will remain coherent if the third particle carries no information on whether configuration $h$ or $j$ is actualized --- that is, if the properties of the third particle are the same in both branches of the superposition. But the \emph{only} properties of the third particle are its \emph{relations} to the other two particles. Thus, if we complete $|\psi\rangle$ to a state on three particles
\begin{equation}
   |\Psi\rangle=\frac 1 {\sqrt{2}}\left(|h,g;\mathbf{1}\rangle+|j,g';\mathbf{1}\rangle\right),
   \label{eqAssignOne}
\end{equation}
then the local state of the first two particles will remain coherent if and only if the third particle inside the configuration $(h,g)$ has \emph{the same relation to the first two particles} as the third particle in the configuration $(j,g')$. But the crucial insight is that this will \emph{depend on what we mean by ``relation to the first two particles''}.

For instance, suppose that $g=4$. If we choose the convention to say that the relation is identical if \emph{the relation to the first particle is the same}, then this will be the case if $g'=g=4$. But if we demand instead that the \emph{relation to the second particle is the same}, then we need $g'=5$. Our definition of $\Phi^{(1)}$ is implicitly relying on the former convention, while $\Phi^{(2)}$ would rely on the latter. The reason why Eq.~(\ref{eqPhi2}) (for $\Phi=\Phi^{(1)}$) looks so simple is that we have implicitly labelled the relations $\mathbf{h}\in\mathcal{G}^{N-1}$ as \emph{relative to the first particle} in all of this work, recall Eq.~(\ref{eqRelationFirstParticle}). This is no loss of generality, and it had no implications whatsoever for Section~\ref{SecTechnical}, but here it becomes relevant. The factor of $h_{i-1}^{-1} j_{i-1}$ in Eq.~(\ref{eqEmbedOther}) adapts the convention. Note that it does not alter the pairwise relations among the first $N$ particles, or the pairwise relations among the last $M$ particles, but only the relation between the two groups of particles.

Is there a way to escape the non-uniqueness of embeddings via some formal construction that is manifestly relational, but does not depend on a choice of reference within the $N$-particle subsystem relative to which the relation of the new $M$ particles is defined? The right-hand side of Eq.~(\ref{eqNonInvariantEmbedding}) shows that the maps $\tilde\Phi^{(i)}$ embed the \emph{invariant} operator $\Pi_{\rm alg}^{(N)}(\hat e_i\otimes A_{\bar i})$ by embedding the original, \emph{non-invariant} operator $\hat e_i\otimes A_{\bar i}$ into the total Hilbert space, followed by the projection into the global  subalgebra $\mathcal{A}_{\rm alg}$. While the resulting map $\tilde\Phi^{(i)}$ is invariant, its \emph{definition} is therefore not. Can we perhaps make the definition invariant by embedding not $\hat e_i\otimes A_{\bar i}$ directly, but its invariant part? The following lemma answers this question in the negative:
\begin{lemma}
\label{LemEmbeddingHomo}
Define the map $\tilde\Phi:\mathcal{A}_{\rm alg}^{(N)}\to\mathcal{A}_{\rm alg}^{(N+M)}$ as
\be
   \tilde\Phi(A^{(N)}):=\Pi_{\rm alg}^{(N+M)}\left(A^{(N)}\otimes\mathbf{1}^{(M)}\right).
\ee
Then this map is not a valid embedding. Namely, it is not in general multiplicative, i.e.\ there exist $A,B\in\mathcal{A}_{\rm alg}^{(N)}$ with $\tilde\Phi(AB)\neq\tilde\Phi(A)\tilde\Phi(B)$.
\end{lemma}
\begin{proof}
A tedious but straightforward calculation yields
\begin{eqnarray}
\tilde\Phi\left(|\mathbf{h};\mathbf{1}\rangle\langle\mathbf{j};\mathbf{1}|\right)&=&\frac 1 {|\mathcal{G}|} \sum_{g\in\mathcal{G}}\sum_{\mathbf{g}\in\mathcal{G}^M}|\mathbf{h},\mathbf{g};\mathbf{1}\rangle\langle\mathbf{j},g\mathbf{g};\mathbf{1}|\nn\\
&&+\frac 1 {|\mathcal{G}|} \delta_{\mathbf{h},\mathbf{j}}\sum_{\mathbf{g}\in\mathcal{G}^M}\Pi_{\mathbf{h},\mathbf{g};\chi\neq\mathbf{1}}.
\end{eqnarray}
But then, for $\mathbf{h}\neq\mathbf{j}$, we obtain
\[
   \tilde\Phi(|\mathbf{h};\mathbf{1}\rangle\langle\mathbf{j};\mathbf{1}|) \tilde\Phi(|\mathbf{j};\mathbf{1}\rangle\langle\mathbf{h};\mathbf{1}|)\neq \tilde\Phi(|\mathbf{h};\mathbf{1}\rangle\langle\mathbf{h};\mathbf{1}|).
\]
\vskip -2em
\end{proof}
A similar argument applies if we try to embed $\mathcal{A}_{\rm inv}^{(N)}$ into $\mathcal{A}_{\rm inv}^{(N+M)}$: the analog of the above construction, with $\Pi_{\rm alg}^{(N+M)}$ replaced by $\Pi_{\rm inv}^{(N+M)}$, does not yield a valid embedding.

\subsection{A class of invariant traces}
\label{SubsecClassTraces}
For every embedding $\Phi^U$ of Lemma~\ref{LemPhiU}, we obtain a corresponding ``invariant trace'':
\begin{lemma}
\label{LemTrinvFinal}
For every $U=\bigoplus_{\mathbf{h}\in\mathcal{G}^{N-1}} U_{g(\mathbf{h})}^{\otimes N}\in\mathcal{U}_{\rm sym}$, define a corresponding ``invariant trace'' ${\rm Trinv}^U:=\left(\Phi^U\right)^\dagger$. It is trace-preserving, and it maps invariant operators to invariant operators. In particular, ${\rm Trinv}^U_{(M)}\left(\mathcal{A}_{\rm alg}^{(N+M)}\right)=\mathcal{A}_{\rm alg}^{(N)}$, and can be explicitly written in the following form:
\begin{widetext}
\[
{\rm Trinv}^U_{(M)}\,\rho= \sum_{\mathbf{h},\mathbf{j}\in\mathcal{G}^{N-1}}|\mathbf{h};\mathbf{1}\rangle\langle\mathbf{j};\mathbf{1}| \sum_{\mathbf{g}\in\mathcal{G}^M} \langle\mathbf{h},g(\mathbf{h})^{-1}\mathbf{g};\mathbf{1}|\rho|\mathbf{j},g(\mathbf{j})^{-1}\mathbf{g};\mathbf{1}\rangle + \sum_{\mathbf{h}\in\mathcal{G}^{N-1}} \Pi_{\mathbf{h};\chi\neq\mathbf{1}}\sum_{\mathbf{g}\in\mathcal{G}^M}\frac{{\rm tr}(\Pi_{\mathbf{h},\mathbf{g};\chi\neq\mathbf{1}}\rho)}{|\mathcal{G}|-1}.
\]
\end{widetext}
\end{lemma}
We omit the straightforward proof. 

Since the different invariant traces formalize different ways to embed the two-particle observables into the three-particle observables, it is clear that the answer to the question raised at the beginning of this section will depend on the embedding. In other words, the question of whether the phase $\theta$ is accessible on the first two particles depends on the operational definition of how to access the first two particles within the \emph{total} three-particle Hilbert space.

To illustrate this fact, let us apply two different invariant traces to the paradox. In contrast to the usual partial trace, every invariant trace yields identical results when applied to the equivalent states $\Psi$ and $\Psi'$ in Eqs.~(\ref{eqThreePInnocuous}) and~(\ref{eqThreeParticleState}) (hence the name ``invariant''). That is,
\be
   {\rm Trinv}_3^U |\Psi\rangle\langle\Psi|={\rm Trinv}_3^U |\Psi'\rangle\langle\Psi'|
\ee
for all $U\in\mathcal{U}_{\rm sym}$. However, different $U$ yield different results. For example, consider the invariant trace ${\rm Trinv}^{\mathbf{1}}$ which is the adjoint of the embedding $\Phi^{(1)}$ that we have constructed in Subsection~\ref{SubsecNonUniqueEmbd}. A straightforward calculation gives
\begin{eqnarray*}
	   {\rm Trinv}^{\mathbf{1}}_3|\Psi\rangle\langle\Psi|&=&\frac 1 {2n}|h;\mathbf{1}\rangle\langle h;\mathbf{1}|+\frac 1 {2n}|j;\mathbf{1}\rangle\langle j;\mathbf{1}|\\
	   &&+ \frac 1 {2n}\Pi_{h;\chi\neq\mathbf{1}}+\frac 1 {2n} \Pi_{j;\chi\neq\mathbf{1}}\\
	   &=&\Pi_{\rm inv}^{(2)}\left[ |e\rangle\langle e|_1\otimes\left(\frac 1 2 |h\rangle\langle h|+\frac 1 2|j\rangle\langle j|\right)_2\right],
\end{eqnarray*}
Up to observational equivalence, we hence get a mixed state
\[
   {\rm Trinv}^{\mathbf{1}}_3|\Psi\rangle\langle\Psi|\sim |e\rangle\langle e|_1\otimes \left(\frac 1 2 |h\rangle\langle h|+\frac 1 2|j\rangle\langle j|\right)_2.
\]
In particular, the phase $\theta$ has disappeared. This is not surprising: ultimately, ${\rm Trinv^{\mathbf{1}}}$ amounts to taking the usual partial trace in the representation of the state relative to particle $1$, i.e.\ of the state $|\Psi'\rangle$ of Eq.~(\ref{eqThreeParticleState}).

On the other hand, consider the symmetry transformation $U$ from Example~\ref{ExCenterMass} which transforms to the center of mass. Like in Ref.~\cite{Angelo}, let us choose $m_1$ and $m_2$ such that $m_1a = m_2 b$. Using Lemma~\ref{LemTrinvFinal}, It is clear that
\be
   {\rm Trinv}_3^{U}\left(\Pi_{\mathbf{h};\chi\neq\mathbf{1}}\right)=\Pi_{h_1;\chi\neq\mathbf{1}}
\ee
and similarly for $\mathbf{h}$ replaced by $\mathbf{j}$. Furthermore,
\begin{widetext}
\begin{eqnarray}
   {\rm Trinv}_3^{U}(|\mathbf{l};\mathbf{1}\rangle\langle\mathbf{p};\mathbf{1}|)&=&\sum_{h,j\in\mathbb{Z}_n}|h;\mathbf{1}\rangle\langle j;\mathbf{1}|\sum_{g\in\mathbb{Z}_n} \langle h,g+\left\lfloor\frac{m_2}m h\right\rfloor;\mathbf{1}|\mathbf{l};\mathbf{1}\rangle\langle\mathbf{p};\mathbf{1}|j,g+\left\lfloor\frac{m_2}m j\right\rfloor;\mathbf{1}\rangle\nn\\
   &=&|l_1;\mathbf{1}\rangle\langle p_1;\mathbf{1}|\delta_{l_2-\lfloor\frac{m_2}m l_1\rfloor,p_2-\lfloor \frac{m_2}m p_1\rfloor}
\end{eqnarray}
\end{widetext}
for all $\mathbf{l},\mathbf{p}\in\mathbb{Z}_n^2$. Since ${\rm Trinv}_3^U|\Psi\rangle\langle\Psi|={\rm Trinv}_3^U\left(\Pi_{\rm alg}(|\Psi\rangle\langle\Psi|)\right)$, we can expand $\Pi_{\rm alg}(|\Psi\rangle\langle\Psi|)$ into basis elements and apply the above equations. As a result, this yields exactly Eq.~(\ref{eqProjection1}) (since $a$ is an integer, it turns out that we can ignore all $\lfloor\cdot\rfloor$). That is,
\be
   {\rm Trinv}_3^U\left(\strut |\Psi\rangle\langle\Psi|\right)=\Pi_{\rm inv}\left(\strut |\psi\rangle\langle\psi|\right)\sim |\psi\rangle\langle\psi|.
\ee
That is, up to observational equivalence, we obtain the original pure alignable state $|\psi\rangle$ of Eq.~(\ref{eqOriginal}). Moreover, in this observational equivalence class, $|\psi\rangle$ is \emph{unique up to symmetry equivalence}: namely, if there is another pure alignable state $|\psi'\rangle$ with ${\rm Trinv}_3^U(|\Psi\rangle\langle\Psi|)\sim|\psi'\rangle\langle\psi'|$, then $\psi\sim\psi'$. But due to Lemma~\ref{LemProjection}, this implies that $\psi\simeq\psi'$. Thus, in particular, there is (up to a global phase) a unique representation of this state relative to the $i$-th particle, $|\psi\rangle\simeq |e\rangle_i\otimes |\varphi\rangle_{\bar i}$. According to Eq.~(\ref{eqOriginal2}), it is
\be
   |\varphi\rangle_{\bar 1}=\frac 1 {\sqrt{2}}\left(\strut |a+b\rangle_2+e^{i\theta}|-a-b\rangle_2\right)
\ee
relative to the first particle. Thus, under the ``center of mass'' relational trace ${\rm Trinv}^U$, the phase $\theta$ survives, in contrast to the result for ${\rm Trinv}^{\mathbf{1}}$.

\subsection{Definition of the relational trace}

Recalling the invariant algebra inclusions $\ca_{\rm phys}\subset\ca_{\rm alg}\subset\ca_{\rm inv}$, we have thus far focused on constructing an invariant trace for $\ca_{\rm alg}$, since by Lemma~\ref{LemAAlg}, it is the smallest algebra containing the invariant part of all alignable states and observables. However, in the previous subsection, we have seen that there does not exist a unique invariant trace on it, and the same conclusion applies to $\mathcal{A}_{\rm inv}$. We will now show that, by contrast, there \emph{does} exists a natural embedding of the algebra $\ca_{\rm phys}^{(N)}$ of relational $N$-particle operators into the  algebra $\ca_{\rm phys}^{(N+M)}$ of relational $(N+M)$-particle operators. This will also lead to a natural definition of an invariant trace in terms of relational states.

This trace, which we hence call the \emph{relational trace}, has a natural and simple definition that is manifestly invariant under relative translations between the two particle groups. It is therefore independent of the various physically distinct conventions discussed in the previous subsections. On the relational subspace $\ch_{\rm phys}$, the paradox is therefore unambiguously resolved. It is sufficient to focus on this subspace because Lemma~\ref{LemProjection} tells us that the relational observables are tomographically complete for the invariant information in all alignable states. Thus, $\ch_{\rm phys}$ contains all the information we are interested in.
\begin{lemma}\label{lem_nathomo}
The map $\tilde\Phi_{\rm phys}:\ca_{\rm phys}^{(N)}\rightarrow\ca_{\rm phys}^{(N+M)}$, defined by
\ba
\tilde\Phi_{\rm phys}(A_{\rm phys}^{(N)}):=\hat\Pi_{\rm phys}^{(N+M)}\left(A_{\rm phys}^{(N)}\otimes \mathbf{1}^{(M)}\right)
\ea
is an embedding. It simplifies to
\ba
\tilde\Phi_{\rm phys}(A_{\rm phys}^{(N)})=A_{\rm phys}^{(N)}\otimes \Pi_{\rm phys}^{(M)},
\ea
but it is \emph{not} unital.
\end{lemma}
Recall that this construction does not yield an embedding of all invariant observables $\mathcal{A}^{(N)}_{\rm inv}$, or even of $\mathcal{A}_{\rm alg}^{(N)}$, as we have seen in Lemma~\ref{LemEmbeddingHomo}. Thus, it is remarkable that it works for the subalgebra of relational observables.

\begin{proof}
Noting that every $A_{\rm phys}^{(N)}\in\ca_{\rm phys}^{(N)}$ satisfies $A_{\rm phys}^{(N)}=\hat\Pi_{\rm phys}^{(N)}(A_{\rm phys}^{(N)})$, we can recast $\tilde\Phi_{\rm phys}$ in the form
\ba
\tilde\Phi_{\rm phys}(A_{\rm phys}^{(N)})&=&\Pi_{\rm phys}^{(N+M)}\left(\Pi_{\rm phys}^{(N)}\otimes\mathbf{1}^{(M)}\right)\left(A_{\rm phys}^{(N)}\otimes\mathbf{1}^{(M)}\right)\nn\\
&&\q\q\q\q\q\times\left(\Pi_{\rm phys}^{(N)}\otimes\mathbf{1}^{(M)}\right)\Pi_{\rm phys}^{(N+M)}.\nn
\ea
By considering the action on an arbitrary basis state $\ket{g_1,\ldots,g_{N+M}}$, it is easy to verify that
\ba
\left(\Pi_{\rm phys}^{(N)}\otimes\mathbf{1}^{(M)}\right)\Pi_{\rm phys}^{(N+M)}=\Pi_{\rm phys}^{(N)}\otimes\Pi_{\rm phys}^{(M)}.\label{projid}
\ea
Hence, using idempotence of the projector,
\ba
\tilde\Phi_{\rm phys}(A_{\rm phys}^{(N)})=A_{\rm phys}^{(N)}\otimes \Pi_{\rm phys}^{(M)},
\ea
and the image commutes with $\Pi_{\rm phys}^{(N+M)}$.
Checking the embedding properties is now trivial. Furthermore, note that the unit element of $\ca_{\rm phys}^{(N)}$ is $\Pi_{\rm phys}^{(N)}$, while the unit element of $\ca_{\rm phys}^{(N+M)}$ is $\Pi_{\rm phys}^{(N+M)}$. However, $\tilde\Phi_{\rm phys}(\Pi_{\rm phys}^{(N)})=\Pi_{\rm phys}^{(N)}\otimes\Pi_{\rm phys}^{(M)}$. To see that $\tilde\Phi_{\rm phys}(\Pi_{\rm phys}^{(N)})$ does not act as the identity on $\ch_{\rm phys}^{(N+M)}$, observe that ${\rm tr}\left(\Pi_{\rm phys}^{(N)}\right)=\dim\mathcal{H}_{\rm phys}^{(N)}=|\mathcal{G}|^{N-1}$, hence ${\rm tr}\left(\Pi_{\rm phys}^{(N)}\otimes\Pi_{\rm phys}^{(M)}\right)=|\mathcal{G}|^{N+M-2}$, but ${\rm tr}\left(\Pi_{\rm phys}^{(N+M)}\right)=|\mathcal{G}|^{N+M-1}$.
\end{proof}
Returning to the discussion of how to define the relations of the additional $M$ particles to the first group of $N$ particles (cf.\ Fig.~\ref{fig_superposition}), we now have a unique answer: the projector $\Pi_{\rm phys}^{(M)}$ in $\tilde\Phi_{\rm phys}$ takes the \emph{coherent average} over all possible such relations. This can also be seen by inspecting the action of $\tilde\Phi_{\rm phys}$ on the basis elements of $\mathcal{A}_{\rm phys}$:
\[
   \tilde\Phi_{\rm phys}\left(|\mathbf{h};\mathbf{1}\rangle\langle\mathbf{j};\mathbf{1}|\right)=\frac 1 {|\mathcal{G}|}\sum_{g\in\mathcal{G},\mathbf{g}\in\mathcal{G}^M} |\mathbf{h},\mathbf{g};\mathbf{1}\rangle\langle\mathbf{j},g\mathbf{g};\mathbf{1}|.
\]
In contrast to Eq.~(\ref{eqAssignOne}) and the embeddings of Subsection~\ref{SubsecNonUniqueEmbd}, this embedding does \emph{not} assign to every $M$-particle configuration $\mathbf{g}$ another one, $\mathbf{g}'$, which has ``the same relation'' to the first $N$ particles in branches $\mathbf{h}$ and $\mathbf{j}$. Instead, it generates the \emph{uniform superposition} of all the possibilities.

This averaging is also the reason for the failure of the unitality property. However, as we will see shortly, the absence of unitality is precisely the reason why the relational trace defined below maps relational $(N+M)$-particle states into relational $N$-particle states. It will thus be rather a feature than a failure.

While $\tilde\Phi_{\rm phys}$ is not unital, note that the embedding of the relational $N$-particle unit $\tilde\Phi_{\rm phys}(\Pi_{\rm phys}^{(N)})=\Pi_{\rm phys}^{(N)}\otimes\Pi_{\rm phys}^{(M)}$ certainly \emph{does} act as the identity on its image,
\ba
\ch_{\rm phys}^{(N\otimes M)}:=\Big\{\Pi_{\rm phys}^{(N)}\otimes\Pi_{\rm phys}^{(M)}\ket{\psi}\,\Big|\,\ket{\psi}\in\ch^{\otimes(N+M)}\Big\}\nn.
\ea
This subspace of the space of relational $(N+M)$-particle states $\ch_{\rm phys}^{(N+M)}$ will be essential below when resolving the paradox of the third particle. 

The natural embedding induces a natural trace.
\begin{definition}[Relational trace]\label{def_trel}
The relational trace is defined to be the Hilbert-Schmidt adjoint ${\rm Trel}:=\Phi_{\rm phys}^\dag$ of the extended embedding map $\Phi_{\rm phys}:\mathcal{L}(\ch^{\otimes N})\rightarrow \mathcal{L}(\mathcal{H}^{\otimes(N+M)})$ defined by $\Phi_{\rm phys}:=\tilde\Phi_{\rm phys}\circ\hat\Pi_{\rm phys}^{(N)}$. That is, the relational trace is the unique map with the property
\ba
{\rm tr}\left(\Phi_{\rm phys}(A^{(N)})\,\rho\right) = {\rm tr}\left(A^{(N)}\,{\rm Trel}_{(M)}\rho\right)
\ea
for all $A^{(N)}\in\mathcal{L}(\ch^{\otimes N})$ and all $\rho\in\mathcal{L}(\ch^{\otimes(N+M)})$ (in particular for all states).
\end{definition}

Specifically, note that 
\ba
{\rm tr}\left(\Phi_{\rm phys}(A^{(N)})\,\rho\right) &=& {\rm tr}\left(\Phi_{\rm phys}(A^{(N)})\,\hat\Pi_{\rm phys}^{(N+M)}(\rho)\right)\nn\\
&=&{\rm tr}\left(\Phi_{\rm phys}(A^{(N)})\,\rho_{\rm phys}\right)\label{trel1}.
\ea
This is precisely the expectation value of the relational observable $\Phi_{\rm phys}(A^{(N)})\in\ca_{\rm phys}^{(N+M)}$ in the relational state $\rho_{\rm phys}\in\mathcal{S}(\ch_{\rm phys}^{(N+M)})$, evaluated in the manifestly invariant inner product on $\mathcal{S}(\ch_{\rm phys}^{(N+M)})$. The relational trace is thus unambiguously defined in terms of the so-called physical inner product of constraint quantization~\cite{Giulini:1998kf,Marolf:2000iq,Thiemann}, i.e.\ the inner product on $\ch_{\rm phys}$. We will analyze this in more detail in our upcoming work~\cite{MMP}.

\begin{theorem}\label{thm_trel}
The relational trace takes the explicit form
\ba
{\rm Trel}_{(M)}\rho = {\rm Tr}_{(M)}\left[\Pi_{\rm phys}^{(N)}\otimes\Pi_{\rm phys}^{(M)}\,\rho\,\Pi_{\rm phys}^{(N)}\otimes\Pi_{\rm phys}^{(M)}\right],\nn
\ea
where ${\rm Tr}_{(M)}$ is the standard partial trace over particles $N+1,\ldots,N+M$.
It maps relational operators onto relational operators, i.e.\ ${\rm Trel}_{(M)}\left(\ca_{\rm phys}^{(N+M)}\right)=\ca_{\rm phys}^{(N)}$, and is trace-preserving for states in $\mathcal{S}(\ch_{\rm phys}^{(N\otimes M)})$, but trace-decreasing outside of it. Furthermore, it preserves observational equivalence, i.e.\ $\rho\sim\sigma$ implies ${\rm Trel}_{(M)}\rho = {\rm Trel}_{(M)}\sigma$.
\end{theorem}
\begin{proof}
The first statement follows from
\ba
&&{\rm tr}\left(\Phi_{\rm phys}(A^{(N)})\,\rho\right) = {\rm tr}\left[\left(\Pi_{\rm phys}^{(N)} A^{(N)}\Pi_{\rm phys}^{(N)}\otimes\Pi_{\rm phys}^{(M)}\right)
\rho\right]\nn\\
&&={\rm tr}\left[\left(A^{(N)}\otimes\mathbf{1}^{(M)}\right)\left(\Pi_{\rm phys}^{(N)}\otimes\Pi_{\rm phys}^{(M)}\,\rho\Pi_{\rm phys}^{(N)}\otimes\Pi_{\rm phys}^{(M)}\right)\right],\nn
\ea
which holds for any $A^{(N)}\in\mathcal{L}(\ch^{\otimes N})$ and any $\rho\in\mathcal{L}(\ch^{\otimes (N+M)})$.

Given the conjugation of its input with the projector $\Pi_{\rm phys}^{(N)}\otimes\Pi_{\rm phys}^{(M)}$, it is clear that ${\rm Trel}_{(M)}$ is  trace-preserving for states in $\mathcal{S}(\ch_{\rm phys}^{(N\otimes M)})$, but not for states outside of it. It is also clear that ${\rm Trel}_{(M)}$ maps operators from $\ca_{\rm phys}^{(N+M)}$ into operators in $\ca_{\rm phys}^{(N)}$, since the projectors $\Pi_{\rm phys}^{(N)}$ can be taken outside of the trace over particles $N+1,\ldots,N+M$. To see that it is surjective, it is straightforward to check that $|\cg|^{1-M}{\rm Trel}_{(M)}\left(A_{\rm phys}^{(N)}\otimes\mathbf{1}^{(M)}\right)=A_{\rm phys}^{(N)}$ for all $A_{\rm phys}^{(N)}\in\ca_{\rm phys}^{(N)}$.

Finally, note that the image of $\Phi_{\rm phys}$ is contained in $\mathcal{A}_{\rm phys}^{(N+M)}$, thus $\Phi_{\rm phys}=\hat\Pi_{\rm phys}^{(N+M)}\circ\Phi_{\rm phys}$. Taking the Hilbert-Schmidt adjoint of this equation yields ${\rm Trel}_{(M)}\circ\hat\Pi_{\rm phys}^{(N+M)}={\rm Trel}_{(M)}$. Now suppose we have $\rho\sim\sigma$, then Lemma~\ref{LemEquivProj} and $\hat\Pi_{\rm phys}^{(N+M)}\circ \Pi_{\rm inv}^{(N+M)}=\hat\Pi_{\rm phys}^{(N+M)}$ imply $\hat\Pi_{\rm phys}^{(N+M)}(\rho)=\hat\Pi_{\rm phys}^{(N+M)}(\sigma)$. Altogether this implies that ${\rm Trel}_{(M)}$ preserves observational equivalence.
\end{proof}
We can write
\begin{equation}
   {\rm Trel}_{(M)}=\hat\Pi_{\rm phys}^{(N)}\circ {\rm Tr}_{(M)}\circ \hat\Pi_{\rm phys}^{(N+M)}.
   \label{eqTrelVersusTr}
\end{equation}
We can thus view the relational partial trace ${\rm Trel}_{(M)}$ as an invariant extension of the standard partial trace ${\rm Tr}_{(M)}$.
The non-unitality of $\tilde\Phi_{\rm phys}$ is reflected in the final application of $\hat\Pi_{\rm phys}^{(N)}$. Without this projection, the image of ${\rm Trel}_{(M)}$ would not in general be contained in $\mathcal{A}_{\rm phys}^{(N)}$. For example, the uniform mixture $\rho^{(N+M)}:=\Pi_{\rm phys}^{(N+M)}/|\mathcal{G}|^{N+M-1}$ on the physical subspace of $N+M$ particles yields ${\rm Tr}_{(M)}\rho^{(N+M)}=\mathbf{1}^{(N)}/|\mathcal{G}|^N$, whose decomposition according to Theorem~\ref{TheAlgebraProjections} contains operators outside of $\mathcal{A}_{\rm phys}^{(N)}$.\footnote{More generally, the image of $\ch_{\rm phys}^{(N+M)}\setminus\ch_{\rm phys}^{(N\otimes M)}$ under $\rm{Tr}_{(M)}$ does not lie in the $N$-particle relational subspace $\ch_{\rm phys}^{(N)}$.} \emph{Thus, non-unitality is the price to pay for remaining relational.}

This leads to ${\rm Trel}_{(M)}$ being trace-\emph{decreasing}, unless the initial state is fully supported on the subspace $\mathcal{H}_{\rm phys}^{(N\otimes M)}$. Should we be worried about this fact --- shouldn't marginals of normalized quantum states be normalized? Not in this case. In contrast to the standard partial trace, the relational trace is \emph{not} supposed to tell us what the reduced quantum state on a subsystem is. Instead, it is constructed to tell us precisely the following:

\begin{shaded}
\begin{theorem}
\label{TheConditionalState}
Given some $(N+M)$-particle state $\rho^{(N+M)}\in\mathcal{S}(\mathcal{H}^{\otimes(N+M)})$, the following \emph{conditional state of $N$ particles} is normalized or subnormalized:
\be
   \rho^{(N)}:=\frac{{\rm Trel}_{(M)} \rho^{(N+M)}}{{\rm tr}\left(\rho^{(N+M)}\Pi_{\rm phys}^{(N+M)}\right)}.
\ee
Consider any relational projector $0\leq E^{(N)}_{\rm phys}\leq \Pi^{(N)}_{\rm phys}$ which we interpret as a ``relational event''. Then the state $\rho^{(N)}$ tells us the probabilities of this $N$-particle event, \emph{conditioned on the $(N+M)$-particle system being relational}:
\be
   {\rm tr}\left(E^{(N)}_{\rm phys}\rho^{(N)}\right)={\rm Prob}\left(E_{\rm phys}^{(N)}|\,\, \Pi_{\rm phys}^{(N+M)}\right).
\ee
\end{theorem}
That is, the renormalized result $\rho^{(N)}$ of the relational trace gives us \textbf{the expectation values of all $N$-particle relational observables, conditioned on the global $(N+M)$-particle state being fully relational.}
\end{shaded}
\begin{proof}
It follows from Theorem~\ref{thm_trel} and Eq.~(\ref{projid}) that
\begin{eqnarray}
   {\rm tr}\left({\rm Trel}_{(M)}\rho^{(N+M)}\right)&=&{\rm tr}\left(\Pi_{\rm phys}^{(N)}\otimes\Pi_{\rm phys}^{(M)}\rho^{(N+M)}\right)\nn\\
   &\leq& {\rm tr}\left(\Pi_{\rm phys}^{(N+M)}\rho^{(N+M)}\right),
\end{eqnarray}
hence $\rho^{(N)}$ is not supernormalized. Now, the probability that an initial global measurement of the projector $\Pi_{\rm phys}^{(N+M)}$ yields outcome ``yes'', and then a subsequent local measurement of $E_{\rm phys}^{(N)}$ yields ``yes'' too, is
\ba
&&{\rm Prob}\left(E_{\rm phys}^{(N)},\Pi_{\rm phys}^{(N+M)}\right)\nn\\
&&\q\q\q={\rm tr}\left(\Phi_{\rm phys}(E_{\rm phys}^{(N)})\hat\Pi_{\rm phys}^{(N+M)}(\rho^{(N+M)})\right)\nn\\
&&\q\q\q= {\rm tr}\left(E_{\rm phys}^{(N)}{\rm Trel}_{(M)}\circ \hat\Pi_{\rm phys}^{(N+M)}(\rho^{(N+M)})\right)\nn\\
&&\q\q\q={\rm tr}\left(E_{\rm phys}^{(N)}{\rm Trel}_{(M)}\rho^{(N+M)}\right),
\ea
where we have used Definition~\ref{def_trel} and Eq.~(\ref{eqTrelVersusTr}). The rest of the claim follows from the definition of conditional probability.
\end{proof}

Recall that in Lemma~\ref{LemProjection}, we have seen that the projection $\rho_{\rm phys}:=\hat\Pi_{\rm phys}^{(N+M)}(\rho)$ for alignable states $\rho$ is subnormalized, but is sufficient to determine the expectation values of all invariant observables. Thus, $\rho_{\rm phys}$ should not be seen as the marginal of $\rho$ on some subsystem, but as the ``relational part'' of $\rho$. The ``relational weight'' ${\rm tr}\,\rho_{\rm phys}$ is in general less than one,\footnote{This assumes that $\rho\in\mathcal{L}(\ch^{\otimes(N+M)})$ in $\rho_{\rm phys} = \hat\Pi^{(N+M)}_{\rm phys}(\rho)$ is normalized, as appropriate in the context of our manuscript where there is an external observer who could measure also non-invariant observables in the presence of an external frame. By contrast, in the perspective-neutral approach~\cite{Vanrietvelde,Vanrietvelde2,Hoehn:2018aqt,Hoehn:2018whn,Hoehn:2019owq,Hoehn:2020epv} (and more generally in constraint quantization), one disregards any external structure and works directly with normalized relational states. That is, one would \emph{define} the normalization of the $(N+M)$-particle relational state as $\rm{tr}\rho_{\rm phys}=1$. Indeed, note that the \emph{full} standard trace $\rm{tr}\rho_{\rm phys} = \rm{tr}\left(\Pi_{\rm phys}^{(N+M)}\rho\right)$ is precisely the extension of the so-called physical inner product \cite{Giulini:1998kf,Marolf:2000iq,Thiemann} to density matrices. However, note also from the previous footnote that the standard \emph{partial} trace $\rm{Tr}_{(M)}$ is not in general appropriate for relational states $\rho_{\rm phys}$. The additional projection with $\hat\Pi_{\rm phys}^{(N)}$ in the relational trace fixes this issue, but reduces the norm of states with support outside of $\ch_{\rm phys}^{(N\otimes M)}$. This has a transparent physical interpretation which is best seen through the norm reduction of invariant basis states:
$\rm{Trel}_{(M)}(\ket{\mathbf{h}_N,\mathbf{h}_M;\mathbf{1}}\!\bra{\mathbf{j}_N,\mathbf{j}_M;\mathbf{1}})$, where $\mathbf{h}_{N},\mathbf{j}_N\in\cg^{N-1}$ and $\mathbf{h}_M,\mathbf{j}_M\in\cg^M$, is equal to $\frac{1}{|\cg|}\,\ket{\mathbf{h}_N;\mathbf{1}}\!\bra{\mathbf{j}_N;\mathbf{1}}$ if $\mathbf{h}_M=g\mathbf{j}_M$ for some $g\in\cg$, and is zero otherwise. The variables $\mathbf{h}_M,\mathbf{j}_M$ lying in $\cg^M$ rather than $\cg^{M-1}$ reflects the fact that they encode not only the $M-1$ internal relations of the  $M$ particles, but also a definite relation between the two particle groups, which is a property of \emph{both} groups together. The normalization reduction factor $1/|\cg|$ comes from the coherent averaging over the relations between the two particle groups (which is partially a property of the $N$ particles), and quantifies the corresponding ignorance of the relational $N$-particle state obtained via $\rm{Trel}_{(M)}$. By construction, the latter contains all information about the relational $N$-particle observables which are independent of the relation between the particle groups. However, the `ignorance factor' $1/|\cg|$ has to be taken into account.} and it can decrease when disregarding some of the particles.

\subsection{Relational resolution of the paradox}

In order to apply these insights to the paradox of the third particle, let us first discuss some further aspects of the interplay of composition and relational trace, for the generic case of $M+N$ particles. Suppose $\rho^{(N)}$ is  a state of the $N$ particles before taking the additional $M$ particles into account,  prepared by an observer with access to the external reference frame. The corresponding relational state is $\rho_{\rm phys}^{(N)}=\hat\Pi_{\rm phys}^{(N)}(\rho^{(N)})$. Next, suppose the $M$ additional particles are prepared in a normalized state $\rho^{(M)}$, and the composite state of all particles is of the product form $\rho^{(N+M)}=\rho^{(N)}\otimes\rho^{(M)}$.  We can then construct the relational state corresponding to this composition. Here it is important to note that $\rho_{\rm phys}^{(N+M)}=\hat\Pi_{\rm phys}^{(N+M)}\left(\rho^{(N+M)}\right)=\hat\Pi_{\rm phys}^{(N+M)}\left(\tilde\rho^{(N+M)}\right)$, for any $\tilde\rho^{(N+M)}\simeq\rho^{(N+M)}$. That is, all members of the symmetry equivalence class of $\rho^{(N+M)}$ (incl.\ states featuring entanglement between the two particle groups) yield the same relational $(N+M)$-particle state. However, due to Theorem~\ref{thm_trel}, it is only the projection into $\ch_{\rm phys}^{(N\otimes M)}$ that matters for the relational trace. But Eq.~\eqref{projid} implies that
\ba
&&\Pi_{\rm phys}^{(N)}\otimes\Pi_{\rm phys}^{(M)}\,\rho_{\rm phys}^{(N+M)}\,\Pi_{\rm phys}^{(N)}\otimes\Pi_{\rm phys}^{(M)}\nn\\
&&\q\q\q=\Pi_{\rm phys}^{(N)}\otimes\Pi_{\rm phys}^{(M)}\left(\rho^{(N)}\otimes\rho^{(M)}\right)\Pi_{\rm phys}^{(N)}\otimes\Pi_{\rm phys}^{(M)}\nn\\
&&\q\q\q=\rho_{\rm phys}^{(N)}\otimes\rho_{\rm phys}^{(M)},
\ea
and so we have ${\rm Trel}_{(M)} \rho^{(N+M)}_{\rm phys} = \rho_{\rm phys}^{(N)} \cdot {\rm tr}_{(M)}\rho_{\rm phys}^{(M)}$. Up to a constant factor, this is precisely the initial relational $N$-particle state that we had before taking the additional $M$ particles into account.

Let us see what this implies for the paradox of the third particle. Thus, let us concretely compute the relational trace ${\rm Trel}_3$ of the state $|\Psi\rangle$ of Eq.~(\ref{eqThreePInnocuous}). Due to invariance, the result will be identical if we apply it to the state $|\Psi'\rangle$ of Eq.~(\ref{eqThreeParticleState}). As a first step, we find
\begin{eqnarray}
\Pi_{\rm phys}^{(3)}(|\Psi\rangle\langle\Psi|)\Pi_{\rm phys}^{(3)}&=& \frac 1 {2 n} |\mathbf{h};\mathbf{1}\rangle\langle\mathbf{h};\mathbf{1}|+\frac {e^{-i\theta}} {2 n} |\mathbf{h};\mathbf{1}\rangle\langle\mathbf{j};\mathbf{1}| \nonumber\\
&&+ \frac {e^{i\theta}}{2n}  |\mathbf{j};\mathbf{1}\rangle\langle\mathbf{h};\mathbf{1}| + \frac 1 {2n} |\mathbf{j};\mathbf{1}\rangle\langle\mathbf{j};\mathbf{1}|.\nonumber
   \end{eqnarray}
Using ${}_3\langle g_3|\mathbf{h};\mathbf{1}\rangle=\frac{1}{\sqrt{|\cg|}}\,|g_3 h_2^{-1},g_3 h_2^{-1} h_1\rangle$ yields
\ba
{\rm Trel}_3 (|\Psi\rangle\langle\Psi|)&=&\Pi_{\rm phys}^{(2)}{\rm Tr}_3 \left[\Pi_{\rm phys}^{(3)}(|\Psi\rangle\langle\Psi|)\Pi_{\rm phys}^{(3)}\right]\Pi_{\rm phys}^{(2)}\nn\\
&=&\Pi_{\rm phys}^{(2)}|\psi\rangle\langle\psi|\,\Pi_{\rm phys}^{(2)}\nn\\
&=& \frac 1 {2 n^2} |h;\mathbf{1}\rangle\langle h;\mathbf{1}|+\frac{e^{-i\theta}}{2 n^2}|h;\mathbf{1}\rangle\langle\ j;\mathbf{1}|\nonumber\\
   &&+\frac{e^{i\theta}}{2 n^2}|j;\mathbf{1}\rangle\langle h;\mathbf{1}|+\frac 1 {2 n^2} |j;\mathbf{1}\rangle\langle j;\mathbf{1}|\nonumber.
\ea
Since $|\Psi\rangle$ is alignable, Lemma~\ref{LemProjection} tells us that $\langle\Psi|\Pi_{\rm phys}^{(N+M)}|\Psi\rangle=1/n$. Thus, computing the conditional state of Theorem~\ref{TheConditionalState}, we obtain the projection of the state in Eq.~(\ref{eqProjection1}) into the relational subalgebra. Hence, we recover exactly the relational state of the first two particles which we had before adding the third. In particular, the phase $\theta$ is preserved: as expected, it remains accessible on the first two particles. 

The algebra $\ca_{\rm phys}$ generated by relational observables and the subspace $\ch_{\rm phys}$ of relational states is the arena of the perspective-neutral approach to QRFs \cite{Vanrietvelde,Vanrietvelde2,Hoehn:2018aqt,Hoehn:2018whn,Hoehn:2019owq,Hoehn:2020epv}. As such, there is no paradox of the third particle in this approach. Furthermore, it  provides a compelling conceptual interpretation of this resolution: the relational states are the perspective-neutral states, i.e.\ they correspond to a description of the composite particle system prior to choosing an internal reference relative to which the state is described. The perspective-neutral states contain the entire information about all internal  QRF perspectives at once. The relational trace is performed at the perspective-neutral level, and consistency at that level implies consistency in all internal perspectives. We will further elaborate on this in Ref.~\cite{MMP}.

\subsection{Comparison to the resolution by Angelo et al.}
Angelo et al.~\cite{Angelo} also propose a resolution to the apparent paradox that they have raised in their paper. Let us recapitulate their resolution in our terminology and compare the two approaches. First, they introduce an operator $T:=e^{-2i(a+b)\hat p_{r_2}}$ which in our notation simply implements a translation of the two particles, given by $T|g_1,g_2\rangle=|g_1-2a,g_2+2b\rangle$. Computing the expectation value in the two-particle state of Eq.~(\ref{eqOriginal}) yields $\langle\psi|T|\psi\rangle=\frac 1 2 e^{i\theta}$. Thus, this operator (or rather its real and imaginary parts) admit the measurement of the phase $\theta$.

Since we are in the framework of Example~\ref{ExCyclic}, we can use the explicit form of the characters to see that
\be
   T|h;\chi_k\rangle=\chi_k(-2a)|h+2a+2b;\chi_k\rangle.
\ee
In particular, the invariant part of this translation can be expressed in the form
\begin{equation}
   T|h;\mathbf{1}\rangle=|h+2a+2b;\mathbf{1}\rangle.
   \label{eqTInv}
\end{equation}
That is, $T$ increases the relative distance of the particles by $2a+2b$. Indeed, if we define
\be
   T_{\rm inv}:=\Pi_{\rm inv}(T)=\Pi_{\rm alg}(T)=\hat\Pi_{\rm phys}(T),
\ee
then $T_{\rm inv}\in\mathcal{A}_{\rm phys}$ satisfies Eq.~(\ref{eqTInv}). Note, however, that $T$ contains strictly more information than $T_{\rm inv}$: not only does it tell us that the relative distance of the particles increases by $2a+2b$, but it also tells us what happens to their \emph{absolute} positions.

Angelo et al.\ write: \emph{``The crucial (and surprising) observation is that $T$ actually shifts the relative coordinate of particle $3$ as well as that of particle $2$.''} Strictly speaking, this is not a claim about $T$, but about the embedding $T^{(3)}$ of $T$ into the three-particle observables. Using the obvious embedding $T^{(3)}=T\otimes\mathbf{1}$, we obtain $T^{(3)}|g_1,g_2,g_3\rangle=|g_1-2a,g_2+2b,g_3\rangle$, and thus
\be
   T^{(3)}|h_1,h_2;\mathbf{1}\rangle=|h_1+2a+2b,h_2+2a;\mathbf{1}\rangle.
\ee
That is, the relative coordinate of particle $3$ is also shifted by $2a$. This reproduces Angelo et al.'s claim, but it is important to understand where the shift of $h_2\mapsto h_2+2a$ comes from. It is certainly not possible to deduce this shift from $T_{\rm inv}$ alone. Instead, it comes form \emph{the specific choice} of implementing $T_{\rm inv}$ (a relative shift of $2a+2b$) via $T$ (absolutely shifting particle $1$ by $-2a$ and particle $2$ by $2b$). And the latter choice comes from Angelo et al.'s decision of \emph{preserving the center of mass}, which fixes the non-invariant action of $T$.

In summary: Angelo et al.'s proposed resolution of the paradox comes from \emph{deciding to embed the two-particle observables via the center-of-mass embedding} that we have described in Subsection~\ref{SubsecClassTraces}. As shown there, this leads to a preservation of the phase $\theta$. While the center-of-mass may be the most democratic choice among the particles, as also discussed in Subsection~\ref{SubsecClassTraces}, there exist other well-motivated, but physically inequivalent choices of embedding, e.g.\ relative to one of the particles, which may or may not preserve $\theta$. Basing a resolution of the paradox in the context of frame covariance on a specific choice of frame seems unsatisfactory. Furthermore, even if one wanted to argue that the center-of-mass is a privileged choice in mechanics, this would have to rely on its \emph{dynamical} properties --- in fact, postulating certain particle masses in the first place is nothing but a claim about how the particles move. However, the paradox of the third particle is a completely \emph{kinematical} thought experiment: it can be formulated without any reference to the particle masses or time evolution whatsoever, as explained at the beginning of Section~\ref{Section:Paradox}. Thus, for the paradox, there is in any case no reason to privilege the center of mass embedding over any other choice of invariant embedding, and hence this type of argumentation is insufficient to conclude that the phase $\theta$ is preserved.

In contrast, our resolution amounts to the construction of a \emph{relational embedding} for which no such choice has to be made in the first place. Nonetheless, Angelo et al's insight is still important: embedding fewer into more particles will in general ``do something'' to the additional particles, and  care has to be taken of how the embedding is accomplished.

\section{Conclusions}
\label{SecConclusions}
The aim of this article is to elucidate the operational essence and interpretation of the recent structural approach to QRFs \cite{Giacomini,Vanrietvelde,Hamette}, and to  also clarify the meaning of its QRF transformations as symmetry transformations. These insights have then been exploited to illuminate the physics behind the apparent 'paradox of the third particle' of Ref.~\cite{Angelo} and to resolve it at a formal level through relational observables.

We began by providing a careful conceptual comparison of the quantum information~\cite{Bartlett,Marvian,Frameness,Modes,Smith2019,ResourceTheoryQRF,Palmer,Smith2016} and structural approaches to QRFs, illustrating the difference in their operational essence in terms of two communication scenarios. While both approaches focus on an external relatum independent description of physical observables and quantum states, they do so in different manners and with different goals. As we have seen, technically a distinction can be drawn between the two  in terms of how they describe external relatum independent states fundamentally: they are the \emph{incoherently} and \emph{coherently} group-averaged states in the QI and structural approach, respectively.

A key ambition of the QI approach is to elucidate how to perform communication protocols between different parties in the absence of a shared external laboratory frame. To this end, it suffices to focus on physical properties of the communicated quantum system that are meaningful relative to an arbitrary choice of external frame. For example, this can be achieved by restricting to speakable information that is encoded in decoherence-free subspaces or by communicating an \emph{additional} reference quantum system that serves as a token for the sender's reference frame. Either way,  the QI approach maintains the reference frame external to the system of interest. For successfully carrying out such operational protocols it is also not necessary to take an extra step and choose an internal reference frame within the system of interest, and to ask how the quantum system is described relative to one of its subsystems.

However, this additional step is precisely what the structural approach aims for. Its primary goal is not the implementation of protocols for communicating physical information. 
It has rather a more fundamental ambition: to dissolve the distinction between quantum systems and reference frames and thereby to extend the set of available reference frame choices to include subsystems of the physical system of interest. Its focal point are thus not only external relatum independent state descriptions, but \emph{internal} state descriptions. 
The operational essence of the structural approach can be illustrated in a scenario in which different agents agree on a (redundancy-free)  \emph{description} of physical quantum states \emph{without} adhering to an external relatum. They can always achieve this task by invoking certain ``canonical choices'' in the representation of quantum states that exploit the \emph{internal} structure of the quantum system to be described. In particular, these canonical choices of representation are related by transformations that coincide with the QRF transformations in the structural approach.

To show this explicitly, we have then formalized these conceptual observations in the context of an $N$-particle quantum system (``$\mathcal{G}$-system'') where the configuration space of each particle is a finite Abelian group $\cg$.  We chose this simple setting in order to avoid technicalities and to render all appearing structures completely transparent.  But we emphasize that our observations are of more general validity. They apply directly to laboratory situations in which agents simply disregard, or do not have access to, a relatum external to the system of interest, but in principle also to the case that no external physical relatum exists in the first place as, e.g., in quantum cosmology (see Refs.~\cite{Hoehn:2017gst,Hoehn:2018whn} for a related discussion).

We determined the symmetry group $\cg_{\rm sym}$ associated with a $\cg$-system, which preserves all its external-relatum-independent structure. We showed that states from the corresponding symmetry equivalence class are alignable to a choice of reference system through a $\cg_{\rm sym}$ transformation: each equivalence class contains  ``canonical choices'' of state representations, and these correspond to selecting any one of the $N$ particles as a reference system to define the origin and to describing the remaining $N-1$ particles relative to it. These canonical choices are the 'internal QRF perspectives' on the $N$-particle system of Refs.~\cite{Giacomini,Vanrietvelde,Hamette}.

The symmetry group $\cg_{\rm sym}$ contains the 'classical translation group' $\cg$ as a strict \emph{subgroup}, but it also contains `relation-conditional translations' which turn out to include the QRF transformations of Ref.~\cite{Hamette}, which are equivalent to those of Refs.~\cite{Giacomini,Vanrietvelde}. While it is evident from these works that the QRF transformations are conditional translations, the present article clarifies that they are \emph{symmetry} transformations with a precise and transparent physical interpretation.

Being  translations conditional on the particle relations, the QRF transformations make sense in a classical context when dealing with, for example, statistical mixtures of particle positions rather than superpositions, and indeed have classical analogs that have been exhibited in Refs.~\cite{Vanrietvelde,Vanrietvelde2,Hoehn:2018aqt,Hoehn:2018whn,Hamette}. Nevertheless, just like the CNOT gate has a classical meaning, but can generate entanglement, the QRF transformations similarly lead to interesting quantum effects such as a QRF dependence of, e.g.\ entanglement and superpositions \cite{Giacomini,Vanrietvelde,Hamette}, classicality \cite{Vanrietvelde,Darwinism1,Darwinism2}, spin \cite{Giacomini-spin1,Giacomini-spin2}, certain quantum resources \cite{Savi}, temporal locality \cite{Castro,Hoehn:2019owq}, and of comparing quantum clock readings \cite{Hoehn:2018aqt,Hoehn:2020epv}.\footnote{It would be interesting to study the recent proposals \cite{MischaMax,Alex1,Alex2} for quantum time dilation effects in terms of the temporal QRF transformations as in Refs.~\cite{Hoehn:2018aqt,Hoehn:2018whn,Castro,Hoehn:2019owq,Hoehn:2020epv}.}

Given the two groups $\cg_{\rm sym}$ and  $\cg$ in the setup, one has \emph{a priori} two distinct ways to construct invariant states and observables. Interestingly, as we have shown, the invariant (pure) states and thus the subspace $\ch_{\rm phys}$ of relational states do \emph{not} in fact depend on whether one requires invariance under the action of $\cg_{\rm sym}$ or its subgroup $\cg$. By contrast, the set of invariant observables \emph{does} depend on which group one works with: the operator algebra invariant under $\cg_{\rm sym}$ is a strict subset of the operator algebra invariant under its subgroup $\cg$. However, the two invariant operator algebras coincide again in their restriction  to the space of relational states $\ch_{\rm phys}$, which is the algebra $\ca_{\rm phys}$ generated by so-called relational observables \cite{Rovellibook,Thiemann,Tambornino,Rovelli1,Rovelli2,Rovelli3,Dittrich1,Dittrich2,Chataignier,Hoehn:2018aqt,Hoehn:2018whn,Hoehn:2019owq,Hoehn:2020epv}. The space of relational states $\ch_{\rm phys}$ and the relational operator algebra $\ca_{\rm phys}$ are key structures in constraint quantization \cite{Giulini:1998kf,Marolf:2000iq,Thiemann} and the platform of the perspective-neutral approach to QRFs \cite{Vanrietvelde,Vanrietvelde2,Hoehn:2018aqt,Hoehn:2018whn,Hoehn:2019owq,Hoehn:2020epv} (part of the structural approach), which are thus independent of the distinction between $\cg_{\rm sym}$ and $\cg$. The difference between $\cg_{\rm sym}$ and its subgroup $\cg$ is, however, crucial when aligning the \emph{non-invariant} description of quantum states to a particle at the level of the full $N$-particle Hilbert space.

These observations also permitted us to first clarify the physics behind the 'paradox of the third particle' discussed in Ref.~\cite{Angelo}, and subsequently to resolve it at a formal level. First, we have illuminated why the usual partial trace is \emph{not} suitable in the context of QRFs, because it ignores the equivalence classes of states that are operationally indistinguishable in the absence of an external relatum. Next, we have explained that, in order to take the observational equivalence classes into account, one has to construct a partial trace in terms of the invariant observables. However, even when attempting to do so, we have seen that there does not exist a physically distinguished choice for such an invariant partial trace \emph{outside the space of relational states}. The reason is that an invariant partial trace demands a suitable embedding of the two-particle invariant observables into the three-particle invariant observables. Yet such an embedding (while invariant under symmetry transformations) depends on how one defines the relation of the third to the first two particles, and there are multiple physically inequivalent ways (e.g., distance to the first particle, center of mass, etc.). The two-particle reduced state then depends on one's convention of how to define the relation between the third and the first two particles, despite restricting attention to invariant observables.

However, when restricting attention further to the algebra $\ca_{\rm phys}$ generated by relational observables and the space of relational states $\ch_{\rm phys}$, we showed that there \emph{does} exist a physically distinguished embedding of the relational two-particle observables and states into the relational three-particle observables and states. Physically, this embedding corresponds to \emph{coherently averaging} over all possible relations between the third and the first two particles, and thereby defines an entirely invariant embedding.
This permitted us to define an unambiguous  relational partial trace that determines the expectation values of relational observables on subsets of particles. In particular, this trace achieves for relational states what a consistent partial trace should do: if  a third particle is independently prepared, then the two-particle reduced state, obtained from the relational three-particle state, coincides with the relational two-particle state prior to taking the third particle into account. At the level of relational observables and relational states, the paradox of the third particle of Ref.~\cite{Angelo} is thus resolved; in this sense, the perspective-neutral approach does not feature any paradox of additional particles. However, we have not discussed what it would mean for an agent to operationally implement this resolution in the lab and, specifically, how they may operationally restrict to relational states and observables (although we believe this to be possible). In this light, our resolution of the paradox is formal.

The paradox of the third particle and our resolution can be viewed as a finite-dimensional analog of the problem of boundaries and edge modes in gauge theory and gravity~\cite{Donnelly:2014fua,Donnelly:2016auv,Geiller:2019bti,Freidel:2020xyx,Gomes:2018shn,Riello:2020zbk,Wieland:2017zkf,Wieland:2017cmf}. Boundaries in space or spacetime usually break gauge-invariance and constitute challenges for gauge-invariant observables. The latter are typically non-local (such as Wilson loops) and can thus have support in two neighbouring regions separated by a boundary. Those gauge-invariant observables with support in both regions determine the physical relation between the two and are accounted for in terms of so-called edge modes when one of the regions is ignored. This is analogous to the joining of two groups of $N$ and $M$ particles and asking for the invariant relations between the two groups. The relative distances between the two groups of particles are the finite-dimensional analog of the gauge-invariant observables in gauge theories and gravity that have support in two neighbouring regions. As we have seen, ignoring one group of particles by simply taking the standard partial trace may indeed lead to an invariance breaking in analogy to the field theory case, i.e.\ $N$-particle states that are \emph{not} relational. This is because the set of relational observables for the joint $(N+M)$-particle system is \emph{not} only the union of the sets of relational $N$- and $M$-particle observables, again in analogy to two neighbouring subregions in spacetime. Our relational trace defines a purely relational, i.e.\ invariant way of `ignoring' a group of particles, and it would be interesting to extend this tool to the study of edge modes in gauge theories and gravity.

Lastly, we emphasize that our novel interpretation of the structural approach applies in particular also to temporal quantum reference frames, i.e.\ quantum clocks. For instance, the example of the cyclic group could model a set of quantum clocks each with a finite set of readings. The external frame would then be some laboratory clock that one external observer has access to, but another  may not. Nevertheless, the two observers can agree on the description of the flow of time by focusing on a purely \emph{internal} choice of clock that leads to a \emph{relational notion of time} entirely independent of any external  clock. It is in this sense that one can interpret the relational quantum dynamics defined by temporal relational observables \cite{Rovellibook,Thiemann,Tambornino,Rovelli1,Rovelli2,Rovelli3,Dittrich1,Dittrich2,Hoehn:2018aqt,Hoehn:2018whn,Hoehn:2019owq,Hoehn:2020epv,Gambini} or the Page-Wootters formalism \cite{Hoehn:2019owq,Hoehn:2020epv,Alex1,Castro,Page,Wootters,Giovanetti,Alex3} (which recently have been shown to be equivalent \cite{Hoehn:2019owq,Hoehn:2020epv}) in the context of laboratory situations.\footnote{These frameworks, however, also apply in the absence of external agents such as in quantum cosmology or gravity.} Indeed, this is precisely how the experimental illustration of the Page-Wootters dynamics reported in \cite{Moreva} is to be understood.

In this manuscript, we focused purely on kinematical aspects of quantum reference frame physics. In our companion article~\cite{MMP}, we study in detail how the insights gained here are affected when we take the dynamics of the $N$-particle system into account. This question will link also with the perspective-neutral approach to QRFs \cite{Vanrietvelde,Vanrietvelde2,Hoehn:2018aqt,Hoehn:2018whn,Hoehn:2019owq,Hoehn:2020epv}, and we will establish in detail the equivalence of its ``quantum coordinate changes''  with the QRF transformations exhibited here.

\section*{Acknowledgments}
We thank Thomas Galley, Leon Loveridge and Alexander R.\ H.\ Smith for many useful comments on an earlier draft version. MK is grateful to \v{C}aslav Brukner, Esteban Castro-Ruiz, and David Trillo Fernandez for helpful discussions. MPM would like to thank Anne-Catherine de la Hamette for inspiring discussions in the initial phase of this project.  MK acknowledges the support of the Vienna Doctoral School in Physics (VDSP) and the support of the Austrian Science Fund (FWF) through the Doctoral Programme CoQuS. PAH is grateful for support from the Foundational Questions Institute under grant number FQXi-RFP-1801A. This work was supported in part by funding from Okinawa Institute of Science and Technology Graduate University. MK and MPM thank the Foundational Questions Institute and Fetzer Franklin Fund, a donor advised fund of Silicon Valley Community Foundation, for support via grant number FQXi-RFP-1815. This research was supported in part by Perimeter Institute for Theoretical Physics. Research at Perimeter Institute is supported by the Government of Canada through the Department of Innovation, Science and Economic Development Canada and by the Province of Ontario through the Ministry of Research, Innovation and Science.

\end{document}